\documentclass[aos,preprint]{imsart}

\RequirePackage{amsthm,amsmath,amsfonts,amssymb}
\RequirePackage[numbers]{natbib} 
\RequirePackage[colorlinks,citecolor=blue,urlcolor=blue]{hyperref}
\RequirePackage{graphicx}

\RequirePackage{threeparttable }

\RequirePackage{longtable}
\RequirePackage{booktabs}

\usepackage{chngcntr}
\usepackage{apptools}

\startlocaldefs
 
\usepackage{mathtools}
\usepackage{cleveref}
\usepackage{hyperref} 
\DeclarePairedDelimiter\ceil{\lceil}{\rceil}
\DeclarePairedDelimiter\floor{\lfloor}{\rfloor}
\usepackage{times,enumitem}

\newtheorem{theorem}{Theorem}[section]
\newtheorem{lemma}{Lemma} 
\newtheorem{assumption}{Assumption}
\newtheorem{remark}{Remark}

\newtheorem{corollary}[theorem]{Corollary}
\usepackage{algorithm}
\usepackage{algpseudocode}
\usepackage[dvipsnames]{xcolor}
\usepackage{tikz}
\usetikzlibrary{arrows.meta}
\usetikzlibrary{decorations.pathreplacing}
\allowdisplaybreaks  
 
\endlocaldefs

\begin{document}

\begin{frontmatter}
\title{High Dimensional Latent Panel Quantile Regression \\with an Application to Asset Pricing  }
\runtitle{}

\begin{aug}
\author[A]{\fnms{Alexandre } \snm{Belloni}\ead[label=e1]{abn5@duke.edu}},
\author[B]{\fnms{Mingli } \snm{Chen}\ead[label=e2]{m.chen.3@warwick.ac.uk}}
\author[C]{\fnms{Oscar Hernan} \snm{Madrid Padilla}\ead[label=e3]{oscar.madrid@stat.ucla.edu		}}\\
\and
\author[D]{\fnms{Zixuan (Kevin) } \snm{Wang}\ead[label=e4]{zixuanwang@alumni.harvard.edu }}
\address[A]{ Fuqua Business School, Duke University
}
\address[B]{Department of Economics, University of Warwick
}
\address[C]{Department of Statistics, University of California, Los Angeles
	}
\address[D]{Harvard University 
}
\end{aug}

\begin{abstract}
	We propose a generalization of the linear panel quantile regression model to accommodate both \textit{sparse} and \textit{dense} parts: sparse means that while the number of covariates available is large, potentially only a much smaller number of them have a nonzero impact on each conditional quantile of the response variable; while the dense part is represent by a low-rank matrix that can be approximated by latent factors and their loadings. Such a structure poses problems for traditional sparse estimators, such as the $\ell_1$-penalised Quantile Regression, and for traditional latent factor estimators such as PCA. We propose a new estimation procedure, based on the ADMM algorithm,  that consists of combining the quantile loss function with $\ell_1$ \textit{and} nuclear norm regularization.
We show, under general conditions, that our estimator can consistently estimate both the nonzero coefficients of the covariates and the latent low-rank matrix.  This is done in a challenging setting that allows for temporal dependence, heavy-tail distributions, and the presence of latent factors.

Our proposed model has a "Characteristics + Latent Factors"  Quantile Asset Pricing Model interpretation: we apply our model and estimator with a large-dimensional panel of financial data and find that (i) characteristics have sparser predictive power once latent factors were controlled (ii) the factors and coefficients at upper and lower quantiles are different from the median.
\end{abstract}



\end{frontmatter}

\section{Introduction}

A central question in asset pricing is to explain why certain assets pay higher returns than others. The Arbitrage Pricing Theory (\cite{ross1976arbitrage}) and the Fama-French three factors model (\cite{fama1993common}) explains the asset return variations by a linear combination of common risk factors. Assets with similar exposure to a common factor shall rise and fall together (\cite{cochrane2009asset}). However, empirical evidence appears to indicate that the firm characteristics, rather than common factors, can also explain the variations in stock returns (\cite{daniel1997evidence}), which suggests a characteristic-based model.

We generalize both modeling approaches and propose a ``Characteristics + Latent Factors'' quantile asset pricing framework. By incorporating the ``Characteristics", we improve the economic interpretability and explanatory power of the model. On the other hand, the finance literature has documented a zoo of new characteristics, and the proliferation of characteristics in this ``variable zoo'' leads to a concern about which characteristics really provide independent information about returns  (\cite{cochrane2011presidential}). Our model addresses this issue by imposing a \textit{sparse} structure, meaning although a large set of characteristics is available, only a much smaller subset of them might have predictive power.  We also incorporate ``Latent Factors'' to capture the common variations in asset returns. One additional benefit of having this part is that it might help alleviate the ``omitted variable bias'' problem (\cite{giglio2018asset}). As in the literature, typically, these latent factors are estimated via principal component analysis, which means all possible latent explanatory variables might be important for prediction although their individual contribution might be small, we term this as the \textit{dense} part. \footnote{More about sparse modeling and dense modeling can be found in \cite{giannone2017economic}. See also \cite{chernozhukov2017lava}.} Hence, our framework allows for ``Sparse + Dense" modeling with large scale panel data that consist of a large number of asset returns that are allowed to be weakly correlated across time.  In addition, we focus on understanding the quantiles (hence the entire distribution) of returns rather than just the mean, in line with the recent interest in quantile factor models (e.g. \cite{ando2020quantile}, \cite{chen2018quantile}, \cite{ma2019estimation}, \cite{feng2019nuclear}, and \cite{sagner2019three}). Our quantile asset pricing framework also inherits micro-foundation from the seminal quantile preference framework(\citep{manski1988ordinal,rostek2010quantile,giovannetti2013asset}) and, in particular, the dynamic quantile preference framework of \cite{de2019dynamic}. 

Specifically, with $Y_{i,t}$ as the excess return of asset $i$ in period $t$, $X_{i,t}$ as a $p$-dimensional vector of observable characteristics  such as return volatility and trading volume, we study the following high dimensional latent panel quantile regression model:
\begin{equation}
\label{eqn:first_model_0}
F_{ Y_{i,t}| X_{i,t}; \theta(\tau), \lambda_{i}(\tau), g_{t}(\tau) }^{-1}(\tau) \,=\,  X_{i,t}^{\prime} \theta(\tau)   + \lambda_i(\tau)^{\prime} g_t(\tau),   \,\,  i \,=\, 1\ldots, n,  \,\,\,   t \,=\,  1,\ldots,T,
\end{equation}
 where $\theta(\tau) \in \mathbb{R}^p$ is the vector of coefficients, $g_t(\tau)$ is an $r_{\tau}$-dimensional vector of unobservable factor returns, $\lambda_i(\tau)$ represents the factor loadings which captures the sensitivity of asset $i$ on the $r_{\tau}$ factors, $\tau \in [0,1]$ is the quantile index.  We allow for the possibility of quantile dependence of sensitivity to risk factors as such evidence has been reported in the literature (e.g., \citep{ando2020quantile}). For notation simplicity,  we often denote $\Pi_{i,t} (\tau) = \lambda_i(\tau)^{\prime} g_t(\tau)$, then $\Pi(\tau)$ is a low-rank matrix with unknown rank $r_\tau$.  Thus, with $ F_{ Y_{i,t}| X_{i,t}; \theta(\tau), \lambda_i(\tau), g_{t}(\tau)}$ is the cumulative distribution function of $Y_{i,t}$ conditioning on $X_{i,t}$, $\theta(\tau)$ and $\lambda_i(\tau), g_t(\tau)$, we model the quantiles of returns (instead of expected returns) as a linear combination of the characteristics and latent factors. Our framework allows for the possibility of lagged dependent data. Here, we allow for the number of characteristics $p$,  and the time horizon $T$, to grow to infinity as $n$  grows.  Throughout, we focus on the case where $p$ is large, possibly much larger than $n T$, but for the true model $\theta(\tau)$ is sparse and has only $ s_{\tau} \ll p$ non-zero components.

Our framework is flexible enough that allows us to jointly answer the following three questions in asset pricing: (i) Which characteristics are important to explain the time series and cross-section of stock returns, after controlling for the factors? (ii) How much would the latent factors explain stock returns after controlling for firm characteristics? (iii) Does the relationship of stock returns and firm characteristics change across quantiles?
The first question is related to the recent literature on variable selection in asset pricing using machine learning (\cite{ kozak2019shrinking, feng2019taming, han2018firm}). The second question is related to an classical literature starting from 1980s on statistical factor models of stock returns (\cite{chamberlain1983arbitrage, connor1988risk} and recently \cite{lettau2018estimating}). The third question extends the literature in late 1990s on stock return and firm characteristics (\cite{daniel1997evidence, daniel1998characteristics}) and further asks whether the relationship is heterogenous across quantiles. 

There are several key features of considering prediction problem at the panel quantile model in this setting. First, stock returns are known to be asymmetric and exhibit heavy tail, thus modeling different quantiles of return provides extra information in addition to models of first and second moments. Second, quantile regression provides a richer characterization of the data, allowing heterogeneous relationship between stock returns and  firm characteristics across the entire return distribution. Third, the latent factors might also be different at different quantiles of stock returns. Finally, quantile regression is more robust to the presence of outliers relative to other widely used mean-based approaches. Using a robust method is crucial when estimating low-rank structures (see e.g. \cite{she2017robust}). As our framework is based on modeling the quantiles of the response variable, we do not put assumptions directly on the moments of the dependent variable.  

Our main goal is to consistently estimate both the sparse part and the low-rank matrix. Recovery of a low-rank matrix, when there are additional high dimensional covariates, in a nonlinear model can be very challenging. The rank constraint will result in the optimization problem NP-hard. In addition, estimation in high dimensional regression is known to be a challenging task, which in our frameworks becomes even  more difficult  due to the additional latent structure. We address the former challenge via nuclear norm regularization which is similar to \cite{candes2009exact} in the matrix completion setting. Without covariates, the estimation can be done via solving a convex problem, and similarly there are strong statistical guarantees of recovery of the underlying low-rank structure. We address the latter challenge by imposing $\ell_1$ regularization on the vector of coefficients  of the control variables, similarly to \cite{belloni2011l1} which mainly focused on the cross-sectional data setting. Note that with regards to sparsity, we must be cautious, specially when considering predictive models (\cite{she2017robust}). Furthermore, we explore the performance of our procedure under settings  where  the vector of coefficients can be dense (due to the low-rank matrix). 

We view our work as complementary to the low dimensional quantile regression with interactive fixed effects framework as of the very recent work of \cite{feng2019nuclear},  and the mean  estimation setting in \cite{moon2018nuclear}. However, unlike \cite{moon2018nuclear} and \cite{feng2019nuclear}, we allow  the number  of covariates  $p$ to be large, perhaps $p \gg nT$.  This comes  with different  significant  challenges.  On the computational side, it  requires  us to develop  novel  estimation algorithms, which turns out can also be used  for the contexts in \cite{moon2018nuclear} and \cite{feng2019nuclear}. On the theoretical side,  allowing  $p \gg nT$  requires  a  systematically different analysis  as compared to \cite{feng2019nuclear}, as it is known   that   ordinary   quantile  regression is inconsistent  in   high dimensional  settings  ($p \gg nT$), see \cite{belloni2011l1}.

\paragraph*{Related Literature}
Our work contributes to the recent growing literature on panel quantile model. \cite{abrevaya2008effects}, \cite{graham2018quantile}, \cite{arellano2017quantile}, considered the fixed $T$ asymptotic case. \cite{kato2012asymptotics} formally derived the asymptotic properties of the fixed effect quantile regression estimator under large $T$ asymptotics, and \cite{galvao2016smoothed} further proposed fixed effects smoothed quantile regression estimator.  \cite{galvao2011quantile} works on dynamic panel.
\cite{koenker2004quantile} proposed a penalized estimation method where the individual effects are treated as pure location shift parameters common to all quantiles, for other related literature see \cite{ lamarche2010robust}, \cite{galvao2010penalized}. We refer to Chapter 19 of \cite{koenker2017handbook} for a review. Furthermore, our framework can be viewed as a generalization of the model in \cite{ando2020quantile} which considered panel quantile model with independent errors and low dimensional covariates.

Our work also contributes to the literature on nuclear norm penalisation, which has been widely studied in the machine learning and statistical learning literature, \cite{fazel2002matrix},  \cite{recht2010guaranteed, koltchinskii2011nuclear, rohde2011estimation}, \cite{negahban2011estimation}, \cite{brahma2017reinforced}. Recently, in the econometrics literature \cite{athey2018matrix} proposes a framework of matrix completion for estimating causal effects, \cite{bai2017principal} for estimating approximate factor model, \cite{chernozhukov2018inference} considered the heterogeneous coefficients version of the linear panel data interactive fixed model where the main coefficients has a latent low-rank structure, \cite{bai2019robust} for robust principal component analysis, and \cite{bai2019matrix} for imputing counterfactual outcome.

Finally, our results contribute to a growing literature on high dimensional quantile regression. \cite{wang2012quantile} considered  quantile  regression with concave  penalties for  ultra-high dimensional data;  \cite{zheng2015globally}
proposed an adaptively weighted $\ell_1$-penalty for globally concerned quantile regression. Screening procedures based on moment conditions motivated by the quantile models have been proposed and analyzed in \cite{he2013quantile} and \cite{ wu2015conditional} in the high-dimensional regression setting. We refer to \cite{koenker2017handbook} for a review.

To sum-up, our paper makes the following contributions. First,  we propose a new class of models that consist of both \textit{high dimensional regressors} and \textit{latent factor} structures. We provide a scalable estimation procedure, and show that the resulting estimator is consistent under suitable regularity conditions.
	Second, the high dimensional and non-smooth objective function require innovative strategies to derive all the above-mentioned results. In particular, our paper allows for serial dependence and this flexibility is important for the panel data case. It is well known that dealing with data dependence is a non-trivial problem,  and there are additional challenges for high dimensional models even for those without incorporating the latent factors: the extension of lasso-based methods with
relaxing the i.i.d. assumption and other restrictive assumptions (for instance, Gaussianity),  is just beginning to occur, e.g. for time series data some recent development along this line can be found in \cite{wong2020lasso}. Those lead to the novel use in our proofs of some techniques from the high dimensional statistics and econometrics literature, such as the localization argument from \cite{belloni2011l1},  and the new loss function introduced in  \cite{padilla2020adaptive}; from spectral theory, namely,  properties of  nuclear norm studied in \cite{elsener2018robust}, and concentrations results by \cite{chatterjee2015matrix}; and from empirical process theory \cite{yu1994rates,wellner2013weak}. We also
	 generalize  the sampling and smoothness assumption of  \cite{belloni2011l1}  by considering panel data with weak correlation across time. In particular, we refer readers to \cite{yu1994rates}  for thorough discussions on $\beta$-mixing.
	 
On the theoretical side,  we also present multiple results that entirely differ from those in \cite{belloni2011l1}.  In addition to allowing time dependence and a latent factor structure, we can consistently  estimate the  conditional quantiles without requiring a minimum eigenvalue condition on the behavior of the design matrix, which is typically required in the literature (e.g. \cite{belloni2011l1}). Relative to approaches that incorporate latent factor structure but relies on the squared loss, the proposed estimators inherit from quantile regression certain robustness properties to the presence of outliers and heavy-tailed distributions. Finally, we apply our proposed model and estimator to a large-dimensional panel of financial data in the US stock market and find that different return quantiles have different selected firm characteristics and that the number of latent factors can be also be different.

\paragraph*{Outline} The rest of the paper is organized as follows. Section \ref{sec:Overview} introduces the high dimensional latent quantile regression model, and provides an overview of the main theoretical results.  Section \ref{sec:sol_algorithms} presents the estimator and our proposed ADMM algorithm. Section \ref{sec:theory} discusses the statistical properties of the proposed estimator. Section \ref{sec:Simulation} provides simulation results. Section \ref{sec:Empirical} consists of the empirical results of our model applied to a real data set. The proofs of the main results are in the Supplementary Material.

\paragraph*{Notation} For $m \in \mathbb{N}$, we write $[m] \,=\, \{1,\ldots,m\}$.
For a  vector  $ v \in \mathbb{R}^p$ we  define  its $\ell_0$ norm as  $\| v\|_0 = \sum_{j =1}^{p} 1\{ v_j  \neq 0 \}  $,  where  $1\{ \cdot \}$  takes  value $1$ if the statement inside $\{\}$  is true, and  zero  otherwise; its $\ell_1$ norm as $\Vert v \Vert_1 = \sum^p_{j=1}\vert v_j\vert$.  We denote $\Vert v \Vert_{1,n,T} = \sum_j^{p}\hat{\sigma}_j\vert v_j\vert$ the $\ell_1$-norm weighted by $\hat{\sigma}_j$'s (defined in eq(\ref{eq:sigma})). The Euclidean  norm is denoted by $\|\cdot \|$, thus $\|v\|  = \sqrt{  \sum_{j=1}^p  v_j^2 }$. If $A \in \mathbb{R}^{n \times T}$ is a matrix, its Frobenius norm is denoted by $\|A\|_F \,=\,\sqrt{ \sum_{ i= 1}^n \sum_{t=1}^T A_{i,t}^2  }$, its spectral norm   by $\| A\|_2 \,=\,  \sup_{x \,:\, \|x\| = 1  }\sqrt{ x^{\prime} A^{\prime} A x}   $, its infinity norm by  $\|A\|_{\infty}   \,= \,   \max\{ \vert A_{i,j }\vert \,:\, i \in [n],\,\, j \in [T]    \} $  , its  rank by $\text{rank}(A)$, and its nuclear  norm  by $\|A\|_* = \text{trace}(\sqrt{A^{\prime}  A})$ where  $A^{\prime}$ is  the transpose of $A$. The $j$th column $A$ is denoted  by $A_{\cdot,j}$. Furthermore,  the multiplication of  a tensor $X \in \mathbb{R}^{I_1 \times \ldots \times I_m   }$  with a vector  $\theta  \in \mathbb{R}^{I_m}$ is denoted by $Z := X  \theta    \in \mathbb{R}^{  I_1  \times \ldots  \times I_{m-1}   }$, and,  explicitly, $Z_{i_1,\ldots i_{m-1}    }   =   \sum_{j = 1}^{I_m}    X_{i_1,\ldots,  i_{m-1},j } \,  \theta_{j} $. We also use the notation $a \vee b= \max\{a,b\}$, $a \land b = \min\{a, b\}$, $(a)_{-} = \max\{-a,0\}$. For a sequence  of  random variables  $\{z_j\}_{j=1}^{\infty}$ we denote  by  $\sigma(z_1,z_2,\ldots)$  the sigma algebra  generated  by    $\{z_j\}_{j=1}^{\infty}$.  Finally,   for sequences  $\{a_n\}_{n=1}^{\infty}$  and  $\{b_n\}_{n=1}^{\infty}$  we write  $a_n \asymp b_n$
if there  exists  positive constants  $c_1$  and  $c_2$  such that  $c_1 b_n \leq a_n  \leq c_2 b_n$ for sufficiently  large $n$.

\section{The Estimator and Overview of Rate Results}\label{sec:Overview}

\subsection{Basic Setting}

The setting of interest corresponds to a high dimension latent panel quantile regression model, where $Y \in \mathbb{R}^{n\times T} $, and  $X \in \mathbb{R}^{n \times T \times p}$  satisfying
\begin{equation}
\label{eqn:first_model}
F_{ Y_{i,t}| X_{i,t}; \theta(\tau), \Pi_{i,t}(\tau) }^{-1}(\tau) \,=\,  X_{i,t}^{\prime} \theta(\tau)   + \Pi_{i,t}(\tau),   \,\,\,\,\,  i \,=\, 1\ldots, n,  \,\,\,   t \,=\,  1,\ldots,T,
\end{equation}
where $i$ denotes subjects,  $t$ denotes time,  $\theta(\tau) \in \mathbb{R}^p$ is the vector of coefficients, $\Pi(\tau) \in  \mathbb{R}^{n \times T}$ is a low-rank matrix with unknown rank $r_\tau \ll \min \{n, T\}$, $\tau \in [0,1]$ is the quantile index, and $ F_{ Y_{i,t}| X_{i,t}; \theta(\tau), \Pi_{i,t}(\tau) }$ is the cumulative distribution function of $Y_{i,t}$  conditioning on $X_{i,t}$, $\theta(\tau)$ and $\Pi_{i,t}(\tau)$. Thus, we model the quantile function at level $\tau$ as a linear combination  of the predictors  plus a low-rank matrix. Here, we allow for the number of covariates $p$,  and the time horizon $T$, to grow to infinity as $n$  grows.  Throughout the paper the quantile index $\tau \in (0,1)$ is fixed. We mainly focus on the case where $p$ is large, possibly much larger than $nT $, but for the true model $\theta(\tau)$ is sparse and has only $ s_{\tau} \ll p$ non-zero components. Mathematically, $s_\tau := \Vert \theta(\tau)  \Vert_0$.

When $\Pi_{i,t}(\tau) \,=\,   \lambda_i(\tau)^{\prime} g_t(\tau)$, with $\lambda_i(\tau), \, g_t(\tau)  \in \mathbb{R}^{r_{\tau}}$, this immediately leads to the following setting
\begin{equation}
\label{eqn:low_rank}
F_{ Y_{i,t}| X_{i,t}; \theta(\tau), \Pi_{i,t}(\tau) }^{-1}(\tau)  \,=\, X_{i,t}^{\prime}\theta(\tau) + \lambda_i(\tau)^{\prime} g_t(\tau).
\end{equation}
where we model the quantile function at level $\tau$ as a linear combination of the covariates (as predictors) plus a latent factor structure. This is directly related to the panel data models with interactive fixed effects literature in econometrics, e.g. linear panel data model (\cite{bai2009panel}), nonlinear panel data models (\cite{chen2014estimation, chen2014nonlinear}).

Note, for eq (\ref{eqn:low_rank}), additional identification restrictions are needed for estimating $\lambda_i(\tau)$ and $g_t(\tau)$ (see \cite{bai2013principal}).  In addition, in nonlinear panel data models, this creates additional difficult in estimation, as the latent factors and their loadings part induce a  nonconvex quantile regression problem. However, we deal with this in the following subsection via using a nuclear norm constraint. \footnote{ Different identification conditions might result in different estimation procedures for $\lambda$ and $f$, see \cite{bai2012statistical} and \cite{chen2014estimation}.}

\subsection{Estimator}\label{sec:problem_formulation}
In this subsection, we describe the high dimensional latent quantile estimator.  With  the sparsity and low-rank constraints  in mind,  a natural formulation  for  the estimation of $(\theta(\tau),\Pi(\tau))$ is
\begin{equation}
\label{eqn:formulation}
\begin{array}{ll}
\underset{  \tilde{\theta}\in \mathbb{R}^{p } ,\,\  \tilde{\Pi}\in \mathbb{R}^{n \times T}   }{\text{minimize} } &  \displaystyle   \frac{1}{nT } \sum_{t=1}^{T}    \sum_{i=1}^n   \rho_{\tau}(Y_{i,t}   - X_{i,t}^{\prime}\tilde{\theta}  - \tilde{\Pi}_{i,t} )     \\
\text{subject to} & \mathrm{rank}(\tilde{\Pi})\leq r_{\tau},\\
& \Vert\tilde{\theta}\Vert_{0}=  s_{\tau},
\end{array}
\end{equation}
where  $\rho_{\tau}(t) \,=\, (\tau -  1\{ t\leq 0 \})t $  is  the quantile loss function as in \cite{Koenker2005},  $s_{\tau}$ is a parameter that directly controls the sparsity of $\tilde{\theta}$, and $r_{\tau}$  controls the rank of the estimated latent matrix.

While the  formulation in (\ref{eqn:formulation}) seems appealing, as it enforces variable selection and low-rank matrix estimation simultaneously,  (\ref{eqn:formulation})   is   a non-convex problem due to the constraints posed  by the $\|\cdot\|_0$  and $\mathrm{rank}(\cdot)$ functions.  We propose a convex relaxation of (\ref{eqn:formulation}).  Inspired  by the seminal works of  \cite{tibshirani1996regression} and \cite{candes2009exact},  we formulate the problem as the following

\begin{equation}
\label{eqn:formulation2}
\underset{ \tilde{\theta}\in \mathbb{R}^{p } ,\,\  \tilde{\Pi}\in \mathbb{R}^{n \times T}  }{\min} \left\{\frac{1}{nT } \sum_{t=1}^{T}    \sum_{i=1}^n   \rho_{\tau}(Y_{i,t}   - X_{i,t}^{\prime}\tilde{\theta}  -  \tilde{\Pi}_{i,t} ) +  \nu_{1}\sum_{j=1}^{p}  w_j  \vert   \tilde{\theta}_j  \vert      + \nu_2  \|\tilde{\Pi}\|_*\right\}  
\end{equation}
where  $ \nu_{1} >0$  and  $\nu_2  >0$ are tuning parameters, and  $w_1,\ldots,w_p$ are user  specified weights (more on this in  Section \ref{sec:theory} ).  Notice that $\|\cdot\|_*$ is the nuclear norm defined on Page 4.  The nuclear norm regularization works on the singular value of a matrix,  the intuition is that via penalization with the nuclear norm, the resulting problem will be convex. Just as $\ell_1$-minimization is the tightest convex relaxation of the combinatorial $\ell_0$-minimization problem, nuclear-norm minimization is the tightest convex relaxation of the NP-hard rank minimization problem, see \cite{candes2010matrix}.

In principle, one can use any convex solver software to solve  (\ref{eqn:formulation2}), since  this is a convex optimization problem. However, for large scale problems a more careful implementation might be needed. Section \ref{sec:sol_algorithms} presents a scheme for solving  (\ref{eqn:formulation2}) that  is based on the ADMM algorithm (\citep{boyd2011distributed}).  

\subsection{Summary of results}
We now summarize our main results. For the model defined  in (\ref{eqn:first_model}):
\begin{itemize}
	\item Under  (\ref{eqn:first_model}), $s_{\tau} \ll \min\{n,T\}$, an assumption that implicitly requires $r_{\tau} \ll \min\{n,T\}$, and other regularity conditions  defined in Section \ref{sec:theory}, we show that our estimator $(\hat{\theta}(\tau), \hat{\Pi}(\tau))$ defined in Section \ref{sec:sol_algorithms} is consistent  for $(\theta(\tau),\Pi(\tau))$. Specifically, for the independent  data case (across  $i$ and $t$),  under suitable regularity conditions that can be found in Section \ref{sec:theory}, we have
	\begin{equation}
	\label{eqn:beta_rate}
	\displaystyle  \|\hat{\theta}(\tau) -\theta(\tau)\|  \,=\,  O_{\mathbb{P}}\left(  \max\{\sqrt{\log p},\sqrt{\log n}   \} (  \sqrt{s_{\tau}} +  \sqrt{r_{\tau}} )\left(  \frac{1}{ \sqrt{n} } + \frac{1}{\sqrt{T}}    \right)    \right).
	\end{equation}
	and
	\begin{equation}
	\label{eqn:pi_rate}
	\displaystyle  \frac{1}{nT} \| \hat{\Pi}(\tau) -\Pi(\tau) \|_F^2    \,=\,  O_{\mathbb{P}}\left(  \max\{\log p,\log n  \}(s_{\tau}  + r_{\tau} )  \left( \frac{1}{n}  + \frac{1}{T} \right)     \right),
	\end{equation}
	Importantly, the rates in (\ref{eqn:beta_rate}) and (\ref{eqn:pi_rate}), up to logarithmic factor, match those in previous works. However,  our setting allows for modeling at different quantile  levels. We also complement our results by allowing for the possibility of lagged dependent  data. Specifically,  under a $\beta$-mixing assumption, Theorem  \ref{thm:1}  provides a statistical  guarantee for  estimating $(\theta(\tau), \Pi(\tau))$. This result can be thought as a generalization of the statements in (\ref{eqn:beta_rate}) and (\ref{eqn:pi_rate}).
	
Let $ q_{i,t} \,=\,  X_{i,t}^{\prime}\theta_{i,t}(\tau)  + \Pi_{i,t}(\tau)$ for  $t \in \{1,\ldots,T\}$ and $i \in \{1,\ldots,n\}$  be the conditional quantiles. We 
  show that, under weaker conditions than the ones needed for Theorem 4.2, our estimates $\{\hat{q}_{i,t}\}$  satisfy 
	\[
\begin{array}{lll}
	\displaystyle  \frac{1}{nT} \sum_{i=1}^{n}  \sum_{t=1}^{T} \min\{   \vert   q_{i,t} -\hat{q}_{i,t}  \vert,(q_{i,t} -\hat{q}_{i,t} )^2    \} &=& \displaystyle  O_{\mathbb{P}}\bigg(\left(   \frac{1}{\sqrt{n}}  +  \frac{1}{\sqrt{T}}  \right)\bigg(      \frac{\| \Pi(\tau)\|_* }{nT}   \,+\,\\
	& &\,\,\,\,\,\,\,\,\,\,\,\sqrt{ \log(\max\{  n,p \}) }\|\beta(\tau)\|_1  \bigg)   \bigg), 
\end{array}
\]
for the independent  data case (across  $i$ and $t$).   This is a particular instance of Theorem  \ref{thm_c} which allows the possibility of time dependence.

	\item An  important  aspect of our analysis is that  we contrast the performance  of our  estimator  in settings  where the possibility of a dense  $\theta(\tau)$ provided that the features are highly correlated. We show that there exist  choices of  the tuning parameters for our estimator that lead to consistent estimation.
	
	\item For estimation, we provide an efficient algorithm (details can be found in Section \ref{sec:sol_algorithms}), which is based on the ADMM algorithm (\cite{boyd2011distributed}).
	
	\item Section \ref{sec:Empirical} provides thorough examples on financial data that illustrate the flexibility and interpretability of our approach. 
\end{itemize}

Although our theoretical analysis builds on the work by \cite{belloni2011l1}, there are multiple challenges that we must face in order to prove the consistency of our estimator. First, the construction of the restricted  set  now involves the nuclear norm penalty. This requires us to define a new restricted set  that captures the contributions of the low-rank matrix. Second, when bounding the empirical processes that  naturally arise in our proof, we have to simultaneously deal with
the sparse and dense components. Furthermore,  throughout our proofs, we have to carefully handle the weak dependence assumption that can be found in Section \ref{sec:theory}.

\section{High Dimensional Latent Panel Quantile Regression}  \label{sec:sol_algorithms}
 
In this subsection, we describe the main steps of our proposed ADMM algorithm, details can be found in Section \ref{sec:admm}. We start by introducing slack variables to the original  problem (\ref{eqn:formulation2}).  As a result, a problem equivalent to (\ref{eqn:formulation2}) is
\begin{equation}
\label{eqn:formulation3}
\begin{array}{lll}
\underset{  \stackrel{\tilde{\theta} , \tilde{\Pi},V }{Z_{\theta},Z_{\Pi },W }   }{\min} &  \displaystyle \frac{1}{nT}\sum_{i=1}^n \sum_{t=1  }^T  \rho_{\tau}(V_{i,t}) \,+\,   \nu_{1}   \sum_{j=1}^{p}  w_j \vert  Z_{\theta_j} \vert \,+\,  \nu_2 \| \tilde{\Pi} \|_*   \\
\text{subject to} &    V =  W,\,\,\,\,\,  W= Y-  X\tilde{\theta} - Z_{\Pi },\\
& \,\,\,\,\, Z_{\Pi }   - \tilde{\Pi} =0,\,\,\,\,\, Z_{\theta}  -\tilde{\theta}   =0.  \\
\end{array}
\end{equation}

To solve  (\ref{eqn:formulation3}), we propose a scaled version of the ADMM algorithm which relies on the following Augmented Lagrangian
\begin{equation}
\label{eqn:augmented_lagranian}
\def\arraystretch{1.6}
\begin{array}{lll}
\mathcal{L}(\tilde{\theta} , \tilde{\Pi},V Z_{\theta},Z_{\Pi },W,U_V, U_W, U_{\Pi},U_{\theta}  ) &=&  \displaystyle \frac{1}{nT}\sum_{i=1}^n \sum_{t=1  }^T  \rho_{\tau}(V_{i,t}) \,+\,   \nu_{1}   \sum_{j=1}^{p}  w_j \vert  Z_{\theta_j} \vert \,+\,  \nu_2 \| \tilde{\Pi} \|_*\\
& &  \displaystyle  +\,  \frac{\eta}{2}\| V -  W +  U_V \|_F^2 \, +\,   \frac{\eta}{2}\|W- Y  +  X\theta + Z_{ \Pi }   + U_W\|_F^2\\
& &   \displaystyle + \,  \frac{\eta}{2} \|   Z_{\Pi} -    \tilde{ \Pi} + U_{\Pi}\|_F^2 \,  +\,  \frac{\eta}{2}  \|    Z_{\theta}   - \tilde{ \theta}  + U_{\theta}\|_F^2,
\end{array}
\end{equation}
where  $\eta >0$ is a penalty parameter.

Notice  that in  (\ref{eqn:augmented_lagranian}),  we have followed the usual construction of ADMM via introducing
the  scaled  dual  variables  corresponding to the constraints  in (\ref{eqn:formulation3}) -- those are   $U_V$, $U_W$, $U_{\Pi}$,
and $U_{\theta}$. Next, recall that ADMM proceeds by  iteratively  minimizing  the  Augmented  Lagrangian  in blocks with respected  to  the original  variables,  in our case  $(V, \tilde{ \theta},\tilde{ \Pi})$  and  $(W,Z_{\theta},Z_{\Pi})$, and  then updating  the  scaled dual  variables (see Equations  3.5--3.7 in  \cite{boyd2011distributed}).   The explicit  updates  can be found in the Supplementary Material.  Here,   we highlight the updates for  $Z_{\theta}$, $\tilde{\Pi}$, and $V$.     For updating $Z_{\theta}$  at iteration  $k+1$, we  solve the problem
\[
\displaystyle	Z_{ \theta   }^{(k+1)}  \,\leftarrow \,\underset{Z_{\theta} \in \mathbb{R}^{ p} }{\arg \min}\left\{ \frac{1}{2}\| Z_{\theta} -\tilde{\theta}^{(k+1)} +  U_{\theta}^{(k)}    \|_F^2  +  \frac{ \nu_{1}}{\eta} \sum_{j =1 }^{p} w_j  \vert (Z_{\theta})_j   \vert   \right\}.
\]
This can be solved   in closed form  exploiting the  well known   thresholding operator, see the details in Section \ref{sec:proog_thm1}. As for updating $\tilde{\Pi}$, we solve

\begin{equation}
\label{eqn:update_pi}
\displaystyle\tilde{ \Pi}^{(k+1)}   \,\leftarrow\,\underset{ \tilde{ \Pi} \in \mathbb{R}^{ n \times T} }{\arg \min}\left\{ \frac{ \nu_2 }{\eta}\|\tilde{ \Pi } \|_*   +      \frac{1}{2}\| Z_{\Pi}^{(k)} - \tilde{\Pi}   +  U_{\Pi }^{(k)}  \|_F^2  \right\},
\end{equation}
via the  singular value shrinkage operator, see Theorem 2.1 in   \cite{cai2010singular}. 

Furthermore,  we update $V$, at iteration $k+1$, via
\begin{equation}
\label{eqn:update_v}
\displaystyle V^{(k+1)} \,\leftarrow \, \underset{V \in \mathbb{R}^{ n \times T} }{\arg \min}\,\,\left\{  \frac{1}{nT}\sum_{ i= 1}^n \sum_{ t= 1}^T  \rho_{\tau}(V_{i,t})  +  \frac{\eta}{2}\|V-W^{(k)}+ U_V^{(k)}\|_F^2  \right\},
\end{equation}
which can be found in closed  formula  by  Lemma 5.1 from \cite{ali2016multiple}.

\begin{remark}
	After estimating  $\Pi(\tau)$, we can estimate  $\lambda_i(\tau) $ and $g_t(\tau)$ via the singular value decomposition of $\hat{\Pi}(\tau)$ and following equation
	\begin{equation}
	\label{eqn:svd}
	\hat{\Pi}(\tau)_{i,t} \,=\,     \hat{\lambda}_i(\tau)^{\prime} \hat{g}_t(\tau),
	\end{equation}
	where  $\hat{\lambda}_i(\tau) $ and $\hat{g}_t(\tau)$  are of dimension $\hat{r}_{\tau}$. This immediately  leads to factors and loadings estimated that can be used  to  obtain insights about the structure of the data. A formal  identification  statement  is given in  Corollary \ref{cor1}.
\end{remark}

Finally, it is immediate to modify the proposed ADMM to the case when there are no covariates ($\beta(\tau) =0$), see  Section \ref{sec:no_covariates} in the Supplementary Material. Hence,   our proposed estimation procedure can be applied to  settings (i) with low dimensional covariates, or (ii) without covariates. However,  in what follows, we focus on the high dimensional covariates setting.

\section{Theory} \label{sec:theory}
The purpose of this section is  to provide  statistical  guarantees  for the estimator  developed in the previous section.  We focus on   estimating the  quantile function, allowing for the high dimensional scenario  where  $p$ and $T$  can  grow  as $n$  grows.

\subsection{ Estimating the quantiles} \label{sec:theory_dense}

We show that our proposed estimator  is consistent for estimating the conditional quantiles.  

Throughout, we treat $\Pi(\tau)$ as  fixed parameters. As for the data generation process,  our next condition requires  that the observations are independent  across $i$, and weakly dependent  across time.
 
\begin{assumption}
	\label{cond1}
	We assume that 
	\[
	Y_{i,t} |   X_{i,t}; \theta(\tau), \Pi_{i,t}(\tau)\,  = \, X_{i,t}^{\prime} \theta(\tau)+\Pi_{i,t}(\tau) +  \epsilon_{i,t},
	\]
	where   $\mathbb{P}( \epsilon_{i,t} \leq 0  | X_{i,t}^{\prime} \theta(\tau)+\Pi_{i,t}(\tau)) = \tau$. Furthermore, the following holds:
	\begin{enumerate}[label=(\roman*)]
		\item  There exists a function $G \,:\, [0,1]^d \rightarrow [g_1,g_2]$ for positive constants $g_1$ and $g_2$ such that,  conditional on $\Pi$ and  $\{  X_{i,t} \}_{i \in [n], t \in [T] }$,  $\epsilon_{i,t} =  \varepsilon_{i,t} G(X_{i,t})$  where 
		$\{ \varepsilon_{i,t}\}_{t=1,\ldots,T}$  are independent  across $i$. Also,  for each $i  \in [n]$, the sequence  $\{ \varepsilon_{i,t}\}_{t=1,\ldots,T}$ is stationary and $\beta$-mixing with mixing coefficients satisfying $\sup_i \gamma_i(k)  \,=\, O(k^{-\mu} )$ for some $\mu>  2$. Moreover, there exists  $\mu^{\prime }   \in (0, \mu) $, such that
		\begin{equation}
		\label{eqn:mixing}
		npT  \left( \floor*{  T^{   1/(1+ \mu^{\prime} ) }  } \right)^{-\mu }\,\rightarrow\, 0.
		\end{equation}
		Here,
		\[
		\begin{array}{lll}
		\gamma_{i}(k) &=&\displaystyle \frac{1}{2}\ \underset{l\geq1}{\sup}\bigg\{\sum_{j=1}^{L}\sum_{j^{\prime}=1}^{L^{\prime}}\vert\mathbb{P}(A_{j}\cap B_{j^{\prime}})-\mathbb{P}(A_{j})\mathbb{P}(B_{j^{\prime}})\vert \,\bigg| \,    \text{with}\,\,\{A_{j}\}_{j=1}^L   \,\,\,\\
		& &\,\,\,\,\,\,\,\,\,\,\,\,\,\,\,\,\,\,\text{paritition of }\,\,\,  \sigma(\{\varepsilon_{i,1}\},\ldots,\{\varepsilon_{i,l}\}),\,\,\,\text{and}\,\,\{B_{j^{\prime}}\}_{j^{\prime}=1}^{ L^{\prime} }\\
		&&\,\,\,\,\,\,\,\,\,\,\,\,\,\,\,\,\,\,\text{paritition of }\,\,\,\sigma(\{\varepsilon_{i,l+k}\},\{\varepsilon_{i,l+k+1}\}\ldots)\bigg\}.
		\end{array}
		\]
	 
		\item There exists $\underline{f} >0$  satisfying
		\[
		\underset{1\leq i\leq n,\,1\leq t\leq T,\,x\in\mathcal{X},   \vert \tilde{\delta}_{i,t} \vert \leq L \,}{\inf}\,f_{Y_{i,t}|X_{i,t};\theta(\tau),\Pi_{i,t}(\tau)}(x^{\prime}\theta(\tau)+\Pi_{i,t}(\tau) +   \tilde{\delta}_{i,t}|x;\theta(\tau),\Pi_{i,t}(\tau))\,>\,\underline{f},
		\]
		for some  $L >0$,	where $f_{ Y_{i,t} | X_{i,t}; \theta, \Pi_{i,t} }$ is the probability  density function associated with $Y_{i,t}$ when conditioning on $X_{i,t}$, and with parameters $\theta(\tau)$ and $\Pi_{i,t}(\tau)$.  
	\end{enumerate}
\end{assumption}

Note that Assumption \ref{cond1} is related to the sampling and smoothness assumption of \cite{belloni2011l1}. Furthermore, we highlight that similar to \cite{belloni2011l1}, our framework is  rich enough  that avoids  imposing  Gaussian modeling constraints. However, unlike \cite{belloni2011l1}, we consider panel data  with  weak correlation  across time. In particular,  we refer readers to \cite{yu1994rates}  for thorough discussions on $\beta$-mixing. Another difference with \cite{belloni2011l1} is that  we do not require any smoothness assumption on the condition density function of  the response given the covariates.

 It is worth mentioning that the parameter $\mu$ in Assumption \ref{cond1}  controls the strength of the time dependence in the data while parameter $\mu’$ will relate to the tuning parameters and will impact the rates. As we decrease the value of $\mu’$ condition (\ref{eqn:mixing}) is easier to satisfy but the tuning parameters increase and slow down the rates of convergence. Furthermore, in the case  that  $\{(Y_{i,t},X_{i,t})\}_{ i \in [n], t \in [T] }$ are independent  our theoretical results will hold without imposing  (\ref{eqn:mixing}).

Next, we require  that along each dimension the second  moment of the covariates is  one.  We also assume that the second moments  can be reasonably well estimated by their empirical  counterparts.

\begin{assumption}
 
	\label{cond2}
	We assume $\mathbb{E}(X_{i,t,j}^2) = 1$  for all $i \in [n],  t\in [T],j \in [p] $. Then 
	
	\begin{equation}\label{eq:sigma}
	\hat{\sigma}_j ^2\,=\,  \frac{1}{n T} \sum_{i=1}^n \sum_{t=1}^T  X_{i,t,j}^2,\,\,\,\,\, \forall  j \in [p],
	\end{equation}	
	 
	and we  require that
	\[
	\mathbb{P}\left( \underset{  1\leq  j\leq p }{\max}\, \,\vert\hat{\sigma}^2_{j}-1 \vert \,\leq \, \frac{1}{4} \right) \,\geq \, 1- \gamma  \, \,\,\,\rightarrow\,\,\,  1, \,\,\,\,\,\,\text{as}\,\,\,\,\,\ n\rightarrow \infty.
	\]
 
\end{assumption}
\vspace{0.1in}

Assumption \ref{cond2}  appeared  as  Condition  D.3 in \cite{belloni2011l1}. It is met by general  models on the covariates, see for instance Design 2 in \cite{belloni2011l1}.

Using the  empirical second order moments  $\{\hat{\sigma}_j^2 \}_{j=1}^p$, we analyze  the performance of the  constrained estimator
\begin{equation}
		\label{eqn:prob_1}
		(\hat{\theta}(\tau) , \hat{\Pi}(\tau) )  \,=\, 
		\underset{( \tilde{\theta} ,  \tilde{\Pi})}{\arg \min} \left\{ \hat{Q}_{\tau}(\tilde{\theta},\tilde{\Pi}) \,+\,   \nu_{1} \|\tilde{\theta}\|_{1,n,T}  +\nu_2  \|\tilde{\Pi}\|_* \right\},\\
	\end{equation}
 
where $\nu_1,\nu_2 \,>\, 0$ are   tuning parameters and   $\|\tilde{\theta}\|_{1,n,T}   \,:=\,   \sum_{j =1 }^{p} \hat{\sigma}_j  \vert  \tilde{\theta}_j\vert$,
\[
\displaystyle   \hat{Q}_{\tau}(\tilde{\theta},\tilde{\Pi}) \,=\,   \frac{1}{nT}\sum_{i=1}^n \sum_{t=1}^T \rho_{\tau}(  Y_{i,t} - X_{i,t}^{\prime}\tilde{\theta}  - \tilde{\Pi}_{i,t}),
\]
with $\rho_{\tau}$ as defined in Section \ref{sec:problem_formulation}.

Our main result  of this subsection  is provided next.
\begin{theorem}
	\label{thm_c}
	Let  $	(\hat{\theta}(\tau) , \hat{\Pi}(\tau) )$  be the estimator defined in (\ref{eqn:prob_1}). Let us write
	\[
	   q_{i,t} \,=\,  X_{i,t}^{\prime}\theta_{i,t}(\tau)    + \Pi_{i,t}(\tau), \,\,\,      \hat{q}_{i,t} \,=\,  X_{i,t}^{\prime}\hat{\theta}_{i,t}(\tau)    + \hat{\Pi}_{i,t}(\tau),
	\]
	for the conditional quantiles and their estimates. 	Then under Assumptions \ref{cond1}--\ref{cond2}, we have that for any sequence  $\{m_n\}$, $m_n \rightarrow \infty$, it holds that
	\[
      \begin{array}{lll}
      		\displaystyle  \frac{1}{nT} \sum_{i=1}^{n}  \sum_{t=1}^{T} \min\{   \vert   q_{i,t} -\hat{q}_{i,t}  \vert,(q_{i,t} -\hat{q}_{i,t} )^2    \} &=& \displaystyle  O_{\mathbb{P}}\bigg(  m_n\sqrt{c_T}\left(   \frac{1}{\sqrt{n}}  +  \frac{1}{\sqrt{d_T}}  \right)\bigg(      \frac{\| \Pi(\tau)\|_* }{nT}   \,+\,\\
      	 & &\,\,\,\,\,\,\,\,\,\,\,\sqrt{ \log(\max\{  n,p c_T\}) }\|\beta(\tau)\|_1  \bigg)   \bigg) 
      \end{array}
	\]
	with probability approaching one provided that  $T,n \rightarrow \infty$, and the tuning parameters satisfy
		\[
	\nu_{1} \,\asymp\,\sqrt{       \frac{   c_T \log( \max\{n,pc_T\} )  }{   n d_T}  }( \sqrt{n}  +  \sqrt{d_T} ) ,
	\]
	and
	\[
	\nu_2  \,\asymp\,  \frac{ c_T}{n T}\left( \sqrt{n}  +  \sqrt{d_T}  \right),
	\]
	where  $c_T \,=\, \ceil{  T^{ 1/(1+ \mu^{\prime} ) }  } $, $d_T   = \floor{  T/(2c_T) } $.
\end{theorem}

	Theorem  \ref{thm_c}  shows that we can consistently  estimate the  conditional quantiles at a rate that combines the sparse and dense signal magnitudes.  This result differs from \cite{belloni2011l1} in  several ways. First, \cite{belloni2011l1} work with cross sectional data and focus on the estimation of the vector of parameter coefficients instead of the conditional quantiles.  Therefore the assumptions for Theorem  \ref{thm_c} and those in  \cite{belloni2011l1} are very different.  For example,    \cite{belloni2011l1} required a minimum eigenvalue condition on the behavior of the design matrix, whereas we can avoid this. Furthermore,  unlike \cite{belloni2011l1}, Theorem  \ref{thm_c}  allows for panel data with dependence across time. This  translates into a technical challenge to arrive at Theorem  \ref{thm_c}, since it is not possible to use the analysis from \cite{belloni2011l1} which relies heavily on independence.  In fact, the minimum eigenvalue condition in \cite{belloni2011l1} (Condition D.4)  can be difficult to be verified in practice  without the independence  data assumption, and hence Theorem  \ref{thm_c}  has a different setting than those results in \cite{belloni2011l1}.  Also,  our setting and estimator involve the latent factors which make both the estimator and theory more complex than the framework in \cite{belloni2011l1}.  Lastly, while some of the ideas in the proof of Theorem  \ref{thm_c} come from \cite{padilla2020adaptive}, all of the technical steps involved in the proof of Theorem  \ref{thm_c} are novel.

	On another note, the rate of $T$ in Theorem  \ref{thm_c} must be such that $T \rightarrow \infty$ and $T$ satisfies  (\ref{eqn:mixing}). The parameter $\mu'$ in (\ref{eqn:mixing}) controls the terms $c_T=\lceil T^{1/(1+\mu')}\rceil$ and $d_T= \lfloor T/(2c_T)\rfloor$ both of which determine the tuning parameters and imply specific rates of convergence.   Cearly, $c_T \cdot d_T \asymp T$.  Moreover, the parameter $c_T$ measures the strength of the dependence in the data, whereas   $d_T$ can be interpreted as the effective number of independent samples across time.   In particular, when $c_T \asymp 1$, we have that  $\{\epsilon_{i,t}\}$ are basically independent across time, and so the final rate depends on $T$. As a different  example, since we would like $d_T \gg c_T$, with $\mu'=2$ we have $c_T \sim T^{1/3}$ and $d_T \sim T^{2/3}$ so that $\nu_1 \sim (T^{-1/6}+ \sqrt{T^{1/3}/n})\sqrt{\log(npT)}$ and $\nu_2 \sim \frac{1}{\sqrt{n}T^{2/3}}+\frac{1}{n T^{1/3}}$. In this case condition (\ref{eqn:mixing})  requires the parameter $\mu$ that governs the mixing speed to be such that $npT \ll T^{\mu/3}$.

\subsection{Estimating the coefficients and latent factors}
\label{sec:disantangle}

We show that our proposed  estimator is  consistent for estimating the vector of coefficients and latent factors separately  in a broad  range of models, and in some cases attains   minimax  rates, as in  \cite{candes2009tight}.

In order to obtain our main result of this subsection, we first provide some additional notation and  assumptions.  For a  fixed $\tau >0$, we assume that  (\ref{eqn:first_model})  holds.  We also let $T_{\tau}$ be the support of $\theta(\tau)$, thus
\[
T_{\tau } \,=\, \left\{j  \in [p] \,:\,  \theta_j(\tau)  \neq 0  \right\},
\]
and we write $s_{\tau} \,=\, \vert T_{\tau} \vert $, and $r_{\tau} \,=\,  \text{rank}(\Pi(\tau))$.

\begin{assumption}
	\label{cond4}
	Conditional on $\Pi$, $\{(X_{i,t},\epsilon_{i,t})\}_{t=1,\ldots,T}$  are independent across $i$. Also,  for each $i  \in [n]$, the sequence  $\{ (X_{i,t},\epsilon_{i,t})\}_{t=1,\ldots,T}$ is stationary and $\beta$-mixing with mixing coefficients satisfying $\sup_i \gamma_i(k)  \,=\, O(k^{-\mu} )$ for some $\mu>  2$. Moreover, there exists  $\mu^{\prime }   \in (0, \mu) $, such that (\ref{eqn:mixing}) holds. In addition,
	 there exists $\underline{f} >0$  satisfying
	\[
	\underset{1\leq i\leq n,\,1\leq t\leq T,\,x\in\mathcal{X},\,}{\inf}\,f_{Y_{i,t}|X_{i,t};\theta(\tau),\Pi_{i,t}(\tau)}(x^{\prime}\theta(\tau)+\Pi_{i,t}(\tau)|x;\theta(\tau),\Pi_{i,t}(\tau))\,>\,\underline{f},
	\]
where $f_{ Y_{i,t} | X_{i,t}; \theta, \Pi_{i,t} }$ is the probability  density function associated with $Y_{i,t}$ when conditioning on $X_{i,t}$, and with parameters $\theta(\tau)$ and $\Pi_{i,t}(\tau)$. Furthermore,  $ f_{ Y_{i,t} | X_{i,t}; \theta(\tau), \Pi_{i,t}(\tau) }(y| x; \theta(\tau), \Pi_{i,t}(\tau)  )$  and $\frac{\partial}{\partial y} f_{ Y_{i,t} | X_{i,t}; \theta(\tau), \Pi_{i,t}(\tau) }(y|x; \theta(\tau), \Pi_{i,t}(\tau))$  are both bounded  by $\bar{f}$ and $\bar{f}^{\prime}$, respectively, uniformly in $y$ and $x$ in the support  of $X_{i,t}$.
	
\end{assumption}

As it can been  seen in Lemma \ref{lem1} from the Supplementary Material,  the error of  our estimator defined in  (\ref{eqn:prob_1}), $(\hat{\theta}(\tau) - \theta(\tau) , \hat{\Pi}(\tau) - \Pi (\tau) ) $, belongs to a  restricted set, which in our framework is defined as
\begin{equation}
	\label{eqn:restricted_set}
	\begin{array}{l}
		A_{\tau}  = \Bigg\{ (\delta, \Delta )\in \mathbb{R}^p\times  \mathbb{R}^{n\times T }\,:\,    \|\delta_{ T_{\tau}^c }\|_1  + \frac{ \|  \Delta \|_* }{\sqrt{nT}  \sqrt{ \log(  \max\{n,p c_T\} )} }  \leq    C_0 \left(  \|\delta_{T_{\tau}}\|_1  +  \frac{ \sqrt{r_{\tau}} \| \Delta\|_F }{\sqrt{nT}  \sqrt{ \log(  \max\{n,p c_T\} )} } \right)   \Bigg\},\\
	\end{array}
\end{equation}
for  an appropriate  positive constant $C_0$.

Similar in spirit to other  high dimensional settings  such as those in \cite{candes2007dantzig}, \cite{bickel2009simultaneous}, \cite{belloni2011l1} and \cite{dalalyan2017prediction}, we impose  an identifiability  condition involving the restricted  set which is expressed  next and  will  be used  in order to attain our main results. This is the key difference with the analysis in Section \ref{sec:theory_dense}.  Before  arriving at our next condition, we introduce some notation.

For $m\geq 0$,  we denote by $\overline{T}_{\tau }(\delta,m) \subset  \{  1,\ldots,p\} \backslash T_{\tau }$  the support ot the $m$  largest components, excluding entries in $T_{\tau }$, of the vector $(\vert\delta_1\vert, \ldots, \vert \delta_p \vert)^T$. We also  use the convention $\overline{T}_{\tau }(\delta,0) \,=\,\emptyset$.

\begin{assumption}
	\label{cond3}
	For $(\delta,\Delta) \in A_{\tau }$, let
	\[
	\displaystyle  J_{\tau }^{1/2}(\delta,\Delta)  \,:=\,   \sqrt{  \frac{ \underline{f}  }{nT}  \sum_{i=1}^n\sum_{t=1}^{T}   \mathbb{E}\left( \left( X_{i,t}^{\prime}\delta + \Delta_{i,t}   \right)^2      \right)       }.
	\]
	Then there exists $m \geq 0$  such that
	\begin{equation}
		\label{eqn:identifiability}
		0\,<\,  \kappa_m \,:=\, \,\underset{  (\delta,\Delta) \in A_{\tau},  \delta \neq 0  }{\inf}\,  \,\frac{     J_{\tau}^{1/2}(\delta,\Delta)   }{   \| \delta_{T_{\tau} \cup  \overline{T}_{\tau}(\delta,m)} \|    +  \frac{\|\Delta\|_F  }{ \sqrt{nT}}   },
	\end{equation}
	where  $c_T \,=\, \ceil{  T^{ 1/(1+ \mu^{\prime} ) }  } $  for $\mu^{\prime}$ as defined in Assumption \ref{cond2}.
	Moreover, we assume that the following  holds
	\begin{equation}
		\label{eqn:nonlinearity}
		0 \,<\,    q\,:=\, \displaystyle     \frac{3}{8} \,  \frac{   \underline{f}^{3/2}  }{  \overline{f}^{\prime} } \,\underset{ (\delta,\Delta) \in A_{\tau},  \delta \neq 0    }{\inf}\,\frac{  \left( \mathbb{E}   \left( \frac{1}{nT}\sum_{i=1}^{n}\sum_{t=1}^{T}  (  X_{i,t}^{\prime} \delta  + \Delta_{i,t}  )^2   \right) \right)^{3/2}   }{ \mathbb{E} \left(    \frac{1}{nT}   \sum_{i=1}^{n}\sum_{t=1}^{T}     \vert  X_{i,t}^{\prime} \delta    + \Delta_{i,t} \vert^3     \right)      },
	\end{equation}
	with  $\underline{f}$ and  $\overline{f}^{\prime}$ as in Assumption \ref{cond1}.
\end{assumption}

Few comments are in order. First, if $\Delta =0$ then  (\ref{eqn:identifiability}) and (\ref{eqn:nonlinearity})  become the restricted identifiability and nonlinearity conditions as of \cite{belloni2011l1}.  Second,   the
denominator of  (\ref{eqn:identifiability})  contains the term $\|\Delta\|_F /(\sqrt{nT})$. To see why this is reasonable,  consider the case where  $\mathbb{E}(X_{i,t}) =0$, and $X_{i,t}$ are i.i.d..  Then
\[
J_{\tau }(\delta,\Delta)  \,=\,  \underline{f} \,  \mathbb{E}( (\delta^{\prime}  X_{i,t})^2 )  +   \frac{ \underline{f}  }{nT}  \|\Delta\|_F^2.
\]
Hence,  $\|\Delta\|_F /(\sqrt{nT})$ appears also in the numerator of (\ref{eqn:identifiability}) and  it is not restrictive its presence in the denominator of (\ref{eqn:identifiability}).

We now state  our  result  for estimating  $\theta(\tau)$ and $\Pi(\tau)$.

\begin{theorem}
	\label{thm:1}
Let  $	(\hat{\theta}(\tau) , \hat{\Pi}(\tau) ) $ be the estimator defined in (\ref{eqn:prob_1}).	Suppose that Assumptions \ref{cond2}--\ref{cond3}  hold and that
	\begin{equation}
		\label{eqn:signal}
		q\, \geq  C  \,   \frac{ \phi_{n} \sqrt{ c_T \log (pc_T \vee n)  }(\sqrt{s_{\tau}+1}+  \sqrt{r_{\tau}/  \log (pc_T \vee n)  }   ) (\sqrt{n} +\sqrt{d_T})  }{ \sqrt{n d_T} \kappa_{0}   \underline{f}^{1/2} },
	\end{equation}
	for a large enough constant $C$, and $\{\phi_{n}\}$  is a sequence  with $\phi_n/(\sqrt{\underline{f} } \log(c_T+1) )\,\rightarrow \, \infty$.
	Then
	\begin{equation}
		\|\hat{\theta}(\tau) -\theta(\tau)\| \,=\, O_{\mathbb{P}}\left(    \frac{  \phi_{n} \left(1+ \sqrt{\frac{s_{\tau}}{m}}\right)  }{\kappa_m   }   \frac{  \sqrt{ c_T \,  \log(p c_T \vee n  )}(\sqrt{ 1+s_{\tau}}+\sqrt{\frac{r_{\tau}}{\log(p c_T \vee n  )}}  ) }{\kappa_{0}   \underline{f}^{1/2} }\left( \frac{1}{\sqrt{n}} +\frac{1}{\sqrt{d_T}} \right )   \right),
	\end{equation}
	and
	\begin{equation}
		\label{eqn:upper_bound1}
		\frac{1}{nT} \| \hat{\Pi}(\tau) -\Pi(\tau) \|_F^2 \,=\,  O_{\mathbb{P}}\left(   \frac{  \phi_{n}^2  c_T   \log(p c_T \vee n  ) ( 1+s_{\tau}   +   \frac{r_{\tau}}{\log(p c_T \vee n  )} ) }{\kappa_{0}^4 \, \, \underline{f} }\left( \frac{1}{n} +\frac{1}{d_T} \right ) \right),
	\end{equation}
	for choices of the tuning parameters as in Theorem \ref{thm_c}.
 
\end{theorem}

Theorem  \ref{thm:1} gives  an upper  bound on the performance  of   $(\hat{\theta}(\tau),\hat{\Pi}(\tau))$ for estimating the vector  of coefficients  $\theta(\tau)$ and the latent matrix $\Pi(\tau)$. For simplicity,  consider  the case of  i.i.d  data. Then 
the convergence  rate of  our estimation  of $\theta(\tau)$,  under the Euclidean norm, is in the order of $(\sqrt{s_{\tau}} +   \sqrt{ r_{\tau} })/\min\{  \sqrt{n}, \sqrt{d_T} \}   $, if we ignore  all the other  factors.  Hence,   we  can consistently  estimate  $\theta(\tau)$ provided  that  $  \max\{s_{\tau},r_{\tau} \}  <<  \min\{ n,T \} $. This is similar to the low-rank condition in \cite{negahban2011estimation}.  In the  low dimensional  case  $s_{\tau} = O(1)$,  the rate   $\sqrt{  r_{\tau} }/\min\{  \sqrt{n}, \sqrt{d_T} \}   $
matches  that of Theorem 1 in \cite{moon2018nuclear}.  However, unlike \cite{moon2018nuclear},  our estimator is based on a loss function that is  robust to outliers, and  our assumptions also allow for weak dependence across time,  making  our framework  potentially  more general. Furthermore,  the same applies to our rate on the mean squared  error  for estimating $\Pi(\tau)$,  which also matches that in Theorem 1 of \cite{moon2018nuclear}.

With regards to the novelty of Theorem \ref{thm:1}, we highlight that, while its proof is similar in spirit to that of Theorem 2 in \cite{belloni2011l1}, there are some significant differences. First, the construction of the restricted set used in Theorem \ref{thm:1} involves two different penalties which makes challenging  to disentangle the behavior of  $\hat{\beta}(\tau)$  and $\hat{\Pi}$, whereas \cite{belloni2011l1} only had to dealt with one penalty. Second, the empirical processes in \cite{belloni2011l1} all involved independent data whereas as the proof of Theorem \ref{thm:1} handles the time dependence of our model and the latent factors.

Interestingly, it is expected  that  the rate  in Theorem \ref{thm:1}  is optimal. To elaborate  on this point, consider the simple case  where $n   = T$, $\theta = 0$, $\tau = 0.5$, and $e_{i,t}  :=  Y_{i,t}  - \Pi_{i,t}(\tau)$  are mean zero  i.i.d. sub-Gaussian($\sigma^2)$. The latter implies that
\[
\mathbb{P}( \vert  e_{1,1}\vert  > z)  \,\leq \, C_1  \exp\left( -\frac{z^2}{2 \sigma^2} \right),\,\,\,\,
\]
for a positive constant  $C_1$, and  for all $z >0$. Then by  Theorem 2.3  in \cite{candes2009tight}, we have the following lower  bound for  estimating  $\Pi(\tau)$:
\begin{equation}
	\label{eqn:lower_bound}
	\underset{ \hat{\Pi}  }{\inf}  \,\,\underset{  \Pi(\tau)  \,:\,   \text{rank}(\Pi(\tau) )\leq r_{\tau}  }{\sup}\,\mathbb{E}\left(  \frac{\| \hat{\Pi}(\tau) - \Pi(\tau)   \|_F^2 }{n T}  \right) \,\geq \, \frac{r_{\tau}  \sigma^2 }{n}.
\end{equation}
Notably,  the lower bound in  (\ref{eqn:lower_bound}) matches the rate implied  by Theorem \ref{thm:1}, ignoring  other factors depending on $s_{\tau}$, $\kappa_0$,  $\kappa_m$, $p$ and $\phi_n$. However, we highlight that  the upper bound (\ref{eqn:upper_bound1}) in Theorem \ref{thm:1} holds without the perhaps restrictive condition that the errors are sub-Gaussian.

We conclude this  section with a result  regarding the estimation of the  factors and loadings of the  latent matrix $\Pi(\tau)$. This is expressed in Corollary
\ref{cor1}  below and is immediate consequence  of Theorem \ref{thm:1} and   Theorem 3 in  \cite{yu2014useful}.

\begin{corollary}
	\label{cor1}
	Suppose  that the all the conditions of Theorem    \ref{thm:1} hold. Let  $\sigma_1(\tau)  \geq   \sigma_2(\tau) \geq \ldots \geq  \sigma_{r_{\tau}}(\tau) > 0$ be the    singular values  of $\Pi(\tau)$, and  $\hat{\sigma}_1(\tau)\geq \ldots \geq \hat{\sigma}_{  \min\{n,T\} }(\tau)$  the singular values of $\hat{\Pi}(\tau)$. Let  $g(\tau), \hat{g}(\tau) \in \mathbb{R}^{ T \times  r_{\tau} }$ and  $\lambda(\tau), \hat{\lambda}(\tau),\tilde{\lambda}(\tau), \tilde{\hat{\lambda}}(\tau) \in \mathbb{R}^{ n \times r_{\tau}  }$ be matrices   with orthonormal   columns   satisfying
	\[
	\displaystyle 	 \Pi(\tau )   \,=\,  \sum_{ j=1}^{r_{\tau}} \sigma_j \,  \tilde{\lambda}_{\cdot,j}(\tau) g_{\cdot,j}(\tau)^{\prime} \,=\,\sum_{ j=1}^{r_{\tau}}  \lambda_{\cdot,j}(\tau) g_{\cdot,j}(\tau)^{\prime}  ,
	\]
	and  $\hat{\Pi}(\tau)  \hat{g}_{\cdot,j}(\tau)  =  \hat{\sigma}_j(\tau)  \tilde{\hat{\lambda}}_{\cdot,j}(\tau) = \hat{\lambda}_{\cdot,j}(\tau) $  for $j =  1,\ldots,r_{\tau}$. Then
	\begin{equation}
		\label{eqn:rate_svd}
		\displaystyle
	 \underset{  O  \in \mathbb{O}_{r_{\tau}}  }{\min}  \|  \hat{g}(\tau) O  - g(\tau)\|_F   \,=\,     O_{\mathbb{P}}\left(  \frac{   (\sigma_1( \tau)  +
			\sqrt{r_{\tau} }\, \mathrm{Err}   )  \mathrm{Err}}{(\sigma_{r_{\tau} -1}( \tau))^2  - (\sigma_{r_{\tau} }( \tau))^2 }    \right),		
	\end{equation}
	and 
	\begin{equation}
		\label{eqn:rate_svd2}
		\begin{array}{lll}
			\displaystyle       \frac{\|   \hat{\lambda}(\tau) -  \lambda(\tau)\|_F^2}{nT} &\,=\,&O_{\mathbb{P}}\Bigg( \frac{  r_{\tau}\,\phi_{n}^2  c_T \log(p c_T \vee n  )((1+s_{\tau} +  r_{\tau}/ \log(p c_T \vee n  )  ) }{\kappa_{0}^4 \, \, \underline{f} }\left( \frac{1}{n} +\frac{1}{d_T} \right )  \,+\,  \\
			& &  \,\,\,\,\,\,\,\,\,\,\,\,\,\,\,\,\,\,\frac{\sigma_1^2}{nT}     \frac{   (\sigma_1( \tau)  +
				\sqrt{r_{\tau} }\, \mathrm{Err}  )^2  \mathrm{Err} ^2 }{  \left( (\sigma_{r_{\tau} -1}( \tau))^2  - (\sigma_{r_{\tau} }( \tau))^2  \right)^2  }    \Bigg).
		\end{array}
	\end{equation}
	%
	Here,  $ \mathbb{O}_{r_{\tau}}$  is the group of $r_{\tau} \times r_{\tau}$  orthonormal matrices, and
	\[
	\mathrm{Err} \,:=\,   \frac{  \phi_{n} c_T  \sqrt{\log(p c_T \vee n  )} (\sqrt{1+s_{\tau} }   +   \sqrt{r_{\tau}/  \log(p c_T \vee n  ) } ) }{\kappa_{0}^2 \, \, \underline{f}^{1/2} }\left( \sqrt{n}  +  \sqrt{d_T} \right ).
	\]
\end{corollary}

A particularly interesting instance of Corollary  \ref{cor1} is when
$$(\sigma_1(\tau ))^2,\,(\sigma_{r_{\tau} -1}( \tau))^2  - (\sigma_{r_{\tau} }( \tau))^2\,\asymp \,n T, $$ a  natural setting   if  the entries of  $\Pi(\tau)$ are $O(1)$. Then the upper bound  (\ref{eqn:rate_svd})  becomes
\[
	 \underset{  O  \in \mathbb{O}_{r_{\tau}}  }{\min}  \|  \hat{g}(\tau) O  - g(\tau)\|_F  \,=\,   O_{\mathbb{P}}\left(      \frac{ \phi_n \sqrt{ c_T  \,\log(p c_T \vee n  ) }(\sqrt{1+s_{\tau}}+   \sqrt{\frac{r_{\tau}}{\log(p c_T \vee n  ) }} ) }{\kappa_{0}^2   \underline{f}^{1/2} }\left( \frac{1}{\sqrt{n}} +\frac{1}{\sqrt{d_T}} \right )   \right),
\]
whereas (\ref{eqn:rate_svd2})  is now
\[
\frac{\|   \hat{\lambda}(\tau) -  \lambda(\tau)\|_F^2}{nT}\,=\,O_{\mathbb{P}}\left(  \frac{  r_{\tau}\,\phi_{n}^2  c_T  \log(p c_T \vee n  )(1+s_{\tau}+  r_{\tau}/ \log(p c_T \vee n  ))}{\kappa_{0}^4 \, \, \underline{f} }\left( \frac{1}{n} +\frac{1}{d_T} \right )  \right).
\]

The conclusion of  Corollary  \ref{cor1}  allows us  to provide an upper bound on the estimation of factors  ($g(\tau)$)  and  loadings ($\lambda(\tau)$)  of the latent  matrix   $\Pi(\tau)$. Notice that we are not claiming that  we provide  consistent estimation of the number latent factors, as Theorem \ref{thm:1}  only guarantees consistent estimation of  $\Pi(\tau)$.  However,  other authors, e.g. \cite{moon2015linear}, have observed that estimation can be  possible even if  $r_{\tau}$ is  unknown.

\subsection{Nearly low rank quantiles}

We conclude  our theory section by studying the case  where the  matrix  $X\theta(\tau)$   has nearley   low rank. We make this formal by  imposing the condition that  $X\theta(\tau)$  can be perturbed into a   low-rank matrix. While the result in this subsection differs in the choice of tuning parameters from those  in Theorem \ref{thm:1}, our result here suggests that even in this low-rank setting  there exists a choice of tuning parameters that allows our estimator to provide consistent estimation.



We  view our setting below  as  an extension of the linear model  in \cite{chernozhukov2018inference} to the quantile framework  and reduced  rank  regression.
The  specific  condition is stated  next.

\begin{assumption}
	\label{cond5}
	With probability  approaching  one,  it holds that   $\mathrm{rank}( X\theta(\tau) +  \xi    ) =  O(r_{\tau})$, and
	\[
	\frac{\|\xi\|_{*}}{\sqrt{nT}}  \,=\,  O_{\mathbb{P}}\left( \frac{ c_T  \phi_{n} \sqrt{ r_{\tau}}(\sqrt{n} +\sqrt{d_T})  }{  \sqrt{nT}    \underline{f} }\right),
	\]
	with $c_T$ as defined in Theorem  \ref{thm:1}.	 Furthermore, $\|  X\theta(\tau) + \Pi(\tau) \|_{\infty}  =  O_{\mathbb{P}}(1)$.
\end{assumption}

Notice that in Assumption \ref{cond5},  $\xi$  is an approximation error. In the case $\xi=0$, the condition implies that     $\mathrm{rank}( X\theta(\tau)   ) =  O(r_{\tau})$  with probability  close to one.

Next, exploiting  Assumption \ref{cond5}, we show that  (\ref{eqn:prob_1})    provides consistent  estimation of the quantile   function, namely, of  $X \theta(\tau)  + \Pi(\tau)$.

\begin{theorem}
	\label{thm3}
	Suppose  that Assumptions  \ref{cond2}--\ref{cond5} hold. Let $(\hat{\theta}(\tau), \hat{\Pi}(\tau) )$  be the solution to (\ref{eqn:prob_1}) with the additional  constraint that $\|\tilde{\Pi}\|_{\infty} \leq  C$, for a large enough positive constant $C$.
	Then 
	\[
	\frac{1}{nT} \| \hat{\Pi}(\tau) -\Pi(\tau) - X\theta(\tau) \|_F^2 \,=\,  O_{\mathbb{P}}\left(   \frac{ ( \overline{f}^{\prime} )^2 \phi_{n}^2  c_T r_{\tau}  }{ \, \underline{f}^4 }\left( \frac{1}{n} +\frac{1}{d_T} \right ) \right),
	\]
	where $\{\phi_{n}\}$  is a sequence  with $\phi_n/(\sqrt{\underline{f} }  \log(1+c_T) )\,\rightarrow \, \infty$, and for choices
	
	$$
	\nu_{1}  \,\asymp\,\displaystyle \frac{1}{nT }\sum_{i=1}^n \sum_{t=1}^T   \underset{j=1,\ldots,p}{\max}\,  \left\vert   \frac{  X_{i,t,j}  }{\hat{\sigma}_j} \right\vert,
	$$
	
	and
	$$\nu_2 \,\asymp\,  \frac{c_T}{n T}\left( \sqrt{n}  +  \sqrt{d_T}  \right).$$
\end{theorem}

Interestingly,  unlike Theorem \ref{thm:1},  Theorem \ref{thm3}  does not show  that  we can estimate $\theta(\tau)$  and $\Pi(\tau)$ separately. Instead,  we  show that  $\hat{\Pi}(\tau)$, the estimated  matrix  of latent factors, captures the overall contribution of  both $\theta(\tau)$  and $\Pi(\tau)$. This is expected since Assumption  \ref{cond5} states that, with high probability, $X  \theta(\tau)$  has rank of the same order as of $\Pi(\tau)$. Notably,   $\hat{\Pi}(\tau)$   is able to estimate   $X \theta(\tau)  + \Pi(\tau)$  via requiring that the value of $ \nu_{1}$ increases  significantly with respect to the choice in Theorem \ref{thm:1},  while  keeping  $\nu_2  \asymp  c_T( \sqrt{n}  +\sqrt{d_T} )/(nT)$.

As for the convergence  rate in Theorem \ref{thm3} for estimating  $\Pi(\tau)$,  this is of the order  $r_{\tau} c_T(n^{-1}  +  d_T^{-1})$, if we ignore $\underline{f}$,  $\overline{f}^{\prime}$, and $\phi_n$.  When  the data  are independent, the rate becomes of the order  $r_{\tau} (  n^{-1}  + T^{-1} ) $. In such framework,  our result  matches  the minimax rate of estimation  in \cite{candes2010matrix} for  estimating  an $n \times T$  matrix  of rank $r_{\tau}$, provided that $n \asymp T$, see our discussion in Section  \ref{sec:disantangle}.

\section{Simulation} \label{sec:Simulation}

In this section,  we evaluate the performance of our proposed approach ($\ell_1$-NN-QR)   with extensive numerical simulations focusing on the median case, namely the case when $\tau$ = 0.5. As benchmarks, we consider  the $\ell_1$-penalized quantile regression studied in \cite{belloni2009computational}, and similarly we refer to this procedure as $\ell_1$-QR. We also  compare with the mean case, which we denote it as  $\ell_1$-NN-LS as it combines the  $\ell_2$-loss function with  $\ell_1$ and nuclear norm   regularization.  We  consider  different  generative scenarios. For each scenario  we randomly generate 100 different data sets  and  compute the  estimates of the methods for a  grid  of values of  $ \nu_{1}$ and  $ \nu_{2}$.  Specifically,  these  tuning parameters are taken to satisfy $ \nu_{1}  \in \{10^{-4}, 10^{-4.5},\ldots,10^{-8}  \}$ and   $ \nu_{2}  \in \{10^{-3}, 10^{-4} ,\ldots,10^{-9} \} $.  Given any  choice  of tuning  parameters, we  evaluate  the performance  of each competing method, averaging  over  the 100 data  sets, and report  values  that  correspond  to the best performance.  These are referred  as optimal tuning parameters and can be thought of as  oracle choices.

We also propose a modified Bayesian Information Criterion (BIC) to select the best pair of tuning parameters. Given  a pair  $( \nu_{1}, \nu_{2})$,  our method produces  a score $(\hat{\theta}(\tau),\hat{\Pi}(\tau))$.  Specifically, denote $\hat{s}_{\tau} = \vert  \{ j\,:\,   \hat{\theta}_j (\tau) \neq 0   \}\vert$ and $\hat{r}_{\tau} = \mathrm{rank}(\hat{\Pi}(\tau))$,
\begin{equation}
\label{eqn:score}
{\displaystyle 
	\begin{array}{lll}
	\mathrm{BIC}(\nu_{1},\nu_{2})&=& \displaystyle  \sum_{i=1}^{n}\sum_{t=1}^{T}\rho_{\tau}(Y_{i,t}-X_{i,t}^{\prime}\hat{\theta}(\nu_{1},\nu_{2})-\hat{\Pi}(\nu_{1},\nu_{2}))+\\
	& &\displaystyle 	\frac{\log(nT)}{2}\left(c_{1}\cdot\hat{s}_{\tau}(\nu_{1},\nu_{2})+(1+n+T)\cdot\hat{r}_{\tau}(\nu_{1},\nu_{2})\right), 
	\end{array}}
\end{equation}
where   $c_1 >0$ is a constant. The intuition  here is  that the first term in the right hand  side  of (\ref{eqn:score})  corresponds  to the  fit  to the data. The second  term  includes the factor  $\log (nT) /2$  to emulate  the usual penalization in BIC.  The  number of parameters in the  model with choices $ \nu_{1}$ and $ \nu_{2}$ is estimated  by  $\hat{s}_{\tau}$ for  the vector of coefficients, and   $(1+n+T) \cdot  \hat{r}_{\tau}$ for the   latent  matrix. The latter  is reasonable since  $\hat{\Pi}(\tau)$ is potentially  a low rank matrix  and we simply count the number of parameters in its  singular  value  decomposition. As for  the  extra quantity  $c_1$,  we have included  this  term to  balance  the  dominating contribution  of  the $(1+n+T) \cdot   \hat{r}_{\tau}   $. We find that in practice   $c_1 =  \log^2 (nT)$  gives reasonable  performance in both simulated and real data. This is the choice that we use in our experiments. Then for each of data set under  each design,  we  calculate the minimum  value of  $\mathrm{BIC}( \nu_{1}, \nu_{2})$,  over the   different  choices   of  $ \nu_{1}$ and  $ \nu_{2} $,  and report   the average  over the 100 Monte Carlo simulations. We refer to  this  as BIC-$\ell_1$-NN-QR.

As  performance measure  we use a scaled version (see Tables \ref{tab1}-\ref{tab2}) of the squared  distance  between the true  vector  of  coefficients  $\theta$ and the corresponding estimate.  We  also  consider   a different metric,  the ``Quantile error" (\cite{koenker1999goodness}):
\begin{equation}
\label{eqn:quantile_error}
\displaystyle \frac{1}{nT}  \sum_{i=1}^n \sum_{t=1}^T (F_{Y_{i,t} |X_{i,t}; \theta(\tau), \Pi(\tau)  }^{-1}(0.5) -   \hat{F}_{Y_{i,t} |X_{i,t}; \theta(\tau), \Pi(\tau)}^{-1}(0.5)  )^2,
\end{equation}
which measures  the average  squared  error  between the  quantile  functions  at the samples and their respective  estimates.  Since our simulations consider models with symmetric mean zero error, the above metric corresponds to the mean squared error for estimating the conditional expectation.

Next, we  provide  a detailed  description of each of the generative models  that we consider  in our experiments. In each model design the  dimensions of the problem are given by $ n \in  \{100,500\}$,  $p \in \{5,30\}$  and $T \in \{100,500\}$, and  we also consider the instance  $n=p=T=50$.  The covariates $\{X_{i,t}\}$ are i.i.d     $N(0,I_p)$.

\medskip
\textbf{Design 1. (Location shift model)} The  data is generated  from the model
\begin{equation}
\label{ex1}
Y_{i,t} =  X_{i,t}^{\prime} \theta +  \Pi_{i,t}  +  \epsilon_{i,t}, \,\,\,\,
\end{equation}
where   $\sqrt{3} \epsilon_{i,t} \overset{i.i.d.}{\sim}  t(3)  $,   $i = 1,\ldots,n$  and  $t =  1,\ldots,T$, with $t(3)$  the Student's t-distribution   with  3  degrees  of freedom. The scaling factor $\sqrt{3}$  simply  ensures that the  errors have variance $1$. In (\ref{ex1}),  we take the vector $\theta \in \mathbb{R}^p$  to satisfy
\[
\theta_j  = \begin{cases}
1   &\text{if}\,\,  j \in \{1,\ldots, \min\{10,p\} \}\\
0 &\text{otherwise.}
\end{cases}
\]
We also construct  $\Pi \in \mathbb{R}^{n \times T}$ to be rank one, defined as
$  \Pi_{i,t} =  5i\,(\mathrm{cos}(    4\pi t/T ))/n.$

\medskip
\textbf{Design 2. (Location-scale shift model)}  We consider the model
\begin{equation}
\label{ex2}
Y_{i,t} =  X_{i,t}^{\prime} \theta +  \Pi_{i,t}  +  ( X_{i,t}^{\prime}  \theta)\epsilon_{i,t}, \,\,\,\,
\end{equation}
where   $\epsilon_{i,t} \overset{i.i.d.}{\sim}  N(0,1) $,   $i = 1,\ldots,n$  and  $t =  1,\ldots,T$. The parameters  in $\theta$  and $\Pi$ in (\ref{ex2})   are taken to be the same as in (\ref{ex1}). The only difference now is that  we have the extra parameter  $\theta \in \mathbb{R}^{p}$, which we define as
$\theta_j = j/(2p)$ for  $j \in \{1,\ldots,p\}$. 

\medskip
\textbf{Design 3. (Location shift model with random factors)}
This  is  the same as  Design 1 with the difference that we now generate  $\Pi$ as
\begin{equation}
\label{eqn:d3}
\Pi_{i,t}   \,=\,   \sum_{k=1}^{5}   c_k  u_k  v_k^T,
\end{equation}
where
\begin{equation}
\label{eqn:d3.1}
c_k \sim U[0,1/4], \,\,\,\,  u_k  =   \frac{\tilde{u}_k}{ \|   \tilde{u}_k \|   },  \,\,\,\,  \tilde{u}_k   \sim N(0, I_n),   \,\,\,\, v_k  =   \frac{\tilde{v}_k}{ \|   \tilde{v}_k \|   },  \,\,\,\,  \tilde{v}_k   \sim N(0, I_n),\,\,\,\,  k = 1,\ldots,5.
\end{equation}

\begin{table}[h!]
	\centering
	\caption{\label{tab1} For Designs  1-2   described in the main  text, under different   values of  $(n, p,T)$, we compare  the  performance  of  different  methods.  The metrics  use  are the scaled $\ell_2$  distance for estimating $\theta(\tau)$, and the Quantile  error  defined in (\ref{eqn:quantile_error}). For  each method we report  the average, over 100 Monte Carlo simulations,  of the  two   performance  measures.  }
	\medskip
	\setlength{\tabcolsep}{4pt}
	\begin{small}
		\begin{tabular}{ rrrr|rr|rr}
			\hline
			&	&  &    &    Design  1  & &   Design  2 &\\
			\hline
			Method &	$n$  & $p$ & $T$   &    $\,\,\,\,\,\frac{\|\hat{\theta}(\tau)  -\theta(\tau) \|^2}{10^{-4}}$ &Quantile error &    $\,\,\,\,\,\frac{\|\hat{\theta}(\tau)  -\theta(\tau) \|^2}{10^{-4}}$ &Quantile error\\
			\hline
			$\ell_1$-NN-QR    &	 300 &  30 & 300   &   \textbf{0.86}              &  \textbf{0.03} &  \textbf{0.69}         & \textbf{0.03}\\
			BIC-$\ell_1$-NN-QR 	  &	 300 &  30 & 300   &   \textbf{0.86}             &\textbf{0.03}  &        0.89             &\textbf{0.03}\\
			$\ell_1$-NN-LS& 300   &  30 &  300  &      2.58                       &0.05&     1.03                        &0.04\\
			$\ell_1$-QR          &  300  &  30 &   300 &     8.18                        &4.19&      6.81                       &4.18\\
			\hline		
			$\ell_1$-NN-QR    &	 300 &  30 & 100   &   \textbf{2.99}         &  \textbf{0.04}  &   \textbf{2.39}         &   \textbf{0.04} \\
			BIC-$\ell_1$-NN-QR &	 300 &  30 & 100   &   \textbf{2.99}           &  \textbf{0.04}  &        \textbf{2.39}       &  \textbf{0.04} \\	
			$\ell_1$-NN-LS& 300   &  30 &  100  &     8.06                &0.12 &   3.11                  &0.08\\
			$\ell_1$-QR          &  300  &  30 &   100 &        41.0                    &4.19 &       26.0                     &4.19\\				
			\hline
			$\ell_1$-NN-QR    &	 300 &  5 & 300   &     \textbf{0.22}              &  \textbf{0.003}   &    \textbf{0.39}               & \textbf{0.03}\\
			BIC-$\ell_1$-NN-QR&	 300 &  5 & 300   &     \textbf{0.22}               &  \textbf{0.003}& 0.48              &  \textbf{0.03} \\
			$\ell_1$-NN-LS& 300   &  5 &  300  &   0.37                     &0.01  &   0.69                        & \textbf{0.03} \\
			$\ell_1$-QR          &  300  &  5 &   300 &     2.6                       &4.19 &   3.27                         &4.19\\
			\hline		
			$\ell_1$-NN-QR    &	 300 &  5 & 100   &         \textbf{0.50}           & \textbf{0.008}&   \textbf{0.80}                 &\textbf{0.03} \\
			BIC-$\ell_1$-NN-QR &	 300 &  5 & 100   &         0.53           &0.009  &      0.97              & \textbf{0.03}\\
			$\ell_1$-NN-LS& 300   &  5 &  100  &        1.12                    &0.02&      1.46                      &0.04\\
			$\ell_1$-QR          &  300  &  5 &   100 &        7.87                        &4.19 &        8.56                        &4.19\\
			\hline		
			$\ell_1$-NN-QR    &	 100 &  30 & 300   &    \textbf{2.97}              & \textbf{0.04}  &  \textbf{2.26}                &  \textbf{0.04} \\
			BIC-$\ell_1$-NN-QR &	 100 &  30 & 300   &   \textbf{2.97}                 & \textbf{0.04} &     2.81             &\textbf{0.04}\\
			$\ell_1$-NN-LS& 100   &  30 &  300  &   9.77                          &0.12&   3.39                          &0.06\\
			$\ell_1$-QR          &  100  &  30 &   300 &     40.0                           &4.23&    24.0                            &4.23\\
			\hline		
			$\ell_1$-NN-QR    &	 100 &  30 & 100   &    \textbf{2.3}           &  \textbf{0.04} &  \textbf{10.0}             &  \textbf{0.03}   \\
			BIC-$\ell_1$-NN-QR&	 100 &  30 & 100   &       \textbf{2.3}          &  0.06  &   11.0           &   0.04\\
			$\ell_1$-NN-LS& 100   &  30 &  100  &   8.4                  &0.79  &     13.0                &0.16\\
			$\ell_1$-QR          &  100  &  30 &   100 &     229.0                          &4.23&    177.0                  &         4.23\\
			\hline		
			$\ell_1$-NN-QR    &	 100 &  5 & 300   &   \textbf{0.64}            & \textbf{0.008}    &       \textbf{0.89}        & \textbf{0.03} \\
			BIC-$\ell_1$-NN-QR&	 100 &  5 & 300   &       0.65        & \textbf{0.008}   &      1.32         &  \textbf{0.03}\\
			$\ell_1$-NN-LS& 100   &  5 &  300  &    1.25                      &0.02&      1.67                    &0.05\\
			$\ell_1$-QR          &  100  &  5 &   300 &   8.79                             &4.23 &      8.61                          &4.23\\
			\hline		
			$\ell_1$-NN-QR    &	 100 &  5 & 100   &    \textbf{1.82}             & \textbf{0.009}&  \textbf{3.30}             & \textbf{0.03} \\
			BIC-$\ell_1$-NN-QR &	 100 &  5 & 100   &      1.85           &  0.01 &        3.74         & 0.04 \\
			$\ell_1$-NN-LS& 100   &  5 &  100  &      4.06                       &0.03&   4.45                          &0.14\\
			$\ell_1$-QR          &  100  &  5 &   100 &      32.0                        &4.23 &      27.0                         &4.23\\
			\hline
			$\ell_1$-NN-QR    &	 50&  50 & 50   &     \textbf{136.0}       &  \textbf{0.09}       &    \textbf{388.0} &  \textbf{0.21}\\
			BIC-$\ell_1$-NN-QR &	 50 &  50 & 50   &      241.0      &     0.33    &522.0& 0.59\\
			$\ell_1$-NN-LS& 50   &  50 &  50  &    2010       &  0.17      & 982&0.45\\
			$\ell_1$-QR          &  50&  50 &   50 &    1258        &     4.27    &2000.6 & 4.33\\
			\hline
		\end{tabular}
	\end{small}
\end{table}

\begin{table}[h!]
	\centering
	\caption{\label{tab2}  For Designs  3-4   described in the main  text, under different   values of  $(n, p,T)$, we compare  the  performance  of  different  methods.  The metrics  use  are the scaled $\ell_2$  distance for estimating $\theta(\tau)$, and the Quantile  error  defined in (\ref{eqn:quantile_error}). For  each method we report  the average, over 100 Monte Carlo simulations,  of the  two   performance  measures. }
	\medskip
	\setlength{\tabcolsep}{4pt}
	\begin{small}
		\begin{tabular}{ rrrr|rr|rr}
			\hline
			&	&  &    &    Design  3  & &   Design  4 &\\
			\hline
			Method &	$n$  & $p$ & $T$   &    $\,\,\,\,\,\frac{\|\hat{\theta}(\tau)  -\theta(\tau) \|^2}{10^{-4}}$ &Quantile error &    $\,\,\,\,\,\frac{\|\hat{\theta}(\tau)  -\theta(\tau) \|^2}{10^{-4}}$ &Quantile error\\
			\hline
			$\ell_1$-NN-QR    &	 300 &  30 & 300   &    \textbf{2.17}          &   \textbf{0.17}  &  \textbf{1.58}& \textbf{0.11} \\
			BIC-$\ell_1$-NN-QR 	  &	 300 &  30 & 300   &   2.17       &0.29 &    2.81              & 0.26\\
			$\ell_1$-NN-LS& 300   &  30 &  300  &     3.33                    &  0.19 & 3.54                           & 0.12\\
			$\ell_1$-QR          &  300  &  30 &   300 &    6.59                      &1.09&          12.0             &1.01\\
			\hline		
			$\ell_1$-NN-QR    &	 300 &  30 & 100   &   \textbf{10.0}     & 0.18  &  \textbf{4.61}       & \textbf{0.13}   \\
			BIC-$\ell_1$-NN-QR &	 300 &  30 & 100   &    11.0      & 0.26  &      \textbf{4.61}            & 0.16 \\	
			$\ell_1$-NN-LS& 300   &  30 &  100  &   11.0              &0.25&             9.21     &0.17\\
			$\ell_1$-QR          &  300  &  30 &   100 &    27.0                    &11.10 &              47.2             &1.10\\				
			\hline
			$\ell_1$-NN-QR    &	 300 &  5 & 300   &       \textbf{1.13}        &  \textbf{0.17}   &       \textbf{0.27}          & \textbf{0.03}\\
			BIC-$\ell_1$-NN-QR&	 300 &  5 & 300   &         1.56          & 0.33&        0.74     &    0.13   \\
			$\ell_1$-NN-LS& 300   &  5 &  300  &       1.58               &  0.19 &                0.49       &  0.05 \\
			$\ell_1$-QR          &  300  &  5 &   300 &      2.47              &1.10&        5.52                  &1.09\\
			\hline		
			$\ell_1$-NN-QR    &	 300 &  5 & 100   &      \textbf{3.04}            & \textbf{0.19}&       \textbf{0.69}          & 0.05\\
			BIC-$\ell_1$-NN-QR &	 300 &  5 & 100   &    4.37              &0.27 &         1.11       & 0.15\\
			$\ell_1$-NN-LS& 300   &  5 &  100  &   4.43                       &0.27&     1.12                     &  \textbf{0.05} \\
			$\ell_1$-QR          &  300  &  5 &   100 &    7.65                           &1.11&    13.4                          &1.10\\
			\hline		
			$\ell_1$-NN-QR    &	 100 &  30 & 300   &          \textbf{11.0}       &   \textbf{0.18}   &     \textbf{7.06}            &  \textbf{0.15}\\
			BIC-$\ell_1$-NN-QR &	 100 &  30 & 300   &    12.0               & 0.27 &    7.29             &0.24\\
			$\ell_1$-NN-LS& 100   &  30 &  300  &  12.0                        &0.34&  11.2                         &0.18\\
			$\ell_1$-QR          &  100  &  30 &   300 &     26.0                           &1.10&    51.0                       &1.12\\
			\hline		
			$\ell_1$-NN-QR    &	 100 &  30 & 100   &      \textbf{6.12}     & \textbf{0.17}  &     \textbf{32.1}       & 0.10     \\
			BIC-$\ell_1$-NN-QR&	 100 &  30 & 100   &      6.64      & 0.22    &        35.4    & 0.15\\
			$\ell_1$-NN-LS& 100   &  30 &  100  &    8.63              & 1.08&      82.0               &  0.84 \\
			$\ell_1$-QR          &  100  &  30 &   100 &     16.5                    &1.12&        267.6          &   1.13    \\
			\hline		
			$\ell_1$-NN-QR    &	 100 &  5 & 300   &       \textbf{2.99}      &  \textbf{0.18}  &     \textbf{0.86}        & \textbf{0.05}  \\
			BIC-$\ell_1$-NN-QR&	 100 &  5 & 300   &       4.19       &  0.26&     1.43       &  0.14\\
			$\ell_1$-NN-LS& 100   &  5 &  300  &     5.27                  &0.33&   1.45                   &  \textbf{0.05} \\
			$\ell_1$-QR          &  100  &  5 &   300 &         8.59                     &1.10&21.0                             &1.09\\
			\hline		
			$\ell_1$-NN-QR    &	 100 &  5 & 100   &    \textbf{12.3}           & \textbf{0.16}&     \textbf{2.15}       &\textbf{0.04} \\
			BIC-$\ell_1$-NN-QR &	 100 &  5 & 100   &      13.1     &  0.20&         2.43     &0.07  \\
			$\ell_1$-NN-LS& 100   &  5 &  100  &     15.0                      &1.09&        3.61                    &0.07\\
			$\ell_1$-QR          &  100  &  5 &   100 &      24.7
			& 1.10&       45.7                     &1.09\\
			\hline
			$\ell_1$-NN-QR    &	 50&  50 & 50   & \textbf{386.7}           &  \textbf{0.32}      &  \textbf{511}&  \textbf{0.27} \\
			BIC-$\ell_1$-NN-QR &	 50 &  50 & 50   &     898.1       &  0.97      & 1128.1 &0.74\\
			$\ell_1$-NN-LS& 50   &  50 &  50  &  605.2         &      0.44   &967.3 &0.49\\
			$\ell_1$-QR          &  50&  50 &   50 &   1035.8         &   1.15     &1697.5 &1.18\\	
			\hline
		\end{tabular}
	\end{small}
\end{table}

\medskip
\textbf{Design 4. (Location-scale shift model with random factors)}
This  is   a combination of  Designs  2 and  3. Specifically,  we generate  data as in (\ref{ex2})  but with $\Pi$  satisfying  (\ref{eqn:d3})  and  (\ref{eqn:d3.1}).

The   results in Tables  \ref{tab1}-\ref{tab2} show  a clear  advantage of our proposed  method  against  the benchmarks across the four designs  we consider. This is  true  for  estimating  the  vector  of coefficients,  and under the measure  of quantile error. Importantly,  our approach  is not only  the best under the optimal choice  of tuning parameters but it remains competitive  with the BIC  type of criteria  defined  with the score (\ref{eqn:score}).   In particular, under Designs  1 and 2, the data driven  version of our estimator, BIC-$\ell_1$-NN-QR,  performs  very closely to the ideally  tuned one $\ell_1$-NN-QR. In the more challenging settings of Designs  3 and 4,  we noticed  that BIC-$\ell_1$-NN-QR  performs  reasonably well compared to $\ell_1$-NN-QR.

\section{Empirical Performance of the ``Characteristics + Latent Factor'' Model in Asset Pricing} \label{sec:Empirical}

\subsubsection*{Data Description}

We use data from CRSP and Compustat to construct 24 firm level characteristics
that are documented to explain the cross section and time series of stock returns
in the finance and accounting literature. The characteristics we choose
include well-known drivers of stock returns such as beta, size, book-to-market,
momentum, volatility, liquidity, investment and profitability. Table
\ref{Table-names} in the Supplementary Material lists details of the characteristics used and the methods to construct the data. We follow
the procedures of \cite{green2017characteristics} to construct the characteristics
of interest. The characteristics used in our model are standardized to have zero
mean and unit variance. Figure \ref{Figure-histogram-of-returns} plots the histogram of monthly stock
returns and $9$ standardized firm characteristics. Each of them have
different distribution patterns, suggesting the potential nonlinear
relationship between returns and firm characteristics, which can be potentially captured by our quantile model.

Our empirical design is closely related to the characteristics model proposed by \cite{daniel1997evidence, daniel1998characteristics}. To avoid any “data snooping” issue cause by grouping, we conduct the empirical analysis at individual stock level. Specifically, we use the sample period from January 2000 to December 2018, and estimate our model using monthly returns (228 months) from 1306 firms that have non-missing values during this period.

\begin{figure}[!h]
	\begin{centering}
		\includegraphics[width=7cm,height=7cm]{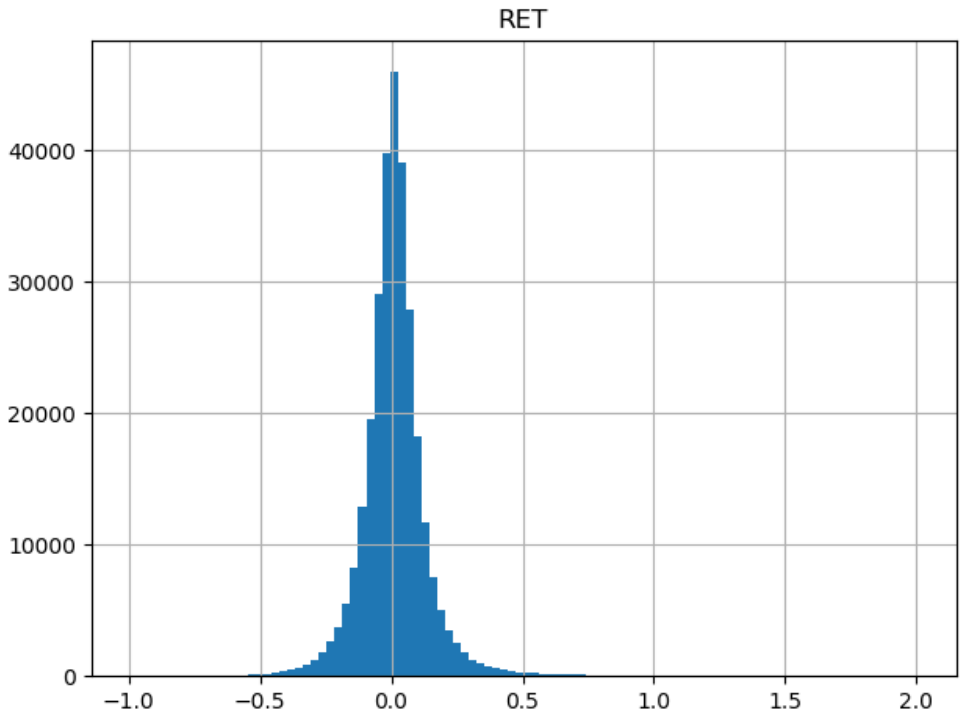}\includegraphics[width=7cm,height=7cm]{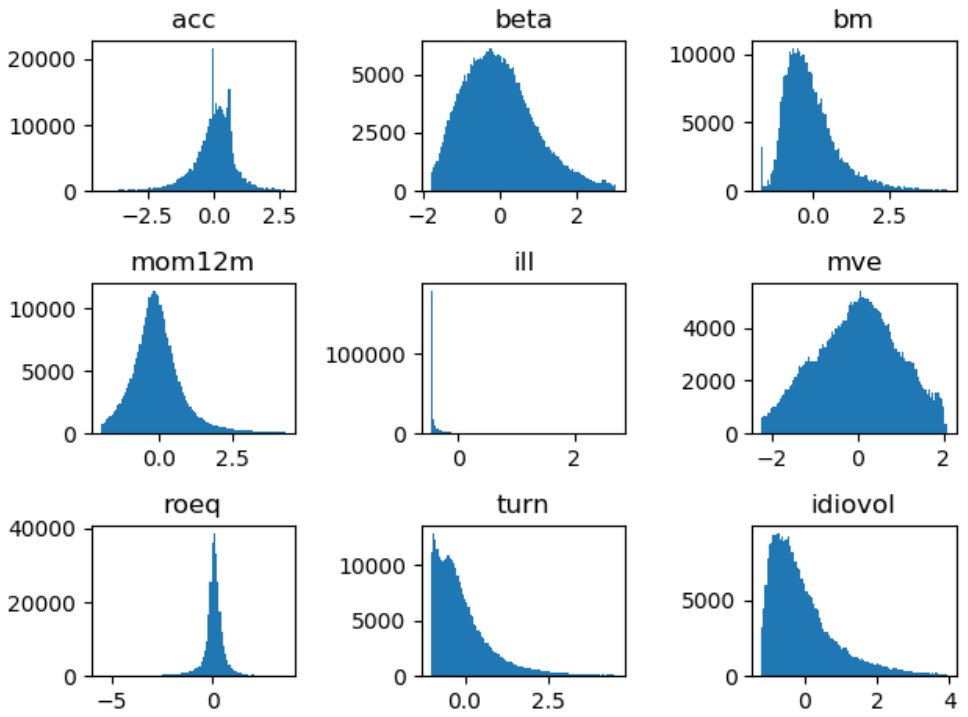}
		\par\end{centering}
	\caption{Histograms of monthly stock returns (left) and
		firm characteristics (right).}\label{Figure-histogram-of-returns}
\end{figure}

\subsubsection*{A ``Characteristic + Latent Factor'' Asset Pricing Model}

We apply our model to fit the cross section and time series of stock returns (\citep{lettau2018estimating}). There are $n$ assets (stocks), and the return of the each asset can potentially be explained by $p$ observed asset characteristics (sparse part) and $r$ latent factors (dense part). The asset characteristics are the covariates in our model. Our model imposes a sparse structure on the $p$ characteristics so that only the characteristics having the strongest explanatory powers are selected by the model. The part that's unexplained by the firm characteristics are captured by latent factors.

\begin{figure}[!h]
	\begin{raggedright}
		\includegraphics[width=15cm,height=8cm]{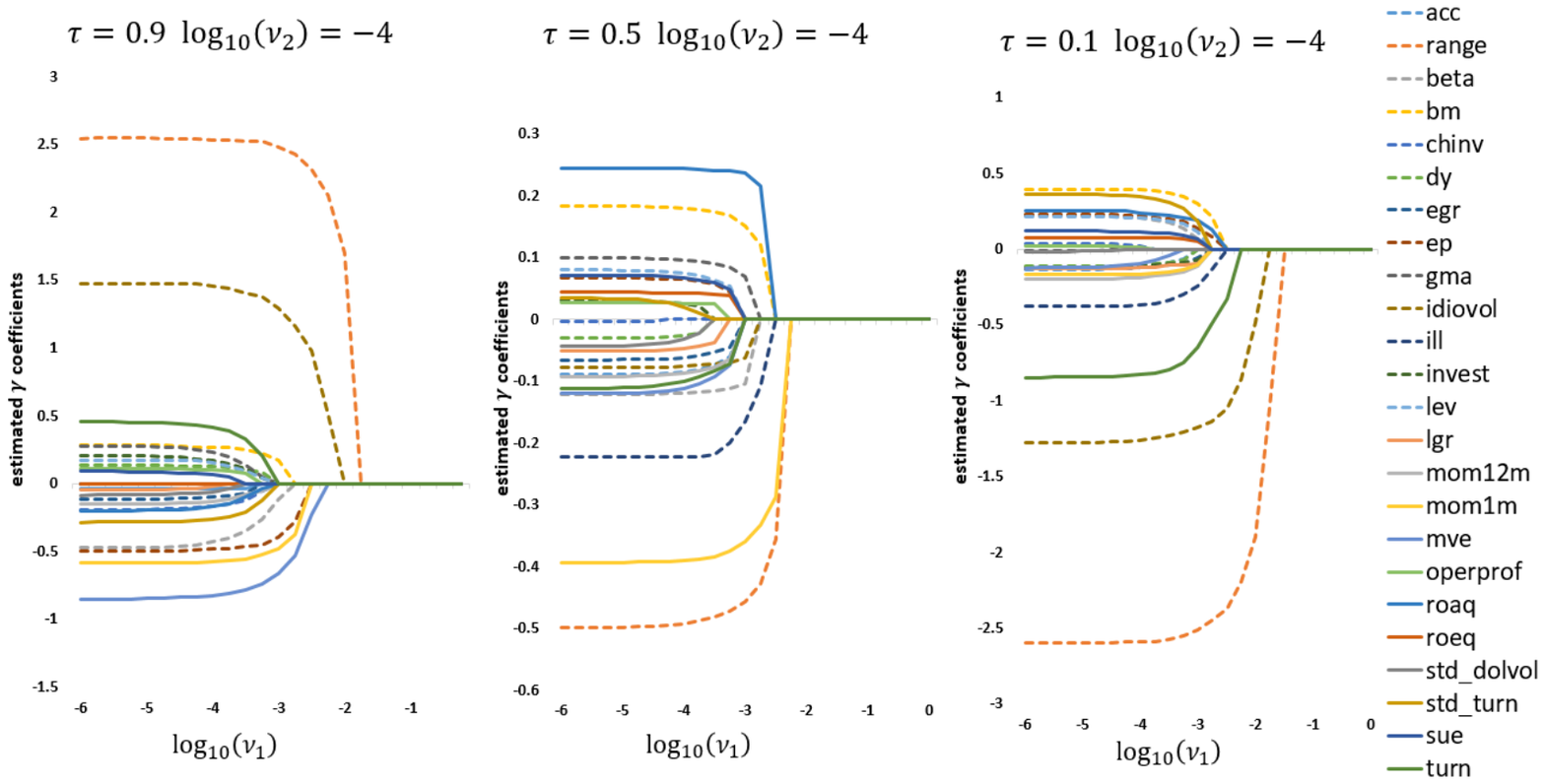}
		\par\end{raggedright}
	\caption{Estimated Coefficients as a Function of $ \nu_{1}$. 
		This figure plots the estimated coefficient $\theta$ when the tuning parameter
		$ \nu_{1}$ changes, for $\tau = \{0.1, 0.5, 0.9\}$. The parameter
		$ \nu_{2}$ is fixed at $\log_{10}( \nu_{2})=-4$.}\label{Figure-solution-pas}
\end{figure}

Suppose we have $n$ stock returns ($R_{1}$,...,$R_{n}$), and $p$ observed firm characteristics ($X_{1}$,...,$X_{p}$) over $T$ periods. The return quantile at level $\tau$ of portfolio $i$ in time $t$ is assumed to be the following:
\[
\begin{array}{ccccc}
 & F_{R_{i,t}\vert X_{i,t-1};\theta(\tau),\lambda_{i}(\tau),g_{t}(\tau)}^{-1}(\tau)=X_{i,t-1,1}\theta_{1}(\tau)+...+X_{i,t-1,p}\theta_{p}(\tau)+ & \lambda_{i}(\tau) & g_{t}(\tau) & \\
 &  & \left(1\times r_{\tau}\right) & \left(r_{\tau}\times1\right)
\end{array}
\]
where $X_{i,t-1,k}$, e.g. $k=1$ or $k=p$, is the $k$-th characteristic (e.g. the
book-to-market ratio) of asset $i$ in time $t-1$. The
coefficient $\theta_k$ captures the extent to which assets
with higher/lower characteristic $X_{i,t,k}$ delivers higher average
return. The term $g_{t}$ contains the $r_{\tau}$ latent factors in period
$t$ which captures systematic risks in the market, and $\lambda_{i}$ contains portfolio
$i$'s loading on these factors (i.e. exposure to risk).

There is a discussion in academic research on ``factor versus characteristics'' in late 1990s and early 2000s. The factor/risk based view argues that an asset has higher expected returns because of its exposure to risk factors (e.g. Fama-French 3 factors) which represent some unobserved systematic risk. An asset's exposure to risk factors are measured by factor loadings. The characteristics view claims that stocks have higher expected returns simply because they have certain characteristics (e.g. higher book-to-market ratios, smaller market capitalization), which might be independent of systematic risk (\cite{daniel1997evidence,daniel1998characteristics}). The formulation of our model accommodates both the factor view and the characteristics view. The sparse part is similar to \cite{daniel1997evidence,daniel1998characteristics}, in which stock return are explained by firm characteristics. The dense part assumes a low-dimensional latent factor structure where the common variations in stock returns are driven by several ``risk factors''.

\subsubsection*{Empirical Results}

We first get the estimates $\hat{\theta}(\tau)$ and $\hat{\Pi}(\tau)$ at three different quantiles, $\tau = \{0.1, 0.5, 0.9\}$  using our proposed ADMM algorithm. We then decompose $\hat{\Pi}(\tau)$  into the products of its $\hat{r}_\tau$ principal components $\hat{g}(\tau)$ and their loadings $\hat{\lambda}(\tau)$ via eq(\ref{eqn:svd}).
The $(i,k)$-th element of $\hat{\lambda}(\tau)$, denoted as $\hat{\lambda}_{i,k}(\tau)$, can be interpreted as the exposure of asset $i$ to the $k$-th latent factor (or in finance terminology, ``quantity of risk''). And the $(t,k)$-th elements of $\hat{g}(\tau)$, denoted as $\hat{g}_{t,k}(\tau)$, can be interpreted as the compensation of the risk exposure to the $k$-th latent factor in time period $t$ (or in finance terminology, ``price of risk''). The model are estimated with different tuning parameters $ \nu_{1}$ and $ \nu_{2}$, and use our proposed BIC to select the optimal tuning parameters. The details of the information criteria can be found in equation (\ref{eqn:score}).

The tuning parameter $ \nu_{1}$ governs the sparsity of the coefficient vector $\theta$. The larger $ \nu_{1}$ is, the larger the shrinkage effect on $\theta$. Figure \ref{Figure-solution-pas} illustrate the effect of this shrinkage. With $\nu_2$ fixed, as the value of $ \nu_{1}$ increases, more coefficients in the estimated $\theta$ vector shrink to zero. From a statistical point of view, the ``effective characteristics'' that
can explain stock returns are those  with non-zero coefficient $\theta$ at relatively large values of $ \nu_{1}$.

Table \ref{Table-estimated-rank} reports the relationship
between tuning parameter $ \nu_{2}$ and rank of estimated $\Pi$
at different quantiles. It shows that the tuning parameter $ \nu_{2}$ governs the rank of matrix $\Pi$, and that as $ \nu_{2}$ increases, we penalize
more on the rank of matrix $\Pi$ through its nuclear norm.

\begin{table}[!ht]\centering
	\begin{threeparttable}
		\caption{The estimated rank of $\Pi$.}\label{Table-estimated-rank}
		\begin{tabular}{@{\extracolsep{8pt}} rrrr}
			\hline\hline
			\multicolumn{1}{l}{$\log_{10}( \nu_{2})$} & \multicolumn{1}{l}{$\tau=0.1$} & \multicolumn{1}{l}{$\tau =0.5$} & \multicolumn{1}{l}{$\tau=0.9$}\tabularnewline
			\hline
			-6.0  & 228  & 228  & 228 \tabularnewline
			-5.5  & 228  & 228  & 228 \tabularnewline
			-5.0  & 228  & 228  & 228 \tabularnewline
			-4.5  & 164  & 228  & 168 \tabularnewline
			-4.0  & 1  & 7  & 2 \tabularnewline
			-3.5  & 1  & 1  & 1 \tabularnewline
			-3.0  & 0  & 0  & 0 \tabularnewline
			\hline \hline \\[-1.8ex]
		\end{tabular}
		\begin{tablenotes}
			\small
			\item Note: Estimated under different values of turning parameter $ \nu_{2}$, when $ \nu_{1}=10^{-5}$ is fixed. The results are reported for quantiles 10\%, 50\% and 90\%.
		\end{tablenotes}
	\end{threeparttable}
\end{table}

The left panel of Table \ref{Figure-sparse-beta} reports the estimated coefficients in the sparse part when we fix the tuning parameters at $\log_{10}( \nu_{1})=-3.5$ and $\log_{10}( \nu_{2})=-4$. The signs of some characteristics are the same across the quantiles, e.g. size (mve), book-to-market (bm), momentum (mom1m, mom12m), accurals (acc), book equity growth (egr), leverage (lev), and standardized unexpected earnings (sue). However, some characteristics have heterogenous effects on future returns at different quantiles. For example, at the 10\% quantile, high beta stocks have high future returns, which is consistent with results found via the CAPM; while at $50\%$ and 90\% quantile, high beta stocks have low future returns, which conforms the ``low beta anomaly'' phenomenon. Volatility (measured by both range and idiosyncratic volatility) is positively correlated with future returns at 90\% quantile, but negatively correlated with future returns at 10\% and 50\% percentile. The result suggests that quantile models can capture a wider picture of the heterogenous relationship between asset returns and firm characteristics at different parts of the distribution (\cite{koenker2000galton}).

\begin{table}\centering
	\begin{threeparttable}
		\caption{Sparse Part Coefficients at Different Quantiles.}\label{Figure-sparse-beta}
		\begin{tabular} {@{\extracolsep{8pt}} lrrrrrr}
			\\[-1.8ex]\hline
			\hline \\[-1.8ex]
			& \multicolumn{3}{c}{Fixed $ \nu_{1}$ and $ \nu_{2}$} & \multicolumn{3}{c}{Optimal $ \nu_{1}$ and $ \nu_{2}$ (BIC)}\tabularnewline
			\cline{2-4}  \cline{5-7}
			& $\tau = 0.1$ & $\tau = 0.5$ & $\tau = 0.9$ & $\tau = 0.1$ &$ \tau = 0.5$ & $\tau = 0.9$\tabularnewline
			\hline \\[-1.8ex]
			acc  & -0.089  & -0.074  & -0.041  & 0  & 0  & 0 \tabularnewline
			range  & -2.574  & -0.481  & 2.526  & -2.372  & -0.356  & 2.429 \tabularnewline
			beta  & 0.174  & -0.116  & -0.406  & 0  & 0  & -0.115 \tabularnewline
			bm  & 0.371  & 0.175  & 0.263  & 0  & 0  & 0.168 \tabularnewline
			chinv  & 0  & 0  & -0.152  & 0  & 0  & 0 \tabularnewline
			dy  & -0.086  & 0  & 0.119  & 0  & 0  & 0 \tabularnewline
			egr  & -0.106  & -0.053  & -0.091  & 0  & 0  & 0 \tabularnewline
			ep  & 0.199  & 0.057  & -0.479  & 0  & 0  & -0.391 \tabularnewline
			gma  & 0  & 0.091  & 0.201  & 0  & 0  & 0 \tabularnewline
			idiovol  & -1.229  & -0.071  & 1.438  & -1.055  & 0  & 1.286 \tabularnewline
			ill  & -0.334  & -0.218  & 0  & 0  & 0  & 0 \tabularnewline
			invest  & -0.097  & 0  & 0.146  & 0  & 0  & 0 \tabularnewline
			lev  & 0.183  & 0.063  & 0.129  & 0  & 0  & 0 \tabularnewline
			lgr  & -0.106  & -0.037  & 0  & 0  & 0  & 0 \tabularnewline
			mom12m  & -0.166  & -0.077  & -0.117  & 0  & 0  & 0 \tabularnewline
			mom1m  & -0.150  & -0.384  & -0.571  & 0  & -0.286  & -0.477 \tabularnewline
			mve  & -0.038  & -0.093  & -0.811  & 0  & 0  & -0.667 \tabularnewline
			operprof  & 0  & 0.025  & 0.088  & 0  & 0  & 0 \tabularnewline
			roaq  & 0.221  & 0.242  & -0.147  & 0  & 0  & 0 \tabularnewline
			roeq  & 0.073  & 0.041  & 0  & 0  & 0  & 0 \tabularnewline
			std\_dolvol  & 0  & 0  & -0.039  & 0  & 0  & 0 \tabularnewline
			std\_turn  & 0.310  & 0  & -0.247  & 0  & 0  & 0 \tabularnewline
			sue  & 0.105  & 0.061  & 0.045  & 0  & 0  & 0 \tabularnewline
			turn  & -0.796  & -0.083  & 0.386  & -0.330  & 0  & 0 \tabularnewline
			\hline
			\hline \\[-1.8ex]
		\end{tabular}
		\begin{tablenotes}
			\small
			\item Note:  The left panel reports the estimated coefficient vector $\theta$ in the sparse part for quantiles 10\%, 50\% and 90\%, when the tuning
			parameters are fixed at $\log_{10}( \nu_{1})=-3.5$, $\log_{10}( \nu_{2})=-4$.
			The right panel reports the estimated coefficient vector $\theta$
			under the when the turning parameters are optimal, as selected by
			BIC (indicated in Table \ref{Figure-optimal-hyperparam}).
		\end{tablenotes}
	\end{threeparttable}
\end{table}

Table \ref{Figure-optimal-hyperparam} reports the selected optimal tuning parameters $ \nu_{1}$ and $ \nu_{2}$ for different quantiles. The tuning parameters
are selected via BIC  based on (\ref{eqn:score}) as discussed in Section \ref{sec:Simulation}. For every $ \nu_{1}$
and $ \nu_{2}$, we get the estimates $\widetilde{\theta}( \nu_{1}, \nu_{2})$
and $\widetilde{\Pi}( \nu_{1}, \nu_{2})$  and the number
of factors $r=\text{rank}(\widetilde{\Pi}( \nu_{1}, \nu_{2}))$.
The $\theta$ vector is sparse with non-zero coefficients on selected
characteristics. The 10\% quantile of returns has only 1 latent factor,
and 3 selected characteristics. The median of returns has 7 latent
factors and 2 selected characteristics. The 90\% quantile of returns
has 2 latent factors and 7 selected characteristics.
Range is the only characteristic selected across all 3 quantiles.
Idiosyncratic volatility is selected at 10\% and 90\% quantiles, with
opposite signs. 1-month momentum is selected at 50\% and 90\% percentiles,
with negative sign suggesting reversal in returns.

\begin{table}[!h] 
\renewcommand{\arraystretch}{1.2}
	\begin{threeparttable}
		\caption{Selected Optimal Tuning Parameters and Number of Factors}\label{Figure-optimal-hyperparam}
        \begin{centering}		
		\begin{tabular}{@{\extracolsep{12pt}} lccc}
			\\[-1.8ex]\hline
			\hline \\[-1.8ex]
			& \multicolumn{1}{l}{$\tau=0.1$} & \multicolumn{1}{l}{$\tau=0.5$} & \multicolumn{1}{l}{$\tau=0.9$}\tabularnewline
			optimal $r$ & 1  & 7  & 2 \tabularnewline
			optimal $ \nu_{1}$ & $10^{-2.5}$ & $10^{-2.5}$ & $10^{-2.75}$\tabularnewline
			optimal $ \nu_{2}$ & $10^{-4}$  & $10^{-4}$ & $10^{-4}$ \tabularnewline
			\hline \hline \\[-1.8ex]
		\end{tabular}
		\par\end{centering}
		\begin{tablenotes}
			\small
			\item Note: This table reports the selected optimal tuning parameter $ \nu_{1}$
			and $ \nu_{2}$ that minimize the objective function in equation
			(\ref{eqn:score}) for different quantiles.
			    
		\end{tablenotes}
	\end{threeparttable}
	 
\end{table}

Overall, the empirical evidence suggests that both firm characteristics and latent risk factors have valuable information in explaining stock returns. In addition, we find that the selected characteristics and number of latent factors differ across the quantiles.

\subsubsection*{Interpretation of Latent Factors}

Table \ref{Table-variance-pi} below reports the variance in the matrix $\Pi$ explained by each
Principal Component (PC) or latent factor. At upper and lower quantiles, the first PC dominates. At the median there are more latent factors accounting
for the variations in $\Pi$, with second PC explaining 13.8\% and
third PC explaining 6.8\%.

\begin{table}[!h]\centering
\begin{small}
		\begin{threeparttable}
		\caption{Percentage of $\Pi$ explained by PC}\label{Table-variance-pi}
		\begin{tabular}{@{\extracolsep{8pt}} cccc}
			\\[-1.8ex]\hline
			\hline \\[-1.8ex]
			& \multicolumn{1}{l}{$\tau=0.1$} & \multicolumn{1}{l}{$\tau=0.5$} & \multicolumn{1}{l}{$\tau=0.9$}\tabularnewline
			PC1  & 100.00\%  & 73.82\%  & 99.68\% \tabularnewline
			PC2  &  & 13.71\%  & 0.32\% \tabularnewline
			PC3  &  & 6.78\%  & \tabularnewline
			PC4  &  & 4.12\%  & \tabularnewline
			PC5  &  & 1.11\%  & \tabularnewline
			PC6  &  & 0.45\%  & \tabularnewline
			PC7  &  & 0.01\%  & \tabularnewline
			Total  & 100.00\%  & 100.00\%  & 100.00\% \tabularnewline
			\hline \hline \\[-1.8ex]
		\end{tabular}
		\begin{tablenotes}
			\small
			\item Note: Variance of matrix $\Pi$ explained by each principal
			component for different quantiles.
		\end{tablenotes}
	\end{threeparttable}
\end{small}
\end{table}

We also found the first PC captures the market returns in all three quantiles:
Figure \ref{Figure-PC-against-ret} plots the first principal component against the monthly returns
of S\&P500 index, showing that they have strong positive correlations.

\begin{figure}[!h]
	\caption{The S\&P 500 Index Return and the First PC at Different Quantiles. 
		This figure plots the first PC of matrix $\Pi$ against S\&P500 index monthly return for quantiles 10\% (left), 50\% (middle), and 90\% (right).}\label{Figure-PC-against-ret}
\includegraphics[width=5.5cm,height=4.5cm]{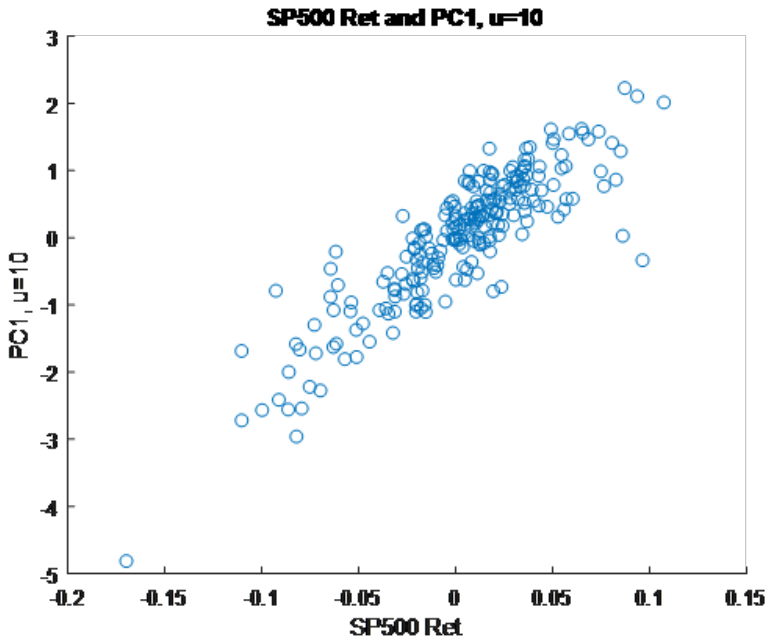}\includegraphics[width=5.5cm,height=4.5cm]{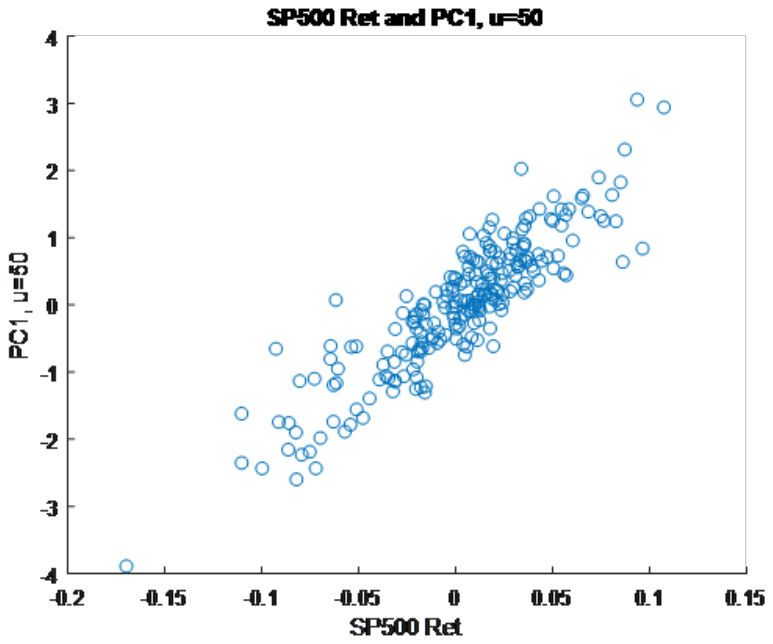}\includegraphics[width=4cm,height=4.5cm]{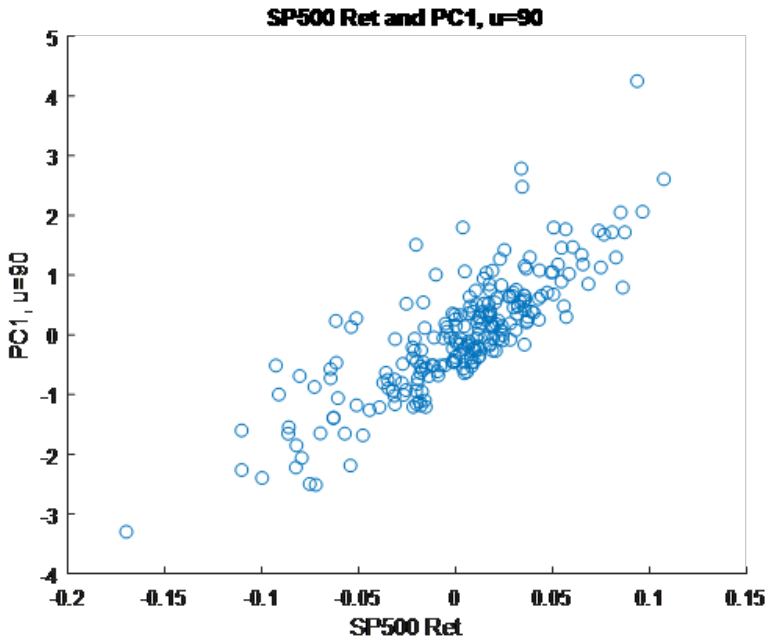}
\end{figure}

\section*{Acknowledgements}
We would like to thank the Editors, the Associate Editors and two anonymous referees for their detailed reviews, which helped to improve the paper substantially. We are also grateful to Victor Chernozhukov, Iv\'{a}n Fern\'{a}ndez-Val, Bryan Graham, Hiroaki Kaido, Anna Mikusheva, Whitney Newey, Eric Renault,  Jeremy Smith, and Vasilis Syrgkanis for helpful discussions.


\bibliographystyle{imsart-nameyear} 
\bibliography{QuantileFactor}       

\newpage



 
\renewcommand{\thepage}{S\arabic{page}}  
\renewcommand{\thesection}{S\arabic{section}}   
\renewcommand{\thetable}{S\arabic{table}}   
\renewcommand{\thefigure}{S\arabic{figure}}
\renewcommand{\thetheorem}{S\arabic{theorem}}
\renewcommand{\thelemma}{S\arabic{lemma}}
\renewcommand{\thedefinition}{S\arabic{definition}}
\renewcommand{\theequation}{S\arabic{equation}}  

\setcounter{section}{0}
\setcounter{figure}{0}
\setcounter{table}{0}
\setcounter{equation}{0}
\setcounter{page}{1}

\title{Supplementary Material for \\``High Dimensional Latent Panel Quantile Regression with an Application to Asset Pricing''}
 
 \begin{center}
     Alexandre Belloni, Mingli Chen, Oscar Hernan Madrid Padilla, Zixuan (Kevin) Wang
\end{center}

\vspace{0.5cm}

\section{Implementation Details of the Proposed ADMM Algorithm} \label{sec:admm}

Denoting  by  $P_{+}(\cdot)$ and $P_{-}(\cdot)$ the element-wise  positive and negative part operators,  the ADMM  proceeds  doing the iterative  updates


\begin{align}
	\label{eqn:iter1}
	& V^{(k+1)}\,\leftarrow\,P_{+}\left(W^{(k)}-U_{V}^{(k)}-\frac{\tau}{nT\eta}\boldsymbol{1}\boldsymbol{1}^{\prime}\right)+P_{-}\left(W^{(k)}-U_{V}^{(k)}-\frac{\tau}{nT\eta}\boldsymbol{1}\boldsymbol{1}^{\prime}\right)\\
	\label{eqn:iter2}
	& \tilde{\theta}^{(k+1)}\,\leftarrow\,  \underset{\theta}{\arg\min} \left\{ \frac{\eta}{2}\sum_{i=1}^{n}\sum_{t=1}^{T}\left(W_{i,t}^{(k)}-Y_{i,t}+X_{i,t}^{\prime}\theta+(Z_{\Pi}^{(k)})_{i,t}+(U_{W}^{(k)})_{i,t}\right)^{2}+\frac{\eta}{2}\|Z_{\theta}^{(k)}-\theta+U_{\theta}^{(k)}\|^{2}\right\} \\
	\label{eqn:iter3}
	& \tilde{\Pi}^{(k+1)}\,\leftarrow\,\underset{  \tilde{ \Pi} }{\arg\min}\left\{ \frac{1}{2}\|Z_{\Pi}^{(k)}-\tilde{\Pi}+U_{\Pi}^{(k)}\|_{F}^{2}+    \frac{ \nu_{2}}{\eta} \|\tilde{\Pi}\|_{*}\right\} \\
 & Z_{\theta}^{(k+1)}\,\leftarrow\, \underset{Z_{\theta}  }{\arg\min} \left\{ \frac{1}{2}\|\tilde{\theta}^{(k+1)}-U_{\theta}^{(k)}-Z_{\theta}\|^{2}+\frac{ \nu_{1}}{\eta}\sum_{j=1}^{p}w_{j}\vert (Z_{\theta})_j \vert\right\} \\
	& (Z_{\Pi}^{(k+1)},W^{(k+1)})\,\leftarrow\,  \underset{Z_{\Pi},W}{\arg\min} \Bigg\{\frac{\eta}{2}\|W-Y+X\tilde{\theta}^{(k+1)}+Z_{\Pi}+U_{W}^{(k)}\|_{F}^{2}+\frac{\eta}{2}\|V^{(k+1)}-W+U_{V}^{(k)}\|_{F}^{2} \\
	\label{eqn:iter6} &\,\,\,\,\,\,\,\,\,\,\,\,\,\,\,\,\,\,\,\,\,\,\,\,\,\,\,\,\,\,\,\,\,\,\,\,\,\,\,\,\,\,\,\,\,\,\,\,\,\,\,\,\,\,\,\,+\frac{\eta}{2}\|Z_{\Pi}-\tilde{\Pi}^{(k+1)}+U_{\Pi}^{(k)}\|_{F}^{2}\Bigg\}
\end{align}

\begin{align*}
	U_{ V }^{(k+1)} \,\leftarrow\,V^{(k+1)}- W^{(k+1)}+ U_V^{(k)},\,\,\,\,\,U_W^{(k+1)}   \,\leftarrow\,W^{(k+1)}- Y+ X\tilde{\theta}^{(k+1)} + Z_{\Pi}^{(k+1)} +U_W^{(k)},\,\,\,\,\,\\
	U_{\Pi}^{(k+1)} \,\leftarrow\,Z_{\Pi}^{(k+1)} - \tilde{ \Pi}^{(k+1)} + U_{\Pi }^{(k)},\,\,\,\,U_{\theta}^{(k+1)} \,\leftarrow\, Z_{\theta}^{(k+1)}
	- \tilde{\theta}^{(k+1)}  +U_{\theta}^{(k)},
\end{align*}
where  $\eta>0$ is the penalty,  see \cite{boyd2011distributed}.

The update  for $\tilde{\theta}$ is
\[
\displaystyle  \tilde{\theta}^{(k+1)}  \,\leftarrow \, \left[ \sum_{i=1}^n\sum_{t=1}^T     X_{i,t} X_{i,t}^{\prime}   +I_p  \right]^{-1} \left[  -  \sum_{ i= 1}^n \sum_{t=1}^T X_{i,t}A_{i,t}  +Z_{\theta}^{(k)}   + U_{\theta}^{(k)}   \right],
\]
where
\[
A  :=   W^{(k)}+ Z_{\Pi}^{(k)}    + U_W^{(k)}   - Y.
\]
The update for   $\tilde{\Pi}$  is
\[
\tilde{\Pi}^{(k+1)} \,\leftarrow\,  P \text{ diag}\left(  \max\left\{ 0,  v_j -  \frac{ \nu_2}{\eta}  \right\}_{  1\leq j \leq l} \right) Q^{\prime},
\]
where
\[
Z_{\Pi}^{(k)} +  U_{  \Pi }^{(k)}  \,=\,   P  \text{ diag}( \{ v_j \}_{1\leq j \leq l}  ) Q^{\prime}.
\]
Furthermore,  for  $Z_{\theta}$,
\[
Z_{\theta,j}^{(k+1)} \,\leftarrow \,   \text{sign}(\tilde{\theta}_j^{(k+1)} -   U_{\theta,j}^{(k)} ) \left[  \vert  \tilde{\theta}_j^{(k+1)}  - U_{\theta,j}^{(k)} \vert    -  \frac{ \nu_{1}  w_j }{\eta} \right].
\]
Finally, defining
\[
\tilde{A} \,=\,  - Y  +  X\tilde{\theta}^{(k+1)}  + U_W^{(k)}, \,\,\,\, \tilde{B} \,=\,-V^{(k+1)}-U_V^{(k)},\,\,\,\, \tilde{C} \,=\, -\tilde{ \Pi}^{(k+1)} + U_{\Pi}^{(k)},
\]
the remaining updates are
\[
Z_{\Pi}^{(k+1)} \,\leftarrow\,   \frac{    -\tilde{A}    -  2\tilde{C}    + \tilde{B}   }{3},
\]
and
\[
W^{(k+1)}  \,\leftarrow\,   -\tilde{A} -\tilde{C} - 2Z_{\Pi}^{(k+1)}.
\]

\subsection{Estimation without Covariates}
\label{sec:no_covariates}
Note, when there are no covariates, our proposed ADMM can be simplified. In this case, we face the following problem
\begin{equation}
	\label{eqn:unsupervised}
	\displaystyle  \underset{ \tilde{\Pi} \in \mathbb{R}^{ n \times T }  }{\min  }\,\,\left\{  \frac{1}{n T} \sum_{i=1}^{n} \sum_{j=1}^{T}  \rho_{\tau}( Y_{i,t}  - \tilde{ \Pi}_{i,t} )     +   \nu_2 \|\tilde{ \Pi}\|_*  \right\}.
\end{equation}
This can be thought  as a convex relaxation  of the estimator  studied in \cite{chen2018quantile}. Problem (\ref{eqn:unsupervised}) is also related to the setting of  robust estimation of a latent low-rank matrix, e.g. \cite{elsener2018robust}. However, our approach can also be used to estimate different  quantile levels. As for solving  (\ref{eqn:unsupervised}),  we can proceed by doing the iterative updates
\begin{equation}
	\label{eqn:update1}
	\displaystyle \tilde{\Pi}^{(k+1)} \,\leftarrow\,   \underset{ \tilde{ \Pi}  }{\arg \min}\,\left\{  \frac{1}{nT}\sum_{i=1}^n\sum_{t=1}^{T}  \rho_{\tau}( Y_{i,t}  - \tilde{ \Pi}_{i,t} )   \,+\,\frac{\eta}{2}\|  \tilde{ \Pi} - Z_{\Pi}^{(k)}  +  U_{ \Pi  }^{(k)} \|_F^2  \right\},	
\end{equation}
\begin{equation}
	\label{eqn:update2}
	Z_{\Pi}^{(k+1)}   \,\leftarrow  \,   \underset{ Z_{\Pi}  }{\arg \min}\,\left\{ \frac{\eta}{2} \| \tilde{\Pi}^{(k+1)}  - Z_{\Pi} + U_{\Pi}^{(k)} \|_F^2 +  \nu_2\|Z_{\Pi}\|_*   \right\},
\end{equation}
and
\begin{equation}
	\label{eqn:update3}
	U_{\Pi}^{(k+1)}   \,\leftarrow\, \Pi^{(k+1)}  - Z_{\Pi}^{(k+1)}  + U_{\Pi}^{(k)},
\end{equation}
where  $\eta> 0 $ is  the penalty parameter (\cite{boyd2011distributed}). The minimization in (\ref{eqn:update1})  is  similar to (\ref{eqn:update_v}), whereas (\ref{eqn:update2}) can be done similarly as in (\ref{eqn:update_pi}).



\subsection{Estimation in unbalanced designs}\label{sec:admm2}

We  now explore the setting of unbalanced designs, specifically, instead of fully observing $(X,Y)$, we now assume that we only observe $\{ (X_{i,t},Y_{i,t}) \}_{(i,t) \in \mathcal{I}}$ for a set $\mathcal{I}\subset  [n] \times [T]$. For instance, if $n=2$ and $T=3$ but for $i=1$ the data  at time $t=2$ is not available, then $\mathcal{I}$ would be $\{ (1,1),(1,3),(2,1),(2,2),(2,3)\}$.

For this setting, inspired by (\ref{eqn:formulation2}), we formulate the problem 
\begin{equation}
	\label{eqn:formulation4}
	\underset{ \tilde{\theta}\in \mathbb{R}^{p } ,\,\  \tilde{\Pi}\in \mathbb{R}^{n \times T}  }{\min} \left\{\frac{1}{ \vert \mathcal{I}\vert  } \sum_{(i,t)\in \mathcal{I} } \rho_{\tau}(Y_{i,t}   - X_{i,t}^{\prime}\tilde{\theta}  -  \tilde{\Pi}_{i,t} ) +  \nu_{1}\sum_{j=1}^{p}  w_j  \vert   \tilde{\theta}_j  \vert      + \nu_2  \|\tilde{\Pi}\|_*\right\}  
\end{equation}
where  $ \nu_{1} >0$  and  $\nu_2  >0$ are tuning parameters, and  $w_1,\ldots,w_p$ are user  specified weights. Thus, comparing with (\ref{eqn:formulation2}), we now only apply the loss function to the indices for which there is data available. Following (\ref{eqn:formulation3}), we write (\ref{eqn:formulation4}) as

\begin{equation}
	\label{eqn:formulation6}
	\begin{array}{lll}
		\underset{  \stackrel{\tilde{\theta} , \tilde{\Pi},V }{Z_{\theta},Z_{\Pi },W }   }{\min} &  \displaystyle \frac{1}{ \vert \mathcal{I}\vert  } \sum_{(i,t)\in \mathcal{I} }\rho_{\tau}(V_{i,t}) \,+\,   \nu_{1}   \sum_{j=1}^{p}  w_j \vert  Z_{\theta_j} \vert \,+\,  \nu_2 \| \tilde{\Pi} \|_*   \\
		\text{subject to} &    V =  W,\,\,\,\,\,  W= Y-  X\tilde{\theta} - Z_{\Pi },\\
		& \,\,\,\,\, Z_{\Pi }   - \tilde{\Pi} =0,\,\,\,\,\, Z_{\theta}  -\tilde{\theta}   =0.  \\
	\end{array}
\end{equation}
Then, as in Section \ref{sec:admm}, we obtain the ADMM updates given by 
\[
V_{i,t}^{(k+1)}\,\leftarrow\,P_{+}\left(W_{i,t}^{(k)}-(U_{V})_{i,t}^{(k)}-\frac{\tau}{\vert \mathcal{I}\vert  \eta}\right)+P_{-}\left(W_{i,t}^{(k)}-(U_{V})_{i,t}^{(k)}-\frac{\tau}{  \vert \mathcal{I}\vert   \eta}\right)\\
\label{eqn:iter4}
\]
for $(i,t) \in \mathcal{I}$,  and 
\[
V_{i,t}^{(k+1)}\,\leftarrow\,  W_{i,t} -  (U_{V})_{i,t}^{(k)}
\]
for $(i,t) \notin \mathcal{I}$, and with the rest of updates given exactly as in  Section \ref{sec:admm}.

 	Notice that our formulation is similar in spirit to \cite{athey2018matrix}. In particular, the objective function (\ref{eqn:formulation4})   resembles Equation (4.3) in  \cite{athey2018matrix}, with the main difference that we allow for covariates and work with the quantile loss.
 	
\section{Proof of Theorem \ref{thm:1} }

\subsection{Auxiliary lemmas  for  proof of Theorem \ref{thm:1}  }
\label{sec:proofs}


Throughout, we use the notation
\[
Q_{\tau}(\tilde{\theta},\tilde{\Pi}) \,=\,  \mathbb{E}(\hat{Q}_{\tau}(\tilde{\theta},\tilde{\Pi})  ).
\]
Moreover, as in \cite{yu1994rates},  we  define the  sequence $\{(\tilde{Y}_{i,t}, \tilde{X}_{i,t}  ) \}_{i \in [n], t \in [T]}$
such  that  
\begin{itemize}
	\item   $\{(\widetilde{Y}_{i,t}, \widetilde{X}_{i,t}  ) \}_{i \in [n], t \in [T]}$  is independent  of   $\{(Y_{i,t},X_{i,t}  ) \}_{i \in [n], t \in [T]}$; 
	\item  for a fixed  $t$ the random vectors  $\{(\widetilde{Y}_{i,t}, \widetilde{X}_{i,t}  ) \}_{i \in [n]}$  are independent;
	\item  for a fixed $i$:
	\[
	\mathcal{L}(  \{ (\widetilde{Y}_{i,t},\widetilde{X}_{i,t}) \}_{ t \in H_l   } ) \,=\,  \mathcal{L}(  \{ (Y_{i,t},X_{i,t}) \}_{ t \in H_l   } ) \,=\, \mathcal{L}(  \{ (Y_{i,t},X_{i,t}) \}_{ t \in H_1  } ) \,\,\,\,\forall  l \,\,\in [d_T],
	\]
	and  the blocks $ \{ (\widetilde{Y}_{i,t},\widetilde{X}_{i,t}) \}_{ t \in H_1   },\ldots,  \{ (\widetilde{Y}_{i,t},\widetilde{X}_{i,t}) \}_{ t \in H_{d_T}   }$ are independent.
\end{itemize}

Here,  we  define  $ \Lambda  \,:=\,   \{ H_1,H_1^{\prime},\ldots,H_{d_T},H_{d_T}^{\prime},R  \}$ with
\begin{equation}
	\label{eqn:intervals}
	\begin{array}{lll}
		H_j &=&  \left\{   t \,:\,    1+ 2(j-1)c_T \leq t \leq (2j-1)c_T \right\},\,\,\,	 \\
		H_j^{\prime} &=&  \left\{   t \,:\,    1+ (2j-1)c_T \leq t \leq 2j c_T \right\}, \,\,j = 1,\ldots, d_T,\\
		&&  \,\,\,\,\text{and}\,\,\, R \,=\,\{ t \,:\,    2c_T d_T +1\,\leq t \leq T  \}.
	\end{array}
\end{equation}
We also  use  the symbol  $\mathcal{L}(\cdot)$ to denote the distribution  of a sequence of random variables.

Next, define the scores $a_{i,t} \,=\,  \tau -  1\{ Y_{i,t}  \leq  X_{i,t}^{\prime} \theta(\tau) + \Pi_{i,t}(\tau)  \} $, and   $\tilde{a}_{i,t} \,=\,  \tau -  1\{ \widetilde{Y}_{i,t}  \leq  \widetilde{X}_{i,t}^{\prime} \theta(\tau) + \Pi_{i,t}(\tau)  \} $.


We start by controlling an empirical process involving  the scores $\{a_{i,t}\}$. This is given next.

\begin{lemma}
	\label{lem:penalty_tuning}
	Under Assumptions \ref{cond1}--\ref{cond3}, we have
	\[
	\mathbb{P}\left(   \underset{j = 1,\ldots, p}{\max }   \,\, \frac{1}{n T} \left\vert   \sum_{i=1}^n \sum_{t=1}^T   \frac{ X_{i,t,j}  a_{i,t}  }{ \hat{\sigma}_j  }    \right\vert  \,\geq \, 9 \sqrt{       \frac{  c_T  \log( \max\{n,p c_T\} )  }{n d_T}  }      \right) \,\leq \,    \frac{16}{ n } + 8npT\left( \frac{1}{c_T} \right)^{\mu} .
	\]
\end{lemma}
\begin{proof}
	Notice  that
	\begin{equation}
		\label{eqn:s0}
		\begin{array}{l}
			\displaystyle     \mathbb{P}\left(    \underset{j = 1,\ldots, p}{\max }   \,\, \frac{1}{n T} \left\vert   \sum_{i=1}^n \sum_{t=1}^T   \frac{ X_{i,t,j}  a_{i,t}  }{ \hat{\sigma}_j  }    \right\vert  \,\geq \, \eta    |  X\right) \\
			\leq      \displaystyle  2p  \underset{j = 1,\ldots, p}{\max }\,\mathbb{P}\left(   \, \frac{1}{n d_T}\left\vert   \sum_{i=1}^n \sum_{l=1}^{d_T}   \left(  \frac{1}{c_T}  \sum_{ t = 2(l-1)+1  }^{ 2(l-1) + c_T }  \frac{ X_{i,t,j}  a_{it}  }{ \hat{\sigma}_j  }    \right)    \right\vert  \,\geq \, \frac{\eta}{9}    |  X       \right) \,+\,\\
			\displaystyle  p \underset{j = 1,\ldots, p}{\max } \, \mathbb{P}\left(   \,    \, \frac{1}{n d_T}\left\vert   \sum_{i=1}^n   \frac{1}{c_T} \left(\sum_{t \in R} \frac{ X_{i,t,j}   }{\hat{\sigma}_j  }   \right)         \right\vert  \,\geq \, \frac{\eta}{9}    |  X       \right)\\
			\leq      \displaystyle  4p  \,\underset{j = 1,\ldots, p}{\max }\,\mathbb{P}\left(   \, \underset{m = 1,\ldots, c_T}{\max }\, \frac{1}{n d_T}\left\vert   \sum_{i=1}^n \sum_{l=0}^{d_T-1}        \frac{ X_{i,\, (2lc_T+ m  ) ,\,j}  \tilde{a}_{i,\, (2lc_T+ m  )}  }{ \hat{\sigma}_j  }     \right\vert  \,\geq \, \frac{\eta}{9}    |  X       \right) \,+\,\\
			\displaystyle  2p  \underset{j = 1,\ldots, p}{\max } \,  \mathbb{P}\left(  \, \underset{m = 1,\ldots, \vert R\vert  }{\max }    \, \frac{1}{n d_T}\left\vert   \sum_{i=1}^n       \frac{ X_{i, \,(    2d_T c_T  + m  ) \, j}  \tilde{a}_{i,\, (2d_T c_T+ m  )}  }{ \hat{\sigma}_j  }      \right\vert  \,\geq \, \frac{\eta}{9}    |  X       \right)  \,+\,   8npT\left( \frac{1}{c_T} \right)^{\mu}  \\
		\end{array}
	\end{equation}
	where the first inequality  follows from union bound, and the second by Lemmas  4.1 and  4.2  from \cite{yu1994rates}.  Hence,
		\begin{equation}
		\label{eqn:s1}
		\begin{array}{l}
			\displaystyle     \mathbb{P}\left(    \underset{j = 1,\ldots, p}{\max }   \,\, \frac{1}{n T} \left\vert   \sum_{i=1}^n \sum_{t=1}^T   \frac{ X_{i,t,j}  a_{i,t}  }{ \hat{\sigma}_j  }    \right\vert  \,\geq \, \eta    |  X\right) \\
						\leq 	      \displaystyle  4p c_T \,\underset{j \in [p], m \in [c_T]}{\max }\,\mathbb{P}\left(   \,  \frac{1}{n d_T}\left\vert   \sum_{i=1}^n \sum_{l=0}^{d_T-1}        \frac{ X_{i,\, (2lc_T+ m  ), \,j}  \tilde{a}_{i,\, (2lc_T+ m  )}  }{ \hat{\sigma}_j  }     \right\vert  \,\geq \, \frac{\eta}{9}    |  X       \right) \,+\,\\
			\displaystyle  2p c_T \,\underset{j \in [p], m \in [ \vert R \vert  ]}{\max }\,\mathbb{P} \left(     \, \frac{1}{n d_T}\left\vert   \sum_{i=1}^n       \frac{ X_{i, \,(    2d_T c_T  + m  ), \, j}  \tilde{a}_{i,\, (2d_Tc_T+ m  )}  }{ \hat{\sigma}_j  }      \right\vert  \,\geq \, \frac{\eta}{9}    |  X       \right)  \,+\,   8npT\left( \frac{1}{c_T} \right)^{\mu}  \\
		\end{array}
	\end{equation}
	
	Therefore, since
	\[
	\displaystyle   \frac{1}{n d_T }  \sum_{i=1}^n \sum_{l=0}^{d_T-1}        X_{i,\, (2lc_T+ m  )\,, j}^2\,\leq \,   3 c_T  \hat{\sigma}_j^2,
	\]
	and with a similar argument  for the second  term in the last inequality of (\ref{eqn:s1}), we obtain the result  by Hoeffding's inequality and integrating over $X$.
\end{proof}

Next we proceed to control the complexity of the set $\{   \Delta  \in \mathbb{R}^{ n \times T }  \,:\,    \| \Delta \|_*  \leq 1\,\,\, \}$ in terms of the scores $\{a_{i,t}\}$.

\begin{lemma}
	\label{lem:latent_error}
	Supposes that Assumptions \ref{cond1}--\ref{cond3} hold, and  let
	\begin{equation}
		\label{eqn:set_g}
		\begin{array}{l}
			\mathcal{G}   \,=\, \Bigg\{   \Delta  \in \mathbb{R}^{ n \times T }  \,:\,    \| \Delta \|_*  \leq 1\,\,\, \Bigg\}.    	
		\end{array}
	\end{equation}
	Then there  exists positive constants $c_1$ and $c_2$ such that
	\[
	\begin{array}{lll}
		\displaystyle 	 \underset{ \Delta \in \mathcal{G}  }{\sup}\,\,  \frac{1}{n T}   \left\vert  \sum_{i=1}^n \sum_{t=1}^T   \Delta_{i,t} a_{i,t}  \right\vert &\leq & \displaystyle  \frac{100 c_T }{n T } \left( \sqrt{n}  + \sqrt{d_T}   \right) ,\\
	\end{array}
	\]
	with probability  at least
	\[
	1 \,-\, 2 n T  \left( \frac{1}{c_T} \right)^{\mu} \,-\, 2c_1\exp(-c_2 \max\{n,T\}    +  \log c_T ),
	\]
	for some positive constants $c_1$ and $c_2$.
\end{lemma}

\begin{proof}
	
	Notice that by Lemma  4.3 from \cite{yu1994rates},
	
	\begin{equation}
		\label{eqn:s02}
		\begin{array}{l}
			\displaystyle  \mathbb{P}\left(  \underset{  \Delta \in \mathcal{G}}{\sup } \left[ \frac{1}{n T   } \left\vert   \sum_{i=1}^n \sum_{l=1}^{d_T}   \sum_{t \in H_l}  \Delta_{i ,t}   a_{i,t}             \right\vert   \,+\,   \frac{1}{n T   } \left\vert   \sum_{i=1}^n \sum_{l=1}^{d_T}  \sum_{t \in H_l^{\prime}  }\Delta_{i, t }   a_{i,t}             \right\vert  \,+\,   \frac{1}{n T   }\left\vert   \sum_{i=1}^n \sum_{t \in R  }\Delta_{i, t }   a_{i,t}             \right\vert \right]   \geq \eta   \right) \\
			\displaystyle  \, \leq \, \mathbb{P}\left(  \underset{  \Delta \in \mathcal{G}}{\sup }\, \frac{1}{n T   } \left\vert   \sum_{i=1}^n \sum_{l=1}^{d_T}   \sum_{t \in H_l}  \Delta_{i, t}   \tilde{a}_{i,t}             \right\vert   \geq   \frac{\eta}{3} \right) \,+\,\mathbb{P}\left(  \underset{  \Delta \in \mathcal{G}}{\sup }  \, \frac{1}{n T   } \left\vert   \sum_{i=1}^n \sum_{l=1}^{d_T}   \sum_{t \in H_l^{\prime}  }  \Delta_{i, t^{\prime}   }   \tilde{a}_{i,t}             \right\vert   \geq   \frac{\eta}{3} \right) \,+\, \\
			\displaystyle 	\mathbb{P}\left(     \underset{  \Delta \in \mathcal{G}}{\sup }   \,\frac{1}{n T   } \left\vert   \sum_{i=1}^n   \sum_{t \in R}  \Delta_{i ,t}   \tilde{a}_{i,t}             \right\vert   \geq   \frac{\eta}{3} \right)  \,+\, 2 n T  \left( \frac{1}{c_T} \right)^{\mu}\\
		\end{array}
	\end{equation}
And so
	\begin{equation}
	\label{eqn:s2}
	\begin{array}{l}
		\displaystyle  \mathbb{P}\left(  \underset{  \Delta \in \mathcal{G}}{\sup } \left[ \frac{1}{n T   } \left\vert   \sum_{i=1}^n \sum_{l=1}^{d_T}   \sum_{t \in H_l}  \Delta_{i ,t}   a_{i,t}             \right\vert   \,+\,   \frac{1}{n T   } \left\vert   \sum_{i=1}^n \sum_{l=1}^{d_T}  \sum_{t \in H_l^{\prime}  }\Delta_{i, t }   a_{i,t}             \right\vert  \,+\,   \frac{1}{n T   }\left\vert   \sum_{i=1}^n \sum_{t \in R  }\Delta_{i, t }   a_{i,t}             \right\vert \right]   \geq \eta   \right) \\
		\leq  \displaystyle 2 c_T  \underset{m \in  [c_T]  }{\max }  \mathbb{P}\left(    \underset{  \Delta \in \mathcal{G}}{\sup }\, \frac{1}{n d_T   } \left\vert   \sum_{ i= 1}^n \sum_{l=0   }^{d_T-1} \Delta_{i ,\, ( 2c_T l + m  ) }  \tilde{a}_{i, \, ( 2c_T l + m  )} \right \vert \,\geq \,\frac{\eta}{9}      \right) \,+\,\\
		\displaystyle  \,\,  c_T  \underset{m \in  [\vert R \vert ]  }{\max }  \mathbb{P}\left(    \underset{  \Delta \in \mathcal{G}}{\sup }\, \frac{1}{n d_T   } \left\vert   \sum_{ i= 1}^n  \Delta_{i, \, ( 2c_T d_T + m  ) }  \tilde{a}_{i, \, ( 2c_T d_T + m  )} \right \vert \,\geq \,\frac{\eta}{9}      \right)\,+\, 2 n T  \left( \frac{1}{c_T} \right)^{\mu}.\\
	\end{array}
\end{equation}

	We now proceed to bound each of the  terms in the upper bound of (\ref{eqn:s2}). For the first term, notice  that for a fixed $m$
	\begin{equation}
		\label{eqn:s3}
		\begin{array}{lll}
			\displaystyle 	 \underset{  \Delta \in \mathcal{G}}{\sup } \,\,  \frac{1}{n d_T   } \left\vert   \sum_{ i= 1}^n \sum_{l=0   }^{d_T-1} \Delta_{i, \, ( 2c_T l + m  ) }  \tilde{a}_{i, \, ( 2c_T l + m  )}              \right\vert    & \leq&  \displaystyle \underset{  \Delta \in \mathcal{G}}{\sup }\,\, \frac{1}{n  d_T }    \left\|  \left\{   \tilde{a}_{i ,\, ( 2c_T l + m  )}    \right\}_{ i \in  [n],   l \in [d_T-1]    }   \right \|_2  \left\| \Delta  \right\|_*\\
			& \leq &  \displaystyle\frac{3 }{n d_T } \left( \sqrt{n}  + \sqrt{d_T}   \right),
		\end{array}
	\end{equation}
	where the first inequality  holds by the duality between the  nuclear norm and spectral norm, and the second inequality  happens  with probability  at least  $1-c_1 \exp\left( - c_2  \max\{ n,d_T \} \right)$  by Theorem  3.4 from \cite{chatterjee2015matrix}.

	On the other hand,
	\begin{equation}
		\label{eqn:s5}
		\begin{array}{lll}
			\displaystyle 	\underset{  \Delta \in \mathcal{G}}{\sup } \,  \frac{1}{n d_T   } \left\vert   \sum_{ i= 1}^n  \Delta_{i, \, ( 2c_T d_T + m  ) }  \tilde{a}_{i ,\, ( 2c_T d_T + m  )} \right \vert
			& \leq & \displaystyle    \underset{  \Delta \in \mathcal{G}}{\sup } \,\frac{    \sqrt{n}  \| \{ \Delta_{i ,\, ( 2c_T d_T + m  ) }  \}_{i\in [n]} \|   }{n  d_T }   \\
			& \leq & \displaystyle\frac{ \sqrt{n}  \|  \Delta  \|_*   }{n  d_T  }.  \\
		\end{array}
	\end{equation}
	
	The claim follows by combining (\ref{eqn:s2}), (\ref{eqn:s3}),  and (\ref{eqn:s5}), taking  $\eta\,=\,30(  \sqrt{n }+ \sqrt{d_T} )/\sqrt{n  d_T }$, and the fact that $c_T/T  \leq  1/3$.

\end{proof}

We will proceed to exploit Lemmas \ref{lem:penalty_tuning} and \ref{lem:latent_error}  to show that  $(\hat{\beta}(\tau) -\beta(\tau),   \hat{\Pi}(\tau) -\Pi(\tau) )$  belongs to the restricted set with probability approaching one. Before that, we recall an important property relating the nuclear norm to the rank of a matrix.

\begin{lemma}
	\label{lem:nuclear_norm}
	For  every  $\tilde{\Pi}, \check{\Pi}  \in  \mathbb{R}^{n \times  T}$, we have that
	\[
	\| \check{\Pi} - \tilde{\Pi}\|_*   +    \| \check{\Pi} \|_*  - \| \tilde{\Pi} \|_*   \,\leq \,  (6\sqrt{ \mathrm{rank}(\check{\Pi})  } +1)\| \check{\Pi}  - \tilde{\Pi}  \|_F
	\]
\end{lemma}
\begin{proof}
	This follows directly  from  Lemma  2.3 in \cite{elsener2018robust}.
\end{proof}

\begin{lemma}
	\label{lem1}Assume  that \ref{cond1}--\ref{cond3} hold. Then, with probability approaching one,
	\begin{equation}
		\label{lem1:eq1}
		\frac{3}{4}\|\theta   \|_1 \,\leq \,  \|\theta\|_{1,n,T}  \,\leq \, \frac{5}{4}\|\theta   \|_1,
	\end{equation}
	for  all $\theta \in \mathbb{R}^p$.

	Moreover,  for  $c_0  \in  (0,1)$ letting
	$$
	\nu_{1} \,=\,\frac{9}{1-c_0} \sqrt{       \frac{   c_T \log( \max\{n,pc_T\} )  }{   n d_T}  }  ( \sqrt{n}  +  \sqrt{d_T} )   ,
	$$
	and 
	$$\nu_2 \,=\,  \frac{200 c_T}{n T}\left( \sqrt{n}  +  \sqrt{d_T}  \right),$$
	we have that
	\[
	( \hat{\theta}(\tau) - \theta(\tau),  \hat{\Pi}(\tau) -\Pi(\tau)  ) \in A_{\tau},
	\]
	with probability approaching one, where
	\[
	\begin{array}{l}
		A_{\tau}  = \Bigg\{ (\delta, \Delta ) \,:\,    \|\delta_{ T_{\tau}^c }\|_1  + \frac{ \| \Delta \|_* }{\sqrt{nT}  \sqrt{ \log(  \max\{n,p c_T\} )} }  \leq    C_0 \left(  \|\delta_{T_{\tau}}\|_1  +  \frac{ \sqrt{r_{\tau}}\|   \Delta\|_F }{\sqrt{nT}  \sqrt{ \log(  \max\{n,p c_T\} )} } \right)   \Bigg\},\\
	\end{array}
	\]
	and $C_0$ is a positive constant that depends  on $\tau$ and $c_0$.
\end{lemma}
\begin{proof}
We notice that 
	 \begin{equation}
	 	\label{eqn:e1}
	 		\def\arraystretch{2.3}
	 	\begin{array}{lll}
	 		0 & \leq &  \displaystyle  \hat{Q}(\theta(\tau),\Pi(\tau))  \,-\,  \hat{Q}(\hat{\theta}(\tau),\hat{\Pi}(\tau))   \,+\, \nu_{1}\left(\|\theta(\tau)\|_{1,n}  - \|\hat{\theta}(\tau)\|_{1,n,T}    \right) \,+\, \nu_2 ( \|  \Pi(\tau)  \|_*   -  \| \hat{\Pi}(\tau) \|_*  )     \\
	 		& \leq &  \displaystyle \underset{1 \leq j \leq p}{\max}\left \vert  \frac{1}{nT}  \sum_{i=1}^{n} \sum_{t=1}^T  \frac{X_{i,t,j}  a_{i,t} }{\hat{\sigma}_j }    \right\vert \left[ \sum_{k=1}^{p}\hat{\sigma}_k\vert  \theta_k(\tau) - \hat{\theta}_k(\tau)  \vert   \right] \,+\, \nu_{1}\left(\|\theta(\tau)\|_{1,n,T}  - \|\hat{\theta}(\tau)\|_{1,n,T}    \right)  \,+\,\\
	 		& & \displaystyle \left\vert   \frac{1}{nT}   \sum_{i=1}^{n}\sum_{t=1}^{T}  a_{i,t} (\Pi_{i,t}(\tau)- \hat{\Pi}_{i,t}(\tau))  \right\vert   \,+\, \nu_2 ( \|  \Pi(\tau)  \|_*   -  \| \hat{\Pi}(\tau) \|_*  )    \\
	 		& \leq &\displaystyle  9 \sqrt{       \frac{   c_T \log( \max\{n,p c_T\} )  }{n d_T}  } \left[ \sum_{k=1}^{p}\hat{\sigma}_k\vert  \theta_k(\tau) - \hat{\theta}_k(\tau)  \vert   \right]\,+\,   \nu_{1} \left(\|\theta(\tau)\|_{1,n,T}  - \|\hat{\theta}(\tau)\|_{1,n,T}    \right) \\
	 		& &\displaystyle  \,+\,    \|  \Pi(\tau)  -  \hat{\Pi}(\tau) \|_{*} \left(  \underset{  \| \tilde{\Delta} \|_* \leq 1  }{\sup}\, \left\vert   \frac{1}{nT}   \sum_{i=1}^{n}\sum_{t=1}^{T}  a_{i,t}  \tilde{\Delta}_{i,t} \right\vert        \right)  +  \nu_2 ( \|  \Pi(\tau)  \|_*   -  \| \hat{\Pi}(\tau) \|_*  )      \\
	 		& \leq &\displaystyle  9 \sqrt{       \frac{   c_T \log( \max\{n,p c_T\} )  }{n d_T}  } \left[ \sum_{k=1}^{p}\hat{\sigma}_k\vert  \theta_k(\tau) - \hat{\theta}_k(\tau)  \vert   \right]\,+\,   \nu_{1} \left(\|\theta(\tau)\|_{1,n,T}  - \|\hat{\theta}(\tau)\|_{1,n,T}    \right) \\
	 		& &\displaystyle  \,+\, \frac{200 c_T }{n T } \left( \sqrt{n}  + \sqrt{d_T}   \right)\|  \Pi(\tau)  -  \hat{\Pi}(\tau)  \|_*  + \frac{200 c_T }{n T } \left( \sqrt{n}  + \sqrt{d_T}   \right)\left( \| \Pi(\tau)  \|_*    -   \| \hat{\Pi}(\tau)  \|_*  \right),\\
	 		&  & \displaystyle \,-\, \,\, \frac{100 c_T }{n T } \left( \sqrt{n}  + \sqrt{d_T}   \right)\|  \Pi(\tau) -  \hat{\Pi}(\tau)  \|_*\\
	 		& \leq& \displaystyle  9 \sqrt{       \frac{   c_T \log( \max\{n,p c_T\} )  }{n d_T}  } \left[ \sum_{k=1}^{p}\hat{\sigma}_k\vert  \theta_k(\tau) - \hat{\theta}_k(\tau)  \vert   \right]\,+\,   \nu_{1} \left(\|\theta(\tau)\|_{1,n,T}  - \|\hat{\theta}(\tau)\|_{1,n,T}    \right) \\
	 		& &\displaystyle  \,+\, \frac{1200 c_T  \sqrt{r_{\tau}}  + 200c_T}{n T } \left( \sqrt{n}  + \sqrt{d_T}   \right)\|  \Pi(\tau)  -  \hat{\Pi}(\tau)  \|_F\\
	 		&  & \displaystyle \,-\,  \frac{100 c_T }{n T } \left( \sqrt{n}  + \sqrt{d_T}   \right)\|  \Pi(\tau) -  \hat{\Pi}(\tau)  \|_*\\
	 	\end{array}
	 \end{equation} 
	with probability at least
	\[
	1 -  \gamma - \frac{16}{ n } - 8npT\left( \frac{1}{c_T} \right)^{\mu}- 2 n T  \left( \frac{1}{c_T} \right)^{\mu} -2c_1\exp(-c_2 \max\{n,T\} + \log c_T ),
	\]
	where  in  (\ref{eqn:e1}) the first inequality follows from optimality of the estimator,  the second from basic properties of the function $\rho_{\tau}$,  the third from the  duality between the spectral  and  nuclear norms and Lemma \ref{lem:penalty_tuning}, the fourth by Lemma \ref{lem:latent_error}, and  the fifth by Lemma \ref{lem:nuclear_norm}.
	
	
	Therefore, with probability  approaching one, for positive constants $C_1$ and $C_2$, we have
	\[
	\begin{array}{lll}
		0 & \leq &\displaystyle \left[ \sum_{j=1}^{p}\left(  (1-c_0) \hat{\sigma}_j \vert  \hat{\theta}_j(\tau) - \theta_j(\tau)    \vert     + \hat{\sigma}_j \vert   \theta_j(\tau)    \vert  -  \hat{\sigma}_j \vert  \hat{\theta}_j(\tau)    \vert   \right)   \right]\\
		& & +  \left[\frac{3 C_1  \sqrt{r_{\tau}} \|\Pi(\tau)- \hat{\Pi}(\tau)\|_F }{\sqrt{nT}  \sqrt{ \log(  \max\{n,pc_T \} )} } - \frac{C_2 \|\Pi(\tau)- \hat{\Pi}(\tau)\|_* }{\sqrt{nT}  \sqrt{ \log(  \max\{n,pc_T \} )} } \right]  , \\
	\end{array}
	\]	
	and the claim follows.
\end{proof}

Our next result shows how the function $	Q_{\tau}$ changes locally in the restricted set around $\theta(\tau), \Pi(\tau)$.

\begin{lemma}
	\label{lem:lowerbound}
	Under Assumption \ref{cond3}, for  all $(\delta,\Delta) \in A_{\tau}$, we have that
	\[
	Q_{\tau}(\theta(\tau) +  \delta,  \Pi(\tau) + \Delta)        -   Q_{\tau}(\theta(\tau), \Pi(\tau))  \,\geq \,   \min\left\{\frac{ \left( J_{\tau}^{1/2}(\delta,\Delta)  \right)^2  }{4} ,   q  J_{\tau}^{1/2}(\delta,\Delta)    \right \}.
	\]
\end{lemma}

\begin{proof}
	Let
	\[
	\begin{array}{l}
		v_{A_{\tau}} \,=\,   \underset{v}{\sup}     \Bigg\{    v\,:\,       Q_{\tau}(\theta(\tau) +  \tilde{\delta},  \Pi(\tau) + \tilde{\Delta})        -   Q_{\tau}(\theta(\tau), \Pi(\tau))  \geq  \frac{ \left( J_{\tau}^{1/2}(\tilde{\delta},\tilde{\Delta}) \right)^2  }{4},\,\,\forall   (\tilde{\delta},\tilde{\Delta})\in A_{\tau}, \\
		\,\,\,\,\,\,\,\,\,\,\,\,\,\,\,\,\,\,	\,\,\,\,\,\,\,\,\,\,\,\,\,\,\,\,\,\,\,\,\,\,\,\,\, J_{\tau}^{1/2}(\tilde{\delta},\tilde{\Delta})   \leq  v      \Bigg\}.	
	\end{array}
	\]
	Then by the convexity of $Q_{\tau}(\cdot)$ and the definition of $v_{A_u}$, we have that
	\[
	\begin{array}{l}
		Q_{\tau}(\theta(\tau) +  \tilde{\delta},  \Pi(\tau) + \tilde{\Delta})        -   Q_{\tau}(\theta(\tau), \Pi(\tau))  	  \\
		\geq  \, \frac{ \left( J_{\tau}^{1/2}(\delta,\Delta)  \right)^2  }{4}   \wedge \left\{ \frac{  J_{\tau}^{1/2}(\delta,\Delta)   }{v_{A_{\tau}} }  \cdot       \underset{(\tilde{\delta},\tilde{\Delta})\in A_{\tau},\, J_{\tau}^{1/2}(\tilde{\delta},\tilde{\Delta}) \geq v_{A_{\tau}}       }{\inf }   Q_{\tau}(\theta(\tau) +  \tilde{\delta},  \Pi(\tau) + \tilde{\Delta})        -   Q_{\tau}(\theta(\tau), \Pi(\tau)  )   \right\}\\
		\geq \,\, \frac{ \left( J_{\tau}^{1/2}(\delta,\Delta)  \right)^2  }{4}   \wedge  \left\{   \frac{ J_{\tau}^{1/2}(\delta,\Delta)    }{ v_{A_{\tau}} }  \frac{   v_{A_{\tau}}^2  }{4} \right\}\\
		\geq  \frac{ \left( J_{\tau}^{1/2}(\delta,\Delta)  \right)^2  }{4} \wedge  q  J_{\tau}^{1/2}(\delta,\Delta),
	\end{array}
	\]
	where  in last inequality we have used the fact that $v_{A_{\tau}} \,\geq \, 4q$. To  see why this is true, notice that  there exists  $z_{X_{it} ,z } \in [0,z]$  such that
	\begin{equation}
		\label{eqn:lower_b}
		\begin{array}{l}
			Q_{\tau}(\theta(\tau) +  \delta,  \Pi(\tau) + \Delta)        -   Q_{\tau}(\theta(\tau), \Pi(\tau))   \\
			=\,\, \displaystyle \frac{1}{nT}  \sum_{i=1}^{n}\sum_{t=1}^T  \mathbb{E}   \left(  \int_0^{X_{i,t}^{\prime} \delta  + \Delta_{it}   }  \left(F_{Y_{i,t}|X_{i,t  } ,\Pi_{i,t} } (X_{i,t}^{\prime} \theta(\tau)  + \Pi_{i,t}   +z ) - F_{Y_{i,t}|X_{i,t  }, \Pi_{i,t} }(X_{i,t}^{\prime} \theta(\tau)  + \Pi_{i,t}(\tau)    )   \right ) dz       \right) \\
			=\,\, \displaystyle \frac{1}{nT}  \sum_{i=1}^{n}\sum_{t=1}^{T}  \mathbb{E} \Bigg(     \int_0^{X_{i,t}^{\prime} \delta  + \Delta_{i,t}   }  \left(     z f_{Y_{i,t} |X_{i,t}, \Pi_{i,t}(\tau)   }( X_{i,t}^{\prime} \theta(\tau) + \Pi_{i,t}(\tau)   )   \right)  +  \\
			\,\,\frac{z^2}{2}f_{Y_{i,t} |X_{i,t},\Pi_{i,t}   }^{\prime}(    X_{i,t}^{\prime}\theta(\tau)   + \Pi_{i,t}(\tau) +z_{X_{i,t} ,z }  ) dz  \Bigg)\\
			\geq      \,\,  \displaystyle  \frac{\underline{f}}{nT}  \sum_{i=1}^{n}\sum_{t=1}^T \mathbb{E} \left(   \left(X_{i,t}^{\prime} \delta  + \Delta_{i,t}  \right)^2  \right)    \,-\,  \frac{1}{6}  \frac{\bar{f}^{\prime}  }{nT}  \sum_{i=1}^{n}\sum_{t=1}^T \mathbb{E} \left(   \left\vert X_{i,t}^{\prime} \delta  + \Delta_{i,t}  \right\vert^3  \right).
		\end{array}
	\end{equation}
	
	Hence,  if  $(\delta,\Delta) \in A_{\tau}$  with  $ J_{\tau}^{1/2}(\delta,\Delta)     \,\leq \, 4q$ then
	\[
	\sqrt{  \frac{ \underline{f}  }{nT}  \sum_{i=1}^n\sum_{t=1}^{T}   \mathbb{E}\left( \left( X_{i,t}^{\prime}\delta + \Delta_{i,t}   \right)^2      \right)       }   \,\leq \,  \frac{3}{2} \frac{  \underline{f}^{3/2}  }{\bar{f}^{\prime}  }\cdot \underset{ (\delta,\Delta) \in A_{\tau},  \delta \neq 0   }{\inf}\,\frac{  \left( \mathbb{E}   \left( \frac{1}{nT}\sum_{i=1}^{n}\sum_{t=1}^{T}  (  X_{i,t}^{\prime} \delta  + \Delta_{i,t}  )^2   \right) \right)^{3/2}   }{ \mathbb{E} \left(    \frac{1}{nT}   \sum_{i=1}^{n}\sum_{t=1}^{T}     \vert  X_{i,t}^{\prime} \delta    + \Delta_{i,t} \vert^3     \right)      }
	\]
	combined with (\ref{eqn:lower_b}) implies
	\[
	Q_{\tau}(\theta(\tau) +  \delta,  \Pi(\tau) + \Delta)        -   Q_{\tau}(\theta(\tau), \Pi(\tau)) \,\geq \, \frac{ \left( J_{\tau}^{1/2}(\delta,\Delta)  \right)^2  }{4} .
	\]
	
\end{proof}

Next we study the behavior of the $\ell_2$ and  nuclear norms in the restricted set $A_{\tau}$.

\begin{lemma}
	\label{lem:norm1}
	Under Assumption \ref{cond3},	for  all $(\delta,\Delta)  \in A_{\tau}$,  we have
	\[
	\|\delta\|_{1,n,T}\,\leq \, \frac{2(C_0+1)}{\kappa_0}\left(\sqrt{s_{\tau}+1}  +  \frac{\sqrt{r_{\tau}}}{\sqrt{\log(n \vee p c_T) }}   \right)  J_{\tau}^{1/2}(\delta,\Delta),
	\]
	and
	\[
	\| \Delta \|_{*}\,\leq \, (C_0+1)\sqrt{n T \log(\max\{ p c_T,n \}   )   } \kappa_0^{-1} \left(\sqrt{s_{\tau}+1}  +  \frac{\sqrt{r_{\tau}}}{\sqrt{\log(n \vee p c_T) }}   \right)  J_{{\tau}}^{1/2}(\delta,\Delta),
	\]
	with  $C_0$ as in Lemma \ref{lem1}.
\end{lemma}

\begin{proof}
	By Cauchy-Schwartz's inequality,  the definition of $A_{\tau}$  and  $J_{{\tau}}^{1/2}(\delta,\Delta)$, and Assumption \ref{cond3},  we have
	\[
	\begin{array}{lll}
		\|\delta\|_{1,n,T} &\leq & \frac{5}{4} \left(\|\delta_{T_{\tau}} \|_{1}  \,+\, \|\delta_{T_{\tau}^c} \|_{1}\right)\\
		&\leq  &    \frac{5}{4} \|\delta_{T_{\tau}} \|_{1}\,  + \, \frac{5C_0}{4} \left(  \|\delta_{T_{\tau}}\|_1 +  \frac{  \sqrt{r_{\tau}} \|  \Delta \|_{F}}{\sqrt{ nT\log ( n \vee p c_T) } }\, \right)\\
		& \leq & 2(C_0+1)\sqrt{s_{\tau}}\|\delta_{T_{\tau}}\|_2 \,+\,2C_0\left(   \frac{\sqrt{r_{\tau}}\| \Delta  \|_{F}}{\sqrt{ nT\log(n \vee p c_T) }}  \right)\\
		& \leq & 2(C_0+1)\left( \sqrt{s_{\tau}+1}  + \frac{  \sqrt{r_{\tau}}}{\sqrt{\log(n \vee p c_T) }} \right)\left(\|\delta_{T_{\tau}}\|_2 \,+\,  \frac{ 1}{\sqrt{nT }} \| \Delta \|_{F}   \right)\\
		& \leq&  2(C_0+1)  \left( \sqrt{s_{\tau}+1}  + \frac{  \sqrt{r_{\tau}}}{\sqrt{\log(n \vee p c_T) }} \right)   \frac{J_{{\tau}}^{1/2}(\delta,\Delta)}{\kappa_0}.
	\end{array}
	\]
	On the other hand, by the triangle  inequality,  the construction of the set $A_{\tau}$, and  Cauchy-Schwartz's inequality
	\[
	\begin{array}{lll}
		\|\Delta\|_*   & \leq & C_0\sqrt{n T \log(n \vee p c_T)   } \left(   \|\delta_{T_{\tau}}\|_1  +  \frac{ \sqrt{r_{\tau}} \| \Delta \|_F}{ \sqrt{n T \log(n \vee p c_T )   } }   \right)\\
		& \leq&  (C_0+1) \sqrt{n T \log(n \vee p c_T)   } \left(   \sqrt{s_{\tau}+1}\|\delta_{T_{\tau}}\|_2  +  \frac{  \sqrt{r_{\tau}} \| \Delta\|_F}{ \sqrt{n T \log(n \vee p c_T)   } }   \right)\\
		& \leq &(C_0+1) \sqrt{n T \log( n \vee p c_T )   }    \left(\sqrt{s_{\tau}+1}   +  \frac{  \sqrt{r_{\tau}}}{\sqrt{\log(n \vee p c_T) }}\right)    \frac{J_{{\tau}}^{1/2}(\delta,\Delta)}{\kappa_0}.
	\end{array}
	\]
\end{proof}

Using all the previous lemmas our next results provides the control of the empirical process associated with our estimator. Control of this empirical process leads to the convergence rates obtained in the paper.

\begin{lemma}
	\label{lem:empiricalbound}
	Let
	\[
	\begin{array}{l}
		\epsilon(\eta) \,=\, \underset{   (\delta,\Delta) \in A_{\tau}  \,:\,     J_{\tau}^{1/2}(\delta,\Delta) \leq \eta  }{\sup} \,\,\  \Bigg \vert  \hat{Q}_{\tau}(\theta(\tau)+ \delta,\Pi(\tau) + \Delta) - \hat{Q}_{\tau}( \theta,\Pi  ) -\\
		\,\,\,\,\,\,\,	\,\,\,\,\,\,\,\,\,\,\,\,\,\,\, Q_{\tau}(\theta(\tau)+ \delta,\Pi(\tau) + \Delta) + Q_{\tau}( \theta(\tau),\Pi(\tau)  )   \Bigg\vert,
	\end{array}
	\]
	and  $\{\phi_{n}\}$  a sequence  with $\phi_n/(\sqrt{\underline{f} } \log(c_T +1)  ) \,\rightarrow \, \infty$. Then for all $\eta>0$
	\[
	\epsilon(\eta) \,\leq \, \frac{\tilde{C}_0 \eta c_T  \phi_{n} \sqrt{\log(p c_T\vee n)}(\sqrt{1+s_{\tau}  }+   \sqrt{ r_{\tau}/ \log(p c_T\vee n) }  )(\sqrt{n} +\sqrt{d_T})  }{  \sqrt{nT} \kappa_{0}   \underline{f}^{1/2} },
	\]
	for some constant $\tilde{C}_0 >0$, 	with probability at least $1-  \alpha_{n}$. Here, the sequence $ \{\alpha_{n}\}$ is independent of $\eta$, and   $\alpha_{n}  \,\rightarrow 0$.
\end{lemma}

\begin{proof}
	Let $\Omega_1$ be the event  $   \max_{j \leq p} \vert \hat{\sigma}_j - 1\vert  \,\leq \, 1/4$. Then, by Assumption ,  $P(\Omega_1) \,\geq \, 1- \gamma$ .  
	Next let $\kappa >0$, and  $f = (\delta, \Delta) \in A_{\tau}$ and write
	\[
	\mathcal{F}\,=\,  \{ (\delta,\Delta) \in A_{\tau}\,:\,    \,\,\,\,\,\,\,   \,\,J_{\tau}^{1/2}(\delta,\Delta) \leq \eta  \}.
	\]
	Then notice that by Lemmas  4.1 and 4.2 from \cite{yu1994rates},	
	\begin{equation}
		\label{eqn:s6}
		\begin{array}{lll}
			\mathbb{P}\left( \epsilon(\eta) \sqrt{nT} \,\geq \, \kappa  \right) &\leq &   \displaystyle  2 \,\mathbb{P}\left(   \underset{ f \in \mathcal{F}   }{\sup}\,  \frac{1}{ \sqrt{n d_T} }\left\vert    \sum_{ i= 1}^n \sum_{l=1}^{d_T}  \sum_{ t \in H_l }  \frac{ Z_{i,t}(f) }{  \sqrt{c_T} }   \right\vert  \,\geq \,  \frac{  \kappa }{3}    \right) \,+\,  \\
			& &   \displaystyle  \mathbb{P}\left(   \underset{ f \in \mathcal{F}   }{\sup}\,  \frac{1}{ \sqrt{n d_T} }\left\vert    \sum_{ i= 1}^n \sum_{ t\in  R }  \frac{ Z_{i,t}(f) }{  \sqrt{c_T} }   \right\vert  \,\geq \,  \frac{  \kappa }{3}    \right) \,+\,2 nT \left( \frac{1}{c_T} \right)^{\mu}\\
			&=: & A_1  + A_2 +2 nT \left( \frac{1}{c_T} \right)^{\mu},
		\end{array}
	\end{equation}	
	where
	\[
	\begin{array}{lll}
		Z_{i,t}(f)&=&   \rho_{\tau}( \widetilde{Y}_{i,t} -  \widetilde{X}_{i,t}^{\prime}(\theta(\tau) +\delta)  -(\Pi_{i,t}(\tau) +\Delta_{i,t} ) ) - \rho_{\tau}( \widetilde{Y}_{i,t} -  \widetilde{X}_{i,t}^{\prime}\theta(\tau)   - \Pi_{i,t}(\tau) ) )\\
		& & - \mathbb{E}\left(  \rho_{\tau}( \widetilde{Y}_{i,t} - \widetilde{X}_{i,t}^{\prime}(\theta(\tau) +\delta)  -(\Pi_{i,t}(\tau) +\Delta_{i,t} ) ) - \rho_{\tau}(\widetilde{Y}_{i,t} - \widetilde{X}_{i,t}^{\prime}\theta(\tau)   - \Pi_{i,t}(\tau) ) ) \right).
	\end{array}
	\]
	Next we proceed  to bound each term in (\ref{eqn:s6}). To that end, notice that
	\[
	\begin{array}{lll}
		\displaystyle	\text{Var}\left(  \sum_{ i= 1}^n \sum_{l=1}^{d_T}  \sum_{ t \in H_l }  \frac{ Z_{i,t}(f) }{  \sqrt{c_T} }  \right) & \leq & \displaystyle \sum_{ i= 1}^n \sum_{l=1}^{d_T} \mathbb{E}\left(   \left(  \frac{1}{\sqrt{c_T}} \sum_{ t \in H_l }     Z_{i,t}(f)  \right)^2\right)   \\
		& \leq&   \displaystyle \sum_{ i= 1}^n \sum_{l=1}^{d_T} \sum_{ t \in H_l } \mathbb{E}\left(  \left(   \tilde{X}_{i,t}^{\prime}\delta  +  \Delta_{i ,t}  \right)^2 \right)\\
		& \leq&  \displaystyle   \frac{  nT }{  \underline{f} } \left(J_{\tau }^{1/2}(\delta,\Delta)\right)^2.
	\end{array}
	\]
	Let $\{\varepsilon_{i,l}\}_{i \in [n],\, l \in [d_T]}$ be i.i.d  Rademacher variables independent of  the data.
	Therefore, by Lemma  2.3.7  in \cite{wellner2013weak}
	\begin{equation}
		\label{eqn:symmetrization}
		\begin{array}{lll}
			\displaystyle 	\,\mathbb{P}\left(   \underset{ f \in \mathcal{F}   }{\sup}\,  \frac{1}{ \sqrt{n d_T} }\left\vert    \sum_{ i= 1}^n \sum_{l=1}^{d_T}  \sum_{ t \in H_l }  \frac{ Z_{i,t}(f) }{  \sqrt{c_T} }   \right\vert \geq \kappa\right)& \leq&  \displaystyle  \frac{ \mathbb{P}\left(  \underset{f \in \mathcal{F}}{\sup} \left\vert \frac{1}{\sqrt{n d_T}}  \sum_{i=1}^{n}\sum_{l=1}^{d_T}  \varepsilon_{i,l} \left(\sum_{t \in H_l} \frac{ Z_{i,t}(f)}{ \sqrt{c_T} }\right)  \right\vert     \,\geq \,  \frac{\kappa}{4}  \right)  }{ 1 -   \frac{4}{nT  \kappa^2 }  \,\underset{f  \in \mathcal{F} }{\sup} \, \text{Var}( \sum_{i=1}^{n}\sum_{l=1}^{d_T}   \sum_{t \in H_l}  \frac{Z_{i,t}(f) }{\sqrt{c_T} }   )   }\\
			&\leq  &  \frac{   \mathbb{P}\left(   A^0(\eta )    \,\geq \,  \frac{\kappa}{12}  | \Omega_1\right) \,+\,   \mathbb{P} (\Omega_1^c) }{ 1-  \frac{24  c_T \eta ^2 }{\underline{f}  \kappa^2} },
		\end{array}
	\end{equation}
	where
	\[
	\begin{array}{l}	A^0(\eta ) \,:=\, \\
		\displaystyle  \underset{f \in \mathcal{F}}{\sup} \left\vert \frac{1}{\sqrt{n d_T}}  \sum_{i=1}^{n}\sum_{l=1}^{d_T}  \varepsilon_{i,l}  \left(\sum_{t \in H_l} \frac{	\rho_{\tau}(\widetilde{Y}_{i,t} - \widetilde{X}_{i,t}^{\prime} (\theta(\tau) + \delta )   -   (\Pi_{i,t}(\tau)+ \Delta_{i,t}) )   -  \rho_{\tau}(\widetilde{Y}_{i,t} - \widetilde{X}_{i,t}^{\prime} \theta(\tau)    -   \Pi_{i,t}(\tau) )}{ \sqrt{c_T} }\right) \right\vert .
	\end{array}
	\]
	Next, note that
	\[
	\begin{array}{lll}
		\rho_{\tau}(\widetilde{Y}_{i,t} - \widetilde{X}_{i,t}^{\prime} (\theta(\tau) + \delta )   -   (\Pi_{i,t}(\tau)+ \Delta_{i,t}) )   -  \rho_{\tau}(\widetilde{Y}_{i,t} - \widetilde{X}_{i,t}^{\prime} \theta(\tau)    -   \Pi_{i,t}(\tau) )   & =& \tau\left(\widetilde{X}_{i,t}^{\prime}\delta +  \Delta_{i,t}  \right)\\
		& &  +    v_{i,t}(\delta,\Delta)   \,+\,  w_{i,t}(\delta,\Delta) ,
	\end{array}
	\]
	where
	\begin{equation}
		\label{eqn:contraction1}
		\begin{array}{lll}
			\vert  v_{i,t}(\delta,\Delta) \vert & =& \left\vert  (\widetilde{Y}_{i,t} - \widetilde{X}_{i,t}^{\prime} (\theta(\tau) + \delta )   -   (\Pi_{i,t}(\tau)+ \Delta_{i,t})  )_{-} -  (Y_{i,t} - \widetilde{X}_{i,t}^{\prime} (\theta(\tau) + \delta )   -   \Pi_{i,t}(\tau)   )_{-}  \right\vert\\
			&\leq  & \vert \Delta_{i,t}\vert.
		\end{array}
	\end{equation}
	and
	\begin{equation}
		\label{eqn:contraction2}
		\begin{array}{lll}
			\vert  w_{i,t}(\delta,\Delta) \vert & =& \left\vert  (\widetilde{Y}_{i,t} - \widetilde{X}_{i,t}^{\prime} (\theta(\tau) + \delta )   -   \Pi_{i,t}(\tau)   )_{-} -  (\widetilde{Y}_{i,t} - \widetilde{X}_{i,t}^{\prime}\theta(\tau)    -   \Pi_{i,t}(\tau)   )_{-}  \right\vert\\
			&\leq  & \vert \widetilde{X}_{i,t}^{\prime} \delta\vert.
		\end{array}
	\end{equation}
	
	Moreover, notice that by Lemma \ref{lem:norm1},
	\[
	\{ (\delta,\Delta ) \in A_{\tau} \,:\,   J_{\tau}^{1/2}(\delta,\Delta) \leq \eta  \} \,\subset   \{ (\delta,\Delta ) \in A_{\tau} \,:\,    \|\delta\|_{1,n,T} \leq  \eta \upsilon \},
	\]
	where
	\[
	\upsilon \,:=\,  \frac{2(C_0+1)}{\kappa_0} \left(\sqrt{1+s_{\tau}} +  \frac{  \sqrt{r_{\tau}}}{\sqrt{\log(n \vee p c_T) }}\right)  .
	\]
	Also by Lemma \ref{lem:norm1}, for  $(\delta,\Delta) \in A_{\tau}$
	\[
	\|\Delta\|_{*} \,\leq \,   \frac{ (C_0+1)   J^{1/2}_{\tau}(\delta,\Delta)   \sqrt{  nT  \log(p c_T\vee n) }}{\kappa_0}\left(\sqrt{1+s_{\tau}} +  \frac{  \sqrt{r_{\tau}}}{\sqrt{\log(n \vee p c_T) }}\right) ,
	\]
	and so,
	\[
	\begin{array}{l}
		\{ (\delta,\Delta ) \in A_{\tau} \,:\,   \,\,  J_{\tau}^{1/2}(\delta,\Delta) \leq \eta  \} \,\subset  \,\\
		\left\{  (\delta,\Delta ) \in A_{\tau} \,:\,  \,\,\|\Delta\|_* \leq  (C_0+1)\eta \sqrt{nT   \log(p c_T\vee n)} (\sqrt{s_{\tau}+1} +  \sqrt{r_{\tau}/\log(p c_T\vee n)} )/\kappa_0  \right \}.
		
	\end{array}
	\]
	Hence, defining
	\[
	\begin{array}{l}
		\displaystyle  B_1^0(\eta) \,=\,  \sqrt{c_T}  \underset{  \delta \,:\,  \|\delta\|_{1,n,T} \leq  \eta\upsilon }{\sup} \left\vert \frac{1}{\sqrt{n d_T}} \sum_{i=1}^{n}\sum_{l=1}^{d_T} \varepsilon_{i,l}\left(\sum_{t \in H_l} \frac{ \widetilde{X}_{i,t}^{\prime}\delta }{c_T}  \right) \right\vert, \,\,\\
		\displaystyle 	  B_2^0(\eta) \,=\, \sqrt{c_T} \underset{   \Delta    \,:\,  \|\Delta\|_* \leq  (C_0+1)\eta \sqrt{nT   \log(p c_T\vee n)} (\sqrt{s_{\tau}+1} +  \sqrt{r_{\tau}/\log(p c_T\vee n)} )/\kappa_0  }{\sup} \left\vert \frac{1}{\sqrt{n d_T}} \sum_{i=1}^{n}\sum_{l=1}^{d_T} \varepsilon_{i,l} \left(\sum_{t \in H_l} \frac{ \Delta_{i,t}}{c_T} \right) \right\vert\\
		\displaystyle  B_3^0(\eta) \,=\, \sqrt{c_T} \underset{ \Delta   \,:\,  \|\Delta\|_* \leq  (C_0+1)\eta \sqrt{nT   \log(p c_T\vee n)} (\sqrt{s_{\tau}+1} +  \sqrt{r_{\tau}/\log(p c_T\vee n)} )/\kappa_0  }{\sup} \left\vert \frac{1}{\sqrt{n d_T}} \sum_{i=1}^{n}\sum_{l=1}^{d_T}  \varepsilon_{i,l} \left( \sum_{t \in H_l} \frac{ v_{i,t}(\delta,\Delta)}{c_T  }\right)    \right\vert, \,\,\\
		\displaystyle  B_4^0(\eta) \,=\, \sqrt{c_T} \underset{ \delta \,:\,  \|\delta\|_{1,n,T} \leq  \eta \upsilon }{\sup} \left\vert \frac{1}{\sqrt{n d_T}} \sum_{i=1}^{n}\sum_{l=1}^{d_T}  \varepsilon_{i,l}\left(\sum_{t\in H_l} \frac{ w_{i,t}(\delta,\Delta)}{  c_T }\right)   \right\vert.\\
	\end{array}
	\]
	By union bound	we obtain that
	\begin{equation}
		\label{eqn:s7}
		\mathbb{P}(A^0(\eta)   \,\geq \, \kappa | \Omega_1) \,\leq \, \displaystyle    \sum_{j=1}^{4} \mathbb{P}(B_j^0(\eta)   \,\geq \, \kappa | \Omega_1),
	\end{equation}
	so we proceed to bound each term in the right hand side of the inequality above.
	
	
	First, notice that
	\[
	\begin{array}{lll}
		B_1^0(\eta) & \leq &\displaystyle   2 c_T \underset{ m \in  [c_T]   }{\max} \,\,\,\underset{  \delta \,:\,  \|\delta\|_{1,n,T} \leq  \eta\upsilon }{\sup} \left\vert \frac{1}{\sqrt{n T}} \sum_{i=1}^{n}\sum_{l=0}^{d_T- 1} \varepsilon_{i,l}\widetilde{X}_{i ,\,( 2l c_T +m )}^{\prime}\delta \right\vert,\\
	\end{array}
	\]
	and hence  by  a union bound and the same  argument on the proof of Lemma 5 in \cite{belloni2011l1}, we have that
	
	\begin{equation}
		\label{eqn:s8}
		\mathbb{P}(B_1^0(\eta)   \,\geq \, \kappa | \Omega_1)  \,\leq \,  2p c_T \exp\left(   -  \frac{\kappa^2}{  4c_T^2 (16\sqrt{2}  \eta \upsilon)^2   }  \right).
	\end{equation}
	
	Next we proceed to bound  $B_3^0(\eta)$, by noticing that
	\[
	\begin{array}{lll}
		B_3^0(\eta) & \leq & \displaystyle   \underset{ m \in  [c_T]   }{\max} \,\,\,  \underset{ \Delta   \,:\,  \|\Delta\|_* \leq  (C_0+1)\eta \sqrt{nT   \log(p c_T\vee n)} (\sqrt{s_{\tau}+1} +  \sqrt{r_{\tau}/\log(p c_T\vee n)} )/\kappa_0  }{\sup} \left\vert \frac{\sqrt{c_T}}{\sqrt{n d_T}} \sum_{i=1}^{n}\sum_{l=0}^{d_T- 1}  \varepsilon_{i,l}  v_{i, \,  \,( 2l c_T +m )}(\delta,\Delta)    \right\vert.
	\end{array}
	\]
	Towards  that end we  proceed to bound the moment  generating function of $B_3^0(\eta)$ and the use that to obtain an upper bound on $B_3^0(\eta)$.
	Now  fix  $m \in [d_T]$  and notice that
	\begin{equation}
		\label{eqn:s10}
		\begin{array}{l}
			\displaystyle	   \mathbb{E} \left(\exp\left( \lambda   \underset{ \Delta    \,:\,  \|\Delta\|_* \leq  (C_0+1)\eta \sqrt{nT   \log(p c_T\vee n)} (\sqrt{s_{\tau}+1} +  \sqrt{r_{\tau}/\log(p c_T\vee n)} )/\kappa_0    }{\sup} \left\vert  \frac{\sqrt{c_T}}{\sqrt{n d_T}} \sum_{i=1}^{n}\sum_{l=0}^{d_T- 1}  \varepsilon_{i,l}  v_{i ,\,  \,( 2l c_T +m )}(\delta,\Delta)  \right\vert \right) \right) \\
			\displaystyle  \leq \,   \mathbb{E} \left(\exp\left( \lambda   \underset{ \Delta    \,:\,  \|\Delta\|_* \leq  (C_0+1)\eta \sqrt{nT   \log(p c_T\vee n)} (\sqrt{s_{\tau}+1} +  \sqrt{r_{\tau}/\log(p c_T\vee n)} )/\kappa_0  }{\sup} \left\vert  \frac{ \sqrt{c_T}}{\sqrt{n d_T}} \sum_{i=1}^{n}\sum_{l=1}^{d_T}  \varepsilon_{i,l} \, \Delta_{i, \,  \,( 2l c_T +m )}   \right\vert \right) \right) \\
			\displaystyle  \leq \,  \mathbb{E} \Bigg(\underset{ \Delta    \,:\,  \|\Delta\|_* \leq (C_0+1)\eta \sqrt{nT   \log(p c_T\vee n)} (\sqrt{s_{\tau}+1} +  \sqrt{r_{\tau}/\log(p c_T\vee n)} )/\kappa_0   }{\sup}  \Bigg(\exp\left( \frac{\lambda  \sqrt{c_T} \| \Delta \|_* \,\mathbb{E}( \|\{\varepsilon_{il}\}\|_2)}{\sqrt{n d_T}}\right)\,  \\
			\,\,\,\,\,\,\,\,\,\displaystyle  \exp \left( \lambda \frac{ \sqrt{c_T} \|\Delta\|_*  ( \|\{\varepsilon_{i,l}\}\|_2  - \mathbb{E}( \|\{\varepsilon_{i,l}\}\|_2 ) )}{\sqrt{n d_T}}  \right)  \Bigg) \Bigg)\\
			\displaystyle  \leq \,\exp\left(   \lambda  \frac{ (C_0+1) c_4\eta  c_T \sqrt{3 \log(p c_T\vee n)  } (   \sqrt{1+s_{\tau}}  +   \sqrt{r_{\tau}/   \log(p c_T\vee n)   }  ) \left( \sqrt{n} + \sqrt{d_T} \right)  }{  \kappa_0 } \right)\cdot \\
			\displaystyle\exp\left(   \frac{ (C_0+1)^2 c_4 \lambda^2  c_T^2 \eta^2   \log(p c_T\vee n)    ( s_{\tau}+1 + r_{\tau}/  \log(p c_T\vee n)   ) }{\kappa_0^2}   \right),
		\end{array}
	\end{equation}
	for a positive constant $c_4>0$, and	where the first inequality holds  by  Ledoux-Talagrand's contraction inequality, the  second   by the  the duality of  the spectral and nuclear norms and the triangle inequality, the third  by Theorem 1.2 in \cite{vu2007spectral}  and  by basic properties of sub-Gaussian random variables.
	
	Therefore,  by Markov's inequality  and (\ref{eqn:s10}),
	\begin{equation}
		\label{eqn:s11}
		\begin{array}{l}
			\,\,\,\,\,\, \mathbb{P}\left( B_3^0(\eta ) \geq \kappa | \Omega_1 \right) \\
			\leq  \displaystyle  c_T  \underset{  m \in [c_T] }{ \max }  \mathbb{P}\Bigg(    \underset{ \Delta\,:\,  \|\Delta\|_* \leq     (C_0+1)\eta \sqrt{nT   \log(p c_T\vee n)} (\sqrt{s_{\tau}+1} +  \sqrt{r_{\tau}/\log(p c_T\vee n)} )/\kappa_0    }{\sup} \left\vert  \frac{\sqrt{c_T}}{\sqrt{n d_T}} \sum_{i=1}^{n}\sum_{l=0}^{d_T- 1}  \varepsilon_{i,l}  v_{i, \,  \,( 2l c_T +m )}(\delta,\Delta)  \right\vert  \\
			 \,\,\,\,\,\,\,\,\,\,\,\, \,\,\,\,\,\,\,\,\,\,\,\,\, \,\,\,\,\,\,\,\,\,\,\,\,\, \,\,\,\geq \, \kappa \Bigg)\\
			\leq    \underset{  \lambda>0 }{  \inf }\,\Bigg[  \exp\left( -  \lambda \kappa \right)  \exp\left(   \lambda   \frac{(C_0+1) c_4\eta  c_T \sqrt{3 \log(p c_T\vee n)  }(\sqrt{1+s_{\tau}}   +\sqrt{r_{\tau}/ \log(p c_T\vee n)    }    )  \left( \sqrt{n} + \sqrt{d_T} \right)  }{  \kappa_0 } \right)\cdot \\
			\,\,\,\,\,\,\,		 \displaystyle\exp\left(   \frac{ (C_0+1)^2 c_4 \lambda^2 c_T^2  \eta^2   \log(p c_T\vee n) ((1 +s_{\tau})   +   r_{\tau}/ \log(p c_T\vee n)     ) }{\kappa_0^2}   +  \log c_T \right) \Bigg]\\
			\leq  c_5\exp\left( -  \frac{  \kappa \kappa_0  }{    (C_0+1)  \eta c_T \sqrt{3  \log(p c_T\vee n) }(\sqrt{1+s_{\tau}}  + \sqrt{r_{\tau}/\log(p c_T\vee n) }  )\left( \sqrt{n}  +  \sqrt{d_T} \right)        }    +  \log c_T  \right),
		\end{array}		
	\end{equation}
	for a positive constant $c_5 >0$.

	Furthermore,  we observe that
	\[
	B_2^0(\eta )\,\leq \,   \underset{  m \in [d_T] }{ \max } 	 \underset{   \Delta    \,:\,  \|\Delta\|_* \leq  (C_0+1)\eta \sqrt{nT   \log(p c_T\vee n)} (\sqrt{s_{\tau}+1} +  \sqrt{r_{\tau}/\log(p c_T\vee n)} )/\kappa_0 }{\sup} \left\vert \frac{\sqrt{c_T}}{\sqrt{n d_T}} \sum_{i=1}^{n}\sum_{l=0}^{d_T-1} \varepsilon_{i,l} \Delta_{i ,  \,( 2l c_T +m )} \right\vert.
	\]
	Hence, with the same argument  for bounding $B_3^0(\eta )$, we have
	\begin{equation}
		\label{eqn:s12}
		\begin{array}{l}
				\mathbb{P}\left( B_2^0(\eta ) \geq 
			\kappa | \Omega_1 \right)\\ \,\leq\,  c_5\exp\left( -  \frac{  \kappa \kappa_0  }{   (C_0+1) \eta c_T \sqrt{3   \log(p c_T\vee n) }( \sqrt{1+s_{\tau}}  + \sqrt{r_{\tau}/\log(p c_T\vee n)}     )\left( \sqrt{n}  +  \sqrt{d_T} \right)        }    +  \log c_T  \right).
		\end{array}
		\end{equation}
	
	Finally,  we proceed to bound $B_4^0(\eta ) $. To  that end, notice that
	\[
	B_4^0(\eta )  \,\leq \, \underset{  m \in [d_T] }{ \max } \,\,	 \underset{ \delta \,:\,  \|\delta\|_{1,n,T} \leq  \eta \upsilon }{\sup} \left\vert \frac{ \sqrt{c_T}}{\sqrt{n d_T}} \sum_{i=1}^{n}\sum_{l=0}^{d_T -1}  \varepsilon_{i,l}  w_{i , \, (  2lc_T +m )   }  (\delta,\Delta) \right\vert,
	\]
	and  by  (\ref{eqn:contraction2}) and Ledoux-Talagrand's inequality,  as in (\ref{eqn:s8}), we obtain
	
	\begin{equation}
		\label{eqn:s13}
		\mathbb{P}(B_4^0(\eta)   \,\geq \, \kappa | \Omega_1)  \,\leq \,  2p c_T \exp\left(   -  \frac{\kappa^2}{  4c_T^2 (16\sqrt{2}  \eta \upsilon)^2   }  \right).
	\end{equation}

	Therefore,  letting
	\[
	\kappa \,=\, \frac{  \eta c_T  \phi_{n} (1+C_0)^2\sqrt{    \log(p c_T\vee n)}(\sqrt{(1+s_{\tau})  }+  \sqrt{r_{\tau}/  \log(p c_T\vee n)} )(\sqrt{n} +\sqrt{d_T})  }{\kappa_{0}   \underline{f}^{1/2} },
	\]
	and repeating the argument  above   for bounding $A_2$ in (\ref{eqn:s6}), combining (\ref{eqn:s6}), (\ref{eqn:symmetrization}), \eqref{eqn:s7}, \eqref{eqn:s8}, \eqref{eqn:s11}, \eqref{eqn:s12} and \eqref{eqn:s13},
	we obtain that
	\[
	\begin{array}{lll}
		\mathbb{P}(  \epsilon(\eta  )  \,\geq \,  \frac{\kappa}{\sqrt{nT}} )  & \leq &  5\frac{\gamma + 4\exp\left(  \log(p c_T\vee n)  -  C_1\frac{ \phi_{n}^2\log(p c_T\vee n)}{  \underline{f}   } (\sqrt{n} +\sqrt{d_T})^2 \right)   +     2 c_5\exp\left(- C_2 \frac{2\phi_{n} }{\underline{f}^{1/2} }  \right)  }{1 - \frac{3\kappa_0^2 }{  c_T\phi_{n} \,(1+C_0)^2( \sqrt{n} +\sqrt{d_T})^2 \log(p c_T\vee n) ( \sqrt{s_{\tau}+1} +   \sqrt{r_{\tau}/  \log(p c_T\vee n) }  )^2   }}   \\
		& &\,+\,   nT \left( \frac{1}{c_T} \right)^{\mu},
	\end{array}
	\]
	for some positive constants $C_1$  and $C_2$.
\end{proof}

Combining all the previous lemmas  we prove  Theorem \ref{thm:1} in the next subsection.

\subsection{Proof of Theorem \ref{thm:1}}
\label{sec:proog_thm1}
\begin{proof}
	
	Recall  from Lemma \ref{lem1}, our choices of $ \nu_{1}$ and  $\nu_2$ are
	\[
	\nu_{1} \,=\,C_{0}^{\prime}\sqrt{       \frac{   c_T \log( \max\{n,pc_T\} )  }{   n d_T}  },
	\]
	and
	\[
	\nu_2 \,=\,\frac{200 c_T}{n T}\left( \sqrt{n}  +  \sqrt{d_T}  \right),
	\]
	for $C_{0}^{\prime}  \,=\, 9/(1-c_0)$, and  $c_0$  as in Lemma  \ref{lem1}.
	
	Let
	\begin{equation}
		\label{eqn:eta_def}
		\eta \,=\, \frac{8    \phi_{n} ( C_{0}^{\prime}(1+C_0) +  \tilde{C}_0 + 200(1+C_0))\sqrt{c_T \,  \log(p c_T\vee n)}( \sqrt{(1+s_{\tau})}+  \sqrt{r_{\tau}/\log(p c_T\vee n)}   )(\sqrt{n} +\sqrt{d_T})  }{ \sqrt{n d_T} \kappa_{0}   \underline{f}^{1/2} },
	\end{equation}
	for  $C_0$ as in Lemma  \ref{lem1},   and $\tilde{C}_{0}$ as in Lemma \ref{lem:empiricalbound}.
	
	Throughout  we assume that the  following events happen:
	\begin{itemize}
		\item $\Omega_1\,:=\,$   the event that  $(\hat{\theta}(\tau) -  \theta(\tau),  \hat{\Pi}(\tau) -\Pi(\tau)) \in A_{\tau }$.
		\item $\Omega_2\,:=\,$  the event  for which the upper bound on $\epsilon(\eta )$  in Lemma \ref{lem:empiricalbound}  holds.
	\end{itemize}

	Suppose  that
	\begin{equation}
		\label{eqn:eta}
		\vert J_{\tau}^{1/2} (\hat{\theta}(\tau) -  \theta(\tau),  \hat{\Pi}(\tau) -\Pi(\tau)) \vert   \,>\, \eta.
	\end{equation}
	Then, by the convexity of $A_{\tau }$, and of the objective  $\hat{Q}$  with its constraint, we obtain that
	\[
	\begin{array}{l}
		0  \,>\,    \underset{  (\delta,\Delta)\in A_{\tau} \,:\, \vert J_{\tau}^{1/2} (\delta,\Delta)\vert = \eta  }{\min} \hat{Q}_{\tau}(\theta(\tau)+\delta, \Pi(\tau)+\Delta)  -    \hat{Q}(\theta(\tau),\Pi(\tau))  +   \nu_{1} \left[ \| \theta(\tau)+ \delta\|_{1,n,T} - \|\theta(\tau)\|_{1,n,T}   \right]  \\
		\,\,\,\,\,\,\,\,\,\,\,\,\,\,\,\,\,\,\,\,\,\,\,\,\,\,\,\,\,\,\,\,\,\,\,\,\,\,\,\,\,\,\,\,\,\,\,\,\,\,\,\,\,\,\,\,\,\,\,\,\,\,\,+\nu_2\left[  \| \Pi(\tau)+ \Delta\|_{*} - \|\Pi(\tau)\|_{*}  \right]
	\end{array}
	\]
	Moreover,  by Lemma \ref{lem:norm1}  and the triangle inequality,
	\[
	\begin{array}{lll}
		\|\theta(\tau)\|_{1,n,T} - \| \theta(\tau)+ \delta\|_{1,n,T}& \,\leq \,&  \|\delta_{T_{\tau}}\|_{1,n,T}\\
		& \,\leq \, & 2(1+C_0) \frac{ J^{1/2}_{\tau}(\delta,\Delta)   }{\kappa_0}\left(\sqrt{1+s_{\tau}}+  \frac{  \sqrt{r_{\tau}}}{\sqrt{\log(n \vee p c_T) }}\right) ,
	\end{array}
	\]
	and
	\[
	\begin{array}{lll}
		\|\Pi(\tau)\|_{*} - \| \Pi(\tau)+ \Delta\|_{*}& \leq &  \|\Delta\|_{*}\\
		&\leq &  (1+C_0) \sqrt{n T  \log(p c_T\vee n) }\frac{ J^{1/2}_{\tau}(\delta,\Delta)   }{\kappa_0}\left(\sqrt{1+s_{\tau}}+  \frac{  \sqrt{r_{\tau}}}{\sqrt{\log(n \vee p c_T) }}\right).
	\end{array}
	\]
	\vspace{-0.1in}
	Therefore,
	\[
	\begin{array}{lll}
		0   & >& \displaystyle    \underset{  (\delta,\Delta)\in A_{\tau} \,:\, \vert J_{\tau}^{1/2} (\delta,\Delta)\vert = \eta }{\min} \hat{Q}(\theta(\tau)+\delta,\Delta+ \Pi(\tau))  - \\
		& &\,\,\,\,\,\,\,\,\,\,\,\,\,\,\,\,\,\,\,\,\,\, \,\,\,\,\,\,\,\,\,  \,\,\,\,\,\,\,\,\,  \,\,\,\,\,\,\,\,\,    \hat{Q}(\theta(\tau),\Pi(\tau))  - 2 \nu_{1}  (1+C_0) \frac{ J^{1/2}_{\tau}(\delta,\Delta)   }{\kappa_0}\left(\sqrt{1+s_{\tau}}+\frac{  \sqrt{r_{\tau}}}{\sqrt{\log(n \vee p c_T) }}\right),\\
		& &\,\,\,\,\,\,\,\,\,\,\,\,\,\,\,\,\,\,\,\,\,\, \,\,\,\,\,\,\,\,\,  \,\,\,\,\,\,\,\,\,  \,\,\,\,\,\,\,\,\,  - \nu_2(1+C_0)\sqrt{n T  \log(p c_T\vee n)  }\frac{ J^{1/2}_{\tau}(\delta,\Delta)   }{\kappa_0}\left( \sqrt{s_{\tau}+1} +    \frac{ \sqrt{r_{\tau}+1} }{ \sqrt{ \log(p c_T\vee n) }  }  \right)\\
		&=  & \displaystyle  \underset{  (\delta,\Delta)\in A_{\tau} \,:\, \vert J_{\tau}^{1/2} (\delta,\Delta)\vert = \eta  }{\min} \Bigg[  \hat{Q}_{\tau}(\theta(\tau)+\delta,\Delta+ \Pi(\tau) )  -    \hat{Q}(\theta(\tau),\Pi(\tau)  )  \\
		& &\,\,\,\,\,\,\,\,\,\,\,\,\,\,\,\,\,\,\,\,\, \,\,\,\,\,\,\,\,\,  \,\,\,\,\,\,\,\,\,  \,\,\,\,\,\,\,\,\,  - Q(\theta(\tau)+\delta,\Delta+ \Pi(\tau))  +   Q(\theta(\tau),\Pi(\tau) )  \\
		& & \,\,\,\,\,\,\,\,\,\,\,\,\,\,\,\,\,\,\,\,\, \,\,\,\,\,\,\,\,\,  \,\,\,\,\,\,\,\,\,  \,\,\,\,\,\,\,\,\,  \displaystyle  +  Q(\theta(\tau)+\delta,\Delta+ \Pi(\tau) )  -  Q(\theta(\tau),\Pi(\tau) )\, - \\
		& &\,\,\,\,\,\,\,\,\,\,\,\,\,\,\,\,\,\,\,\,\, \,\,\,\,\,\,\,\,\,  \,\,\,\,\,\,\,\,\,  \,\,\,\,\,\,\,\,\,  \,2C_{\tau} (1+ C_0)\sqrt{\frac{c_T \,  \log(p c_T\vee n)}{n d_T}}  (\sqrt{n} +\sqrt{d_T})\frac{ J^{1/2}_{\tau}(\delta,\Delta)   }{\kappa_0}\left( \sqrt{s_{\tau}+1} +     \sqrt{r_{\tau}/\log(p c_T\vee n)} \right)-\\
		& &\,\,\,\,\,\,\,\,\,\,\,\,\,\,\,\,\,\,\,\,\, \,\,\,\,\,\,\,\,\,  \,\,\,\,\,\,\,\,\,  \,\,\,\,\,\,\,\,\,  \,  \frac{200 c_T}{n T}\left( \sqrt{n}  +  \sqrt{d_T}  \right) (C_0+1)\sqrt{n T \log(p c_T\vee n) }\frac{J^{1/2}_{\tau}(\delta,\Delta)   }{\kappa_0}\cdot \\
		 && \,\,\,\,\,\,\,\,\,\,\,\,\,\,\,\,\,\,\,\,\, \,\,\,\,\,\,\,\,\,  \,\,\,\,\,\,\,\,\,  \,\,\,\,\,\,\,\,\,  \, \left( \sqrt{s_{\tau}+1} +     \sqrt{r_{\tau}/\log(p c_T\vee n)} \right)\Bigg]\\
	\end{array}
	\]
	\[
	\begin{array}{lll}
		& \geq &\displaystyle \underset{  (\delta,\Delta)\in A_{\tau} \,:\, \vert J_{\tau}^{1/2} (\delta,\Delta)\vert = \eta  }{\min} Q(\theta(\tau)+\delta,\Delta+ \Pi(\tau) )  -   Q(\theta(\tau),\Pi(\tau) ) \\
		& &\,\,\,\,\,\,\,\,\,\,\,\,\,\,\,\,\,\,\,\,\, \,\,\,\,\,\,\,\,\,  \,\,\,\,\,\,\,\,\,  \,\,\,\,\,\,\,\,\,     - [2 C_{0}^{\prime}(1+ C_0)  +  200(C_0+1)]\sqrt{\frac{c_T \,\log(p c_T\vee n)}{n d_T}}  (\sqrt{n} +\sqrt{d_T})\frac{ J^{1/2}_{\tau}(\delta,\Delta)   }{\kappa_0}\cdot\\
		 & &\,\,\,\,\,\,\,\,\,\,\,\,\,\,\,\,\,\,\,\,\, \,\,\,\,\,\,\,\,\,  \,\,\,\,\,\,\,\,\,  \,\,\,\,\,\,\,\,\, \left( \sqrt{s_{\tau}+1} + \sqrt{r_{\tau}/\log(p c_T\vee n) } \right)\\
		&  &\,\,\,\,\,\,\,\,\,\,\,\,\,\,\,\,\,\,\,\,\, \,\,\,\,\,\,\,\,\,  \,\,\,\,\,\,\,\,\,  \,\,\,\,\,\,\,\,\,  \displaystyle  - \frac{\tilde{C}_0 \eta c_T  \phi_{n} \sqrt{\log(p c_T\vee n)}(\sqrt{1+s_{\tau}  }+   \sqrt{ r_{\tau}/ \log(p c_T\vee n) }  )(\sqrt{n} +\sqrt{d_T})  }{  \sqrt{nT} \kappa_{0}   \underline{f}^{1/2} }\\
	\end{array}
	\]
	\begin{equation}
		\label{eqn:contradiction}
		\begin{array}{lll}
			&  \geq& \displaystyle \frac{\eta ^2}{4} \wedge  (\eta q)  - \,[2 C_{0}^{\prime} (1+ C_0)  +  200(C_0+1)]\sqrt{\frac{c_T (1+ s_{\tau})\, \log(p c_T\vee n)}{n d_T}}  (\sqrt{n} +\sqrt{d_T})\frac{ \eta  }{\kappa_0} \cdot\\
			 & &  \left( \sqrt{s_{\tau}+1} + \sqrt{r_{\tau}/\log(p c_T\vee n) } \right)-  \frac{\tilde{C}_0 \eta c_T  \phi_{n} \sqrt{\log(p c_T\vee n)}(\sqrt{1+s_{\tau}  }+   \sqrt{ r_{\tau}/ \log(p c_T\vee n) }  )(\sqrt{n} +\sqrt{d_T})  }{  \sqrt{nT} \kappa_{0}   \underline{f}^{1/2} }\\
			& \geq & \displaystyle  \frac{\eta ^2}{4}   -  \frac{2 \eta   \phi_{n} (  C_{0}^{\prime}(1+C_0) +  \tilde{C}_0+200(C_0+1))\sqrt{c_T \,\log(p c_T\vee n)}(\sqrt{n} +\sqrt{d_T})  }{ \sqrt{n d_T} \kappa_{0}   \underline{f}^{1/2} }\cdot\\
			 & &(\sqrt{(1+s_{\tau}) }  +  \sqrt{r_{\tau}  /\log(p c_T\vee n) }     )\\
			& =& 0,
		\end{array}
	\end{equation}
	where the the second inequality follows  from Lemma \ref{lem:empiricalbound}, the third from Lemma \ref{lem:lowerbound},   the fourth from our choice of $\eta$  and (\ref{eqn:signal}), and the equality also from our choice of $\eta$. Hence,  (\ref{eqn:contradiction})   leads to a contradiction  which shows that (\ref{eqn:eta})  cannot happen in the first place. As a result, by Assumption \ref{cond3},
	\[
	\frac{\|\hat{\Pi}(\tau)\, -\, \Pi(\tau)  \|_F}{\sqrt{nT}}   \,\leq \,     \frac{ 1 }{\kappa_0 }\vert J_{\tau}^{1/2} (\hat{\theta}(\tau)\, -\,  \theta(\tau),  \hat{\Pi}(\tau) -\Pi(\tau) ) \vert   \,\leq \,   \frac{ \eta }{\kappa_0},
	\]
	which holds  with probability  approaching one.
	
	To conclude the proof,  let  $\hat{\delta} \,=\,  \hat{\theta} -\theta $ and notice that
	\[
	\begin{array}{lll}
		\displaystyle     \|  \hat{\delta}_{  (T_{\tau} \cup \overline{T}_{\tau}(\hat{\delta},m))^c }  \|^2  & \,\leq\, & \displaystyle \sum_{k \geq m+1}  \frac{  \| \hat{\delta}_{T_{\tau}^c}\|_1^2 }{k^2} \\
		&\,\leq \, &  \displaystyle \frac{\|\hat{\delta}_{T_{\tau}^c}\|_1^2}{m}\\
		&\,\leq \,&  \displaystyle  \frac{4C_0}{m} \left[\|\hat{\delta}_{T_{\tau}}\|_1^2  +  \frac{  r_{\tau} \|\Pi(\tau)  -  \hat{\Pi}(\tau) \|_F^2}{nT \log( p c_T \vee n ) } \right]\\
		&\,\leq \, & \displaystyle \frac{4C_0}{m}\left[  s_{\tau}\, \|\hat{\delta}_{  T_{\tau} \cup \overline{T}_{\tau}(\hat{\delta},m) }   \|^2  + \frac{  r_{\tau} \|\Pi(\tau) -\hat{\Pi}(\tau) \|_F^2}{nT \log( p c_T \vee n ) }   \right],
	\end{array}
	\]
	which implies
	\[
	\begin{array}{lll}\displaystyle  \|\hat{\delta}\|\,&\leq& \, \left( 1+  2C_0 \sqrt{\frac{s_{\tau}}{m}}  \right)\left(  \|\hat{\delta}_{  T_{\tau} \cup \overline{T}_{\tau}(\hat{\delta},m) }   \|  + \frac{  \sqrt{r_{\tau}}\|\Pi(\tau) -\hat{\Pi}(\tau) \|_F}{\sqrt{nT \log( c_T p \vee n )} } \right) \\
		&\leq&      \frac{  J_{{\tau}}^{1/2}(\hat{\delta}, \hat{\Pi}(\tau)\, -\, \Pi(\tau) )  }{\kappa_m   } \left( 1+  2C_0 \sqrt{\frac{s_{\tau}}{m}}  \right),
	\end{array}
	\]
	and the  result  follows.
\end{proof}

\section{Proof  of Theorem   \ref{thm3}  }

\subsection{Auxiliary  lemmas  for proof of Theorem  \ref{thm3}}
\begin{lemma}
	\label{lem:low_rank}
	Suppose  that Assumptions \ref{cond1}--\ref{cond2} and \ref{cond5} hold.
	Let
	$$
	\nu_{1} \,=\,\displaystyle \frac{2}{nT }\sum_{i=1}^n \sum_{t=1}^T     \|X_{i,t}\|_{\infty}.
	$$
	and 
	$$\nu_2 \,=\,  \frac{200 c_T}{n T}\left( \sqrt{n}  +  \sqrt{d_T}  \right).$$
	We have that
	\[
	(   \hat{\Pi}(\tau) -\Pi(\tau)  - X\theta(\tau) ) \in A^{\prime}_{\tau},
	\]
	with probability approaching one, where
	\[
	\begin{array}{l}
		A_{\tau}^{\prime }  = \Bigg\{  \Delta  \in \mathbb{R}^{n \times T } \,:\,   \| \Delta \|_*    \leq     c_0\sqrt{r_{\tau}}\left(    \|\Delta \|_F+   \| \xi\|_* \right),\,\,\, \| \Delta \|_{\infty} \leq c_1   \Bigg\},\\
	\end{array}
	\]
	and $c_0$ and $c_1$ are positive constants that depend  on $\tau$. Furthermore,  $\hat{\theta}(\tau )  = 0$.
\end{lemma}

\begin{proof}
	First, we observe that $C$ in the statement of Theorem \ref{thm3}  can be take as  $C = \| X\theta(\tau) +  \Pi(\tau)  \|_{\infty}$. And so,
	\[
	\left\|   X\theta(\tau) +  \Pi(\tau)   -  \hat{\Pi}(\tau)\right\|_{\infty}  \,\leq \,  2C \,=:\,c_1.
	\]
	
	Next,  notice that  for any $\check{\Pi} \in \mathbb{R}^{n \times  T}$ and  $\check{\theta} \in \mathbb{R}^p \backslash \{0\}$,
	\[
	\begin{array}{lll}
		\displaystyle    \hat{Q}_{\tau}(0, \check{\Pi}  )   -    \hat{Q}_{\tau}( \check{\theta} , \check{\Pi}  )   -   \nu_{1} \|\check{\theta }\|_{1,n,T}     & \leq & 	\displaystyle  \frac{1}{nT}  \sum_{i=1}^n \sum_{t=1}^T  \vert  X_{i,t}^{\prime} \check{\theta}\vert    -  \nu_{1} \|\check{\theta}\|_1\\
		& \leq & 	\displaystyle  \frac{1}{nT}  \sum_{i=1}^n \sum_{t=1}^T  \vert  X_{i,t}^{\prime} \check{\theta}\vert    -  \nu_{1} \|\check{\theta}\|_{1,n,T}  \\
		& \leq  & \displaystyle  \frac{1}{nT}  \sum_{i=1}^n \sum_{t=1}^T     \underset{j =1 \ldots,p}{\max} \left \vert\frac{  X_{i,t,j}  }{ \hat{\sigma}_j  } \right\vert   \| \check{\theta}\|_{1,n,T}    -  \nu_{1} \|\check{\theta}\|_{1,n,T}  \\
		& <&0,
	\end{array}
	\]
	where   the   first inequality follows since  $\rho_{\tau}$ is a contraction  map.  Therefore,   $\hat{\theta}(\tau) = 0$.  Furthermore,    we have
	\[
	\begin{array}{lll}
		0 & \leq & \displaystyle  \hat{Q}_{\tau}( 0,  X\theta(\tau)  +  \Pi(\tau)  )   - \hat{Q}_{\tau}(  0,  \hat{\Pi}(\tau) )  \,+\,  \nu_2\left(  \| X\theta(\tau)  +  \Pi(\tau)  \|_*  - \|\hat{\Pi}(\tau)\|_*  \right)\\
		& \leq&\displaystyle  \left\vert   \frac{1}{nT}    \sum_{ i= 1}^n \sum_{t=1}^T    a_{i,t} \left(  ( X\theta(\tau)  +  \Pi(\tau)) -  \hat{\Pi}(\tau)  \right)  \right\vert    \\
		& &   \,+\, \nu_2\left(  \| X\theta(\tau)  +  \Pi(\tau)  \|_*  - \|\hat{\Pi}(\tau)\|_*  \right)\\
		& \leq&  \displaystyle   \|  X\theta(\tau)  +  \Pi(\tau)  -  \hat{\Pi}(\tau) \|_{*}\left(  \underset{  \| \tilde{\Delta} \|_* \leq 1  }{\sup}\, \left\vert   \frac{1}{nT}   \sum_{i=1}^{n}\sum_{t=1}^{T}  a_{i,t}  \tilde{\Delta}_{i,t} \right\vert        \right)\\
		& &  \,+\, \nu_2\left(  \| X\theta(\tau)  +  \Pi(\tau)  \|_*  - \|\hat{\Pi}(\tau)\|_*  \right)\\
		& \leq&   \displaystyle   \frac{200 c_T}{nT}\left(  \sqrt{n}+ \sqrt{d_T} \right) \left(  \|  X\theta(\tau)  +  \Pi(\tau)  -  \hat{\Pi}(\tau) \|_{*}   +    \|  X\theta(\tau)  +  \Pi(\tau)  \|_{*}  -\| \hat{\Pi}(\tau) \|_*      \right)   \\
		& &  \displaystyle \,\,\,\,-\,\,  \frac{100 c_T}{n T   }\left(  \sqrt{n}+ \sqrt{d_T} \right)   \|  X\theta(\tau)  +  \Pi(\tau)  -  \hat{\Pi}(\tau) \|_{*}\\
		& \leq& \displaystyle   \frac{200 c_T}{nT}\left(  \sqrt{n}+ \sqrt{d_T} \right) \left(  \|  X\theta(\tau)  +  \Pi(\tau)  + \xi  -  \hat{\Pi}(\tau) \|_{*}   +    \|  X\theta(\tau)  +  \Pi(\tau)  + \xi \|_{*}  -\| \hat{\Pi}(\tau) \|_*      \right) \\
		& &\displaystyle   \,\,\,\,+\,\frac{400 c_T}{nT}\left(  \sqrt{n}+ \sqrt{d_T} \right) \|\xi\|_*\,\,-\,\, \frac{100 c_T}{n T   }\left(  \sqrt{n}+ \sqrt{d_T} \right)   \|  X\theta(\tau)  +  \Pi(\tau)  -  \hat{\Pi}(\tau) \|_{*}\\
		& \leq& \displaystyle   \frac{c_1\,c_T}{nT}\left(  \sqrt{n}+ \sqrt{d_T} \right) \sqrt{r_{\tau}} \| X\theta(\tau)  +  \Pi(\tau)  + \xi  -  \hat{\Pi}(\tau)\|_F \\
		& &\displaystyle   \,\,\,\,+\,\frac{400 c_T}{nT}\left(  \sqrt{n}+ \sqrt{d_T} \right) \|\xi\|_*\,\,-\,\, \frac{100 c_T}{n T   }\left(  \sqrt{n}+ \sqrt{d_T} \right)   \|  X\theta(\tau)  +  \Pi(\tau)  -  \hat{\Pi}(\tau) \|_{*}\\
	\end{array}
	\]
	for some  positive  constant  $c_1$,  where the first inequality follows from the optimality  of the estimator, the second as in the proof of  Lemma \ref{lem1},  the third by a basic property of the nuclear norm,  the fourth by Lemma \ref{lem:latent_error},  the fifth by the triangle inequality,   and the six by Assumption \ref{cond5} and Lemma  \ref{lem:nuclear_norm}.
\end{proof}

\begin{lemma}
	\label{lem:empiricalbound2}
	Let
	\[
	\begin{array}{l}
		\epsilon^{\prime}(\eta) \,=\, \underset{   (\delta,\Delta) \in A_{\tau}  \,:\,     J_{\tau}^{1/2}(\delta,\Delta) \leq \eta  }{\sup} \,\,\  \Bigg \vert  \hat{Q}_{\tau}(0,X\theta(\tau) + \Pi(\tau) + \Delta) - \hat{Q}_{\tau}( 0,  X\theta(\tau) + \Pi(\tau)) -\\
		\,\,\,\,\,\,\,	\,\,\,\,\,\,\,\,\,\,\,\,\,\,\, Q_{\tau}(0,X\theta(\tau) + \Pi(\tau) + \Delta) + Q_{\tau}(0, X\theta(\tau) + \Pi(\tau))   \Bigg\vert,
	\end{array}
	\]
	and  $\{\phi_{n}\}$  a sequence  with $\phi_n/ (\sqrt{\underline{f} }\log(c_T+1)    )\,\rightarrow \, \infty$. Then for all $\eta>0$
	\[
	\epsilon^{\prime}(\eta) \,\leq \, \frac{\tilde{C}_0 \eta c_T  \phi_{n} \sqrt{ r_{\tau}}(\sqrt{n} +\sqrt{d_T})  }{  \sqrt{nT} \underline{f}},
	\]
	for some constant $\tilde{C}_0 >0$, 	with probability at least $1-  \alpha_{n}$. Here, the sequence $ \{\alpha_{n}\}$ is independent of $\eta$, and   $\alpha_{n}  \,\rightarrow 0$.
\end{lemma}

\begin{proof}
	This follows similarly to the proof of Lemma \ref{lem:empiricalbound}.
\end{proof}

\begin{lemma}
	\label{lem:lower_bound}	
	Let
	\[
	A_{\tau}^{\prime \prime } \,=\,\left\{    \Delta \in A_{\tau}^{\prime }   \,:\,  q(\Delta )  \geq 2\eta_0  ,  \,\,\,\Delta \neq 0  \right\},
	\]
	with
	\[
	\eta_0   \,=\, \frac{\tilde{C}_1 c_T  \phi_{n} \sqrt{ r_{\tau}}(\sqrt{n} +\sqrt{d_T})  }{  \sqrt{nT}    \underline{f} },
	\]
	for an appropriate  constant $\tilde{C}_1 >0$,
	and
	\[
	q(\Delta) \,=\, \frac{3}{2} \frac{  \underline{f}^{3/2}  }{\bar{f}^{\prime}  }\,\frac{  \left(   \frac{1}{nT}\sum_{i=1}^{n}\sum_{t=1}^{T}  (  \Delta_{i,t}  )^2  \right)^{3/2}   }{     \frac{1}{nT}   \sum_{i=1}^{n}\sum_{t=1}^{T}   \vert \Delta_{i,t} \vert^3        }.
	\]
	Under Assumptions  \ref{cond1}-\ref{cond2} and  \ref{cond5},
	for any   $\Delta \in A_{\tau}^{\prime \prime }$  we have that
	\[
	\displaystyle    Q_{\tau}(0,X\theta(\tau)  +  \Pi(\tau)   +  \Delta   )    -   Q_{\tau}(  0,X\theta(\tau)  +  \Pi(\tau) ) \,\geq \,  \min\left\{\frac{    \underline{f}  \|\Delta\|^2  }{4 nT} ,    \frac{ 2\eta   \underline{f}^{1/2}  \|\Delta\| }{ \sqrt{n T} }   \right \}.
	\]
\end{lemma}
\begin{proof}
	This follows as the proof of Lemma \ref{lem:lowerbound}.
\end{proof}

\subsection{Proof of Theorem \ref{thm3}}

The proof  of Theorem  \ref{thm3}  proceeds by exploiting  Lemmas \ref{lem:low_rank} and \ref{lem:lower_bound}.  By Lemma  \ref{lem:low_rank}, we have that
\begin{equation}
	\label{eqn:event}
	\displaystyle 	\hat{\Delta} \,:=\,\hat{\Pi}(\tau) -X\theta(\tau) -  \Pi(\tau)      \in  A_{\tau}^{\prime },
\end{equation}
with high probability. Therefore,  we assume that (\ref{eqn:event})  holds.  Hence, if $	\hat{\Delta} \notin  A_{\tau}^{\prime \prime}$,
then
\begin{equation}
	\label{eqn:bound1}
	\displaystyle    \frac{1}{ \sqrt{nT} }    \|\hat{\Delta}\|_F\,< \,    \frac{  4\,\eta\,\left( \sum_{i=1}^n \sum_{t=1}^T   \vert \hat{\Delta}_{i,t} \vert^3   \right)     }{3\,\left( \sum_{i=1}^n \sum_{t=1}^T   \hat{ \Delta}_{i,t}^2   \right)  }  \frac{    \overline{f}^{\prime} }{    \underline{f}^{3/2} }\,\leq \,\frac{4      \overline{f}^{\prime}  \|        \hat{ \Delta} \|_{\infty}  \,\eta  }{  3  \underline{f}^{3/2}  }.
\end{equation}

If $	\hat{\Delta} \in  A_{\tau}^{\prime \prime}$,  then  we proceed as in the proof of Theorem \ref{thm:1} by  exploiting Lemma \ref{lem:empiricalbound2}, and treating
$X\theta(\tau) + \Pi(\tau)$ as the latent factors matrix, the design matrix as the matrix zero, $A_{\tau}^{\prime \prime}$ as $A_{\tau}$, and
\[
\kappa_0   \,=\,   \underline{f}^{1/2}, 
\]
in Assumption \ref{cond3}.  This  leads to
\begin{equation}
	\label{eqn:bound2}
	\displaystyle    \frac{1}{ \sqrt{nT} }    \|\hat{\Delta}\|_F\,\leq \,  \eta,
\end{equation}
and the claim in Theorem  \ref{thm3}  follows  combining (\ref{eqn:bound1})  and (\ref{eqn:bound2}).

\subsection{Proof of  Corollary  \ref{cor1} }

First notice that  by  Theorem  \ref{thm:1}  and Theorem 3  in \cite{yu2014useful},
\begin{equation}
	\label{eqn:cor_step1}
	v \,:=\,	\max\left\{ \underset{  O  \in \mathbb{O}_{r_{\tau}}  }{\min}  \|  \hat{g}(\tau) O  - g(\tau)\|_F , \underset{  O  \in \mathbb{O}_{r_{\tau}}  }{\min}  \|  \tilde{\hat{\lambda}}(\tau) O  - \tilde{\lambda}(\tau)\|_F   \right \}   \,=\,     O_{\mathbb{P}}\left(  \frac{   (\sigma_1( \tau)  +
		\sqrt{r_{\tau} }\, \mathrm{Err}   )  \mathrm{Err}}{(\sigma_{r_{\tau} -1}( \tau))^2  - (\sigma_{r_{\tau} }( \tau))^2 }    \right).	
\end{equation}
Furthermore,
\[
\begin{array}{lll}
	\displaystyle \frac{\|   \hat{\lambda}(\tau) -  \lambda(\tau)\|_F^2}{nT}  &  =  &   \displaystyle   \frac{1}{nT} \sum_{j=1}^{ r_{\tau} } \|  \lambda_{\cdot,j}(\tau)  -  \hat{\lambda}_{\cdot,j}(\tau) \|^2  \\
	& \leq&  \displaystyle    \frac{2}{nT}\sum_{j=1}^{ r_{\tau} } (\sigma_j  -\hat{\sigma}_j  )^2  \,+\, \frac{2}{nT}\sum_{j=1}^{ r_{\tau}  } \,\sigma_j^2 \| \tilde{\hat{\lambda}}_j(\tau)  -  \tilde{ \lambda}_j(\tau) \|^2\\
	& \leq&  \displaystyle    \frac{2}{nT}\sum_{j=1}^{ r_{\tau} } (\sigma_j  -\hat{\sigma}_j  )^2  \,+\, \frac{2\,\sigma_1^2 }{nT} \sum_{j=1}^{ r_{\tau} }  \| \tilde{\hat{\lambda}}_j(\tau)  -  \tilde{ \lambda}_j(\tau) \|^2\\
	& \leq  &   \displaystyle  \frac{2\, r_{\tau}}{nT} (  \sigma_1 - \hat{\sigma}_1 )^2  \,+\, \frac{2\, \sigma_1^2 }{nT}\| \tilde{\hat{\lambda}}(\tau) -  \tilde{ \lambda}(\tau) \|_F^2\\ 
	& \leq &  \displaystyle  \frac{2\, r_{\tau}}{nT}   \|  \Pi(\tau)  \,-\, \hat{\Pi}(\tau)   \|_F^2  \,+\, \frac{2\, \sigma_1^2 }{nT}\| \tilde{\hat{\lambda}}(\tau) -  \tilde{ \lambda}(\tau) \|_F^2\\
	& =&  \displaystyle O_{ \mathbb{P} } \left(      \frac{   r_{\tau} \phi_{n}^2  c_T  \log(p c_T \vee n  ) (1+s_{\tau}  +   r_{\tau}/  \log(p c_T \vee n  )   )  }{\kappa_{0}^4 \, \, \underline{f} }\left( \frac{1}{n} +\frac{1}{d_T} \right )    \right) \,+\, \\
	&    &  \displaystyle O_{\mathbb{P}}\left( \frac{ \sigma_1^2}{nT}     \frac{   (\sigma_1( \tau)  +
		\sqrt{r_{\tau} }\, \mathrm{Err}  )^2  \mathrm{Err} ^2 }{  \left( (\sigma_{r_{\tau} -1}( \tau))^2  - (\sigma_{r_{\tau} }( \tau))^2  \right)^2  }       \right),
\end{array}
\]
where  the third  inequality  follows from  Weyl's inequality, and the last one  from (\ref{eqn:cor_step1}).

\section{Proof  of Theorem \ref{thm_c} }

Conditioning on $\Pi$ and $\{X_{i,t}\}$, 
as in \cite{yu1994rates},  we  define the  sequence $\{\tilde{\varepsilon}_{i,t} \}_{i \in [n], t \in [T]}$
such  that  
\begin{itemize}
	\item   $\{\tilde{\varepsilon}_{i,t}\}_{i \in [n], t \in [T]}$  is independent  of   $\{\varepsilon_{i,t}\}_{i \in [n], t \in [T]}$; 
	\item  for a fixed  $t$ the random variables  $\{\tilde{\varepsilon}_{i,t} \}_{i \in [n]}$  are independent;
	\item  for a fixed $i$:
	\[
	\mathcal{L}(  \{ \tilde{\varepsilon}_{i,t}\}_{ t \in H_l   } ) \,=\,  \mathcal{L}(  \{ \varepsilon_{i,t} \}_{ t \in H_l   } ) \,=\, \mathcal{L}(  \{ \varepsilon_{i,t} \}_{ t \in H_1  } ) \,\,\,\,\forall  l \,\,\in [d_T],
	\]
	and  the blocks $ \{ \tilde{\epsilon}_{i,t} \}_{ t \in H_1   },\ldots,  \{ \tilde{\epsilon}_{i,t} \}_{ t \in H_{d_T}   }$ are independent.
\end{itemize}

Here,  we  define  $ \Lambda  \,:=\,   \{ H_1,H_1^{\prime},\ldots,H_{d_T},H_{d_T}^{\prime},R  \}$ with
\begin{equation}
	\label{eqn:intervals}
	\begin{array}{lll}
		H_j &=&  \left\{   t \,:\,    1+ 2(j-1)c_T \leq t \leq (2j-1)c_T \right\},\,\,\,	 \\
		H_j^{\prime} &=&  \left\{   t \,:\,    1+ (2j-1)c_T \leq t \leq 2j c_T \right\}, \,\,j = 1,\ldots, d_T,\\
		&&  \,\,\,\,\text{and}\,\,\, R \,=\,\{ t \,:\,    2c_T d_T +1\,\leq t \leq T  \}.
	\end{array}
\end{equation}
We then let
\[
\tilde{Y}_{i,t}\,  = \, X_{i,t}^{\prime} \theta(\tau)+\Pi_{i,t}(\tau) +  \tilde{\varepsilon}_{i,t} G(X_{i,t})
\]
for all $i\in [n]$ and $t\in [T]$.

Furthermore, we define the scores $a_{i,t} \,=\,  \tau -  1\{ Y_{i,t}  \leq  X_{i,t}^{\prime} \theta(\tau) + \Pi_{i,t}(\tau)  \} $, and   $\tilde{a}_{i,t} \,=\,  \tau -  1\{ \widetilde{Y}_{i,t}  \leq  X_{i,t}^{\prime} \theta(\tau) + \Pi_{i,t}(\tau)  \} $.


\subsection{Auxiliary lemmas for proof of Theorem \ref{thm_c} }

Throughout  we use the notation  from Section \ref{sec:proofs} and  $(\hat{\beta}(\tau), \hat{\Pi}(\tau) )$ denotes the estimator defined in (\ref{eqn:prob_1}).

We begin by defining 
\[
\displaystyle  \hat{M}( \tilde{\beta},\tilde{ \Pi}  ) \,=\,  \frac{1}{nT}\sum_{i=1}^n \sum_{t=1}^{T} \left[ \rho_{\tau}(Y_{i,t}  - X_{i,t}^{\prime} \tilde{\beta} - \tilde{ \Pi}_{i,t}     )  - \rho_{\tau}(Y_{i,t}  - X_{i,t}^{\prime} \beta(\tau) - \Pi_{i,t}(\tau)     )  \right]
\]	
and 
\[
\displaystyle  M(\tilde{\beta},\tilde{ \Pi})   \,=\,  \mathbb{E}\left(  \hat{M}( \tilde{\beta},\tilde{ \Pi}  )   \big|  \{ X_{i,t}\}  \right).
\]
We also set
\begin{equation}
	\label{mat_e}
	\mathcal{E} \,:=\, \left\{ \underset{  1\leq  j\leq p }{\max}\, \,\vert\hat{\sigma}^2_{j}-1 \vert \,\leq \, \frac{1}{4}\right\},
\end{equation}
with the notation in Assumption 2.

We   start by recalling a result from \cite{padilla2020adaptive} involving the 	behavior of  $M$ locally around the true quantiles.

\begin{lemma}[Lemma  13 in  \cite{padilla2020adaptive}]
	\label{lem10}
	With the notation and  assumptions of  Theorem \ref{thm_c} we have that  
	\[
	\displaystyle  M(\hat{\beta},\hat{ \Pi})   \,\geq \,  \frac{c_0}{nT}  \sum_{i=1}^n \sum_{t=1}^{T}\min\{   \vert   q_{i,t} -\hat{q}_{i,t}  \vert,(q_{i,t} -\hat{q}_{i,t} )^2    \} ,
	\]
	for some constant  $c_0>0$.
\end{lemma}
We now proceed to construct  a restricted set $K$ where the solution  $(\hat{\beta}(\tau),\hat{\Pi}(\tau)) $ lies with high probability.

\begin{lemma}
	\label{lem12}
	Let  
	\[
	\phi_{n,T} \,=\, \|\beta(\tau)\|_1      +     \frac{1}{\sqrt{  nT \log(  \max\{n,pc_T \} )    }}     \|\Pi(\tau)\|_*,
	\]
	and
	\[
	\psi_{n,T}\,=\, \sqrt{  nT \log(  \max\{n,pc_T \} )    }   \|\beta(\tau)\|_1 \,+\,\|\Pi(\tau)\|_*.
	\]
	Then there exists a positive constant  $C_0$  such that the event  
	\[
	\mathcal{B} \,=\,   \left\{  (\hat{\beta}(\tau),\hat{\Pi}(\tau)) \in K    \right\}
	\]
	holds with probability at least 
	\[
	1 -  \gamma - \frac{16}{ n } - 8npT\left( \frac{1}{c_T} \right)^{\mu}- 2 n T  \left( \frac{1}{c_T} \right)^{\mu} -2c_1\exp(-c_2 \max\{n,T\} + \log c_T ),
	\]
	with the notation from  Lemma \ref{lem1}, and where
	\[
	K\,:=\,\left\{      (\tilde{\beta},\tilde{ \Pi}) \,:\,      \|\tilde{\beta}\|_{1}\,\leq \, C_0 \phi_{n,T},\,\,\,\,   \|\tilde{\Pi}\|_*\,\leq \, C_0 \psi_{n,T} \right\}
	\]
	and provided that
	\[
	\nu_1 \,=\, 18\sqrt{       \frac{   c_T \log( \max\{n,p c_T\} )  }{n d_T}  } \left( \sqrt{n}  + \sqrt{d_T}   \right),
	\]
	and
	\[
	\nu_2\,=\,\frac{400 c_T }{n T } \left( \sqrt{n}  + \sqrt{d_T}   \right).
		\]
\end{lemma}
\begin{proof}
	First notice that  by the proof Lemma \ref{lem1} we have that
	\begin{equation}
		\label{eqn:basic}
		\begin{array}{lll}
			0 &\leq  &\displaystyle \underset{1 \leq j \leq p}{\max}\left \vert  \frac{1}{nT}  \sum_{i=1}^{n} \sum_{t=1}^T  \frac{X_{i,t,j}  a_{i,t} }{\hat{\sigma}_j }    \right\vert \left[ \sum_{k=1}^{p}\hat{\sigma}_k\vert  \theta_k(\tau) - \hat{\theta}_k(\tau)  \vert   \right] \,+\,
			\\
			& & \nu_{1}\left(\|\theta(\tau)\|_{1,n,T}  - \|\hat{\theta}(\tau)\|_{1,n,T}    \right)  \,+\,\\
			& & \displaystyle \left\vert   \frac{1}{nT}   \sum_{i=1}^{n}\sum_{t=1}^{T}  a_{i,t} (\Pi_{i,t}(\tau)- \hat{\Pi}_{i,t}(\tau))  \right\vert   \,+\, \nu_2 ( \|  \Pi(\tau)  \|_*   -  \| \hat{\Pi}(\tau) \|_*  ).    \\
		\end{array}
	\end{equation}
	Therefore,
	\[
	\begin{array}{lll}
		\displaystyle \frac{\nu_1}{2}\|\hat{\theta}(\tau)\|_{1,n,T}   +   \frac{\nu_2}{2} \| \hat{\Pi}(\tau) \|_*     & \leq& \displaystyle \displaystyle \underset{1 \leq j \leq p}{\max}\left \vert  \frac{1}{nT}  \sum_{i=1}^{n} \sum_{t=1}^T  \frac{X_{i,t,j}  a_{i,t} }{\hat{\sigma}_j }    \right\vert \left[ \sum_{k=1}^{p}\hat{\sigma}_k\vert  \theta_k(\tau) - \hat{\theta}_k(\tau)  \vert   \right] \,+\,   \\  & &\nu_{1}\left(\|\theta(\tau)\|_{1,n,T}  - \frac{1}{2}\|\hat{\theta}(\tau)\|_{1,n,T}    \right)   \\
		& &\displaystyle  \,+\,     \left\vert   \frac{1}{ \nu_1  nT}   \sum_{i=1}^{n}\sum_{t=1}^{T}  a_{i,t} (\Pi_{i,t}(\tau)- \hat{\Pi}_{i,t}(\tau))  \right\vert  \,+\, \nu_2 ( \|  \Pi(\tau)  \|_*   -  \frac{1}{2} \| \hat{\Pi}(\tau) \|_*  ) \\
		&\leq &  \displaystyle 9 \sqrt{       \frac{   c_T \log( \max\{n,p c_T\} )  }{n d_T}  } \left[ \sum_{k=1}^{p}\hat{\sigma}_k\vert  \theta_k(\tau) - \hat{\theta}_k(\tau)  \vert   \right]\,+\, 
		\\  & &\displaystyle \nu_{1}\left(\|\theta(\tau)\|_{1,n,T}  - \frac{1}{2}\|\hat{\theta}(\tau)\|_{1,n,T}    \right)   \,+\, \frac{200 c_T }{n T } \left( \sqrt{n}  + \sqrt{d_T}   \right)\|  \Pi(\tau)  -  \hat{\Pi}(\tau)  \|_* \\
		& &  \displaystyle     \,+\, \nu_2 ( \|  \Pi(\tau)  \|_*   -  \frac{1}{2} \| \hat{\Pi}(\tau) \|_*  ) \\
		&\leq & \displaystyle 27\sqrt{       \frac{   c_T \log( \max\{n,p c_T\} )  }{n d_T}  }(\sqrt{n} + \sqrt{d_T})   \|  \theta(\tau)\|_{1,n,T}  \\
		& & \displaystyle \,+\,   \frac{600 c_T }{n T } \left( \sqrt{n}  + \sqrt{d_T}   \right)  \|  \Pi(\tau)  \|_*  \\
		& &
	\end{array}
	\]
	where the second inequality holds by Lemma \ref{lem:penalty_tuning}, Lemma \ref{lem:latent_error}, and Assumption \ref{cond2},
	with probability at least 
	\[
	1 -  \gamma - \frac{16}{ n } - 8npT\left( \frac{1}{c_T} \right)^{\mu}- 2 n T  \left( \frac{1}{c_T} \right)^{\mu} -2c_1\exp(-c_2 \max\{n,T\} + \log c_T ),
	\]
	and third inequality holds by triangle inequality. 	The claim then follows.
\end{proof}

Next we combine the previous two results to arrive at an upper bound on the estimation error of the quantiles. 
\begin{lemma}
	\label{lem11}
	Let  $\eta$ such that
	\[
	\displaystyle  \eta  \,>\,  \frac{3 \nu_2}{c_0}\|\Pi(\tau) \|_* 
	\]
	with $c_0$ as in Lemma \ref{lem10}.   Then 
	\begin{equation}
		\label{eqn:uper}\begin{array}{l}
			\displaystyle   \mathbb{P}\left(\frac{1}{nT}  \sum_{i=1}^n \sum_{t=1}^{T}\min\{   \vert   q_{i,t} -\hat{q}_{i,t}  \vert,(q_{i,t} -\hat{q}_{i,t} )^2    \}  \geq 
			\eta\right)   \\
			\displaystyle       \,\, \leq  \,  \mathbb{P}\left(  \left\{    \underset{(\tilde{\beta},\tilde{ \Pi}) \in K  }{\sup}\left[   M(\tilde{\beta},\tilde{ \Pi})     - \hat{M}(\tilde{\beta},\tilde{ \Pi})     \right] \geq \frac{c_0 \eta}{3} \right\} \cap   \mathcal{E} \cap \mathcal{B} \right)   \,+\,  \mathbb{P}\left(  \left\{\frac{\nu_1}{c_0} \|\beta(\tau) \|_{1,n,T} \geq   \frac{c_0 \eta}{3}\right\} \cap \mathcal{E} \right) \\
			\displaystyle 	\,\, \,+\, \gamma \,+\, \mathbb{P}(\mathcal{B}^c),
		\end{array}
	\end{equation}
	with $\mathcal{E}$ as in (\ref{mat_e}),  $\mathcal{B}$ as in Lemma \ref{lem12}, and $\gamma$ as in Assumption \ref{cond2}.
\end{lemma}

	\begin{proof}
	First, 	notice that
	\[
	\begin{array}{lll}
		\displaystyle  \frac{1}{nT}  \sum_{i=1}^n \sum_{t=1}^{T}\min\{   \vert   q_{i,t} -\hat{q}_{i,t}  \vert,(q_{i,t} -\hat{q}_{i,t} )^2    \}    & \leq&  	\displaystyle c_0^{-1}M( \hat{\beta}(\tau),\hat{\Pi}  )  \\ 
		&=& 	\displaystyle  c_0^{-1}\left[  M( \hat{\beta}(\tau),\hat{\Pi}  )  -  \hat{M}( \hat{\beta}(\tau),\hat{\Pi}  )     +   \hat{M}( \hat{\beta}(\tau),\hat{\Pi}  )   \right]\\
		& \leq & 	\displaystyle  c_0^{-1} \left[  M( \hat{\beta}(\tau),\hat{\Pi}  )  -  \hat{M}( \hat{\beta}(\tau),\hat{\Pi}  )     \right] \,+\,\\
		& & \displaystyle \nu_1  c_0^{-1} \left[ \|\beta(\tau) \|_{1,n,T}  -\|\hat{\beta}(\tau) \|_{1,n,T}      \right] \,+\,\\
		& & \displaystyle      \nu_2  c_0^{-1} \left[ \|\Pi(\tau) \|_{*}  -\|\hat{\Pi}(\tau) \|_{*}      \right] \\
		& \leq& \displaystyle  c_0^{-1}\left[  M( \hat{\beta}(\tau),\hat{\Pi}  )  -  \hat{M}( \hat{\beta}(\tau),\hat{\Pi}  )     \right] \,+\,\\
		& &  \displaystyle  \nu_1  c_0^{-1} \|\beta(\tau) \|_{1,n,T}\,+\,  \frac{\nu_2}{c_0}\|\Pi(\tau)\|_*,
	\end{array}
	\]
	where  the first inequality follows from Lemma \ref{lem10} and the second by optimality of  $(\hat{\beta}(\tau),\hat{\Pi}(\tau))$. Therefore,  by Lemma \ref{lem12} and Assumption \ref{cond2},
	\[
	\begin{array}{l}
		\displaystyle   \mathbb{P}\left(\frac{1}{nT}  \sum_{i=1}^n \sum_{t=1}^{T}\min\{   \vert   q_{i,t} -\hat{q}_{i,t}  \vert,(q_{i,t} -\hat{q}_{i,t} )^2    \}  \geq 
		\eta\right)   \\
		\displaystyle       \,\, \leq  \,  \mathbb{P}\left(  \left\{    \underset{(\tilde{\beta},\tilde{ \Pi}) \in K  }{\sup}\left[   M(\tilde{\beta},\tilde{ \Pi})     - \hat{M}(\tilde{\beta},\tilde{ \Pi})     \right] \geq \frac{c_0 \eta}{3} \right\} \cap   \mathcal{E} \cap \mathcal{B} \right)   \,+\,  \mathbb{P}\left(  \left\{\frac{\nu_1}{c_0} \|\beta(\tau) \|_{1,n,T} \geq    \frac{c_0 \eta}{3}\right\} \cap \mathcal{E} \right) \\
		\displaystyle 	\,\, \,+\, \mathbb{P}(\mathcal{E}^c)   \,+\, \mathbb{P}(\mathcal{B}^c),
	\end{array}
	\]
	and the claim follows.
	
\end{proof}

	We now proceed to give an  upper bound on the second   term in  the right hand side of (\ref{eqn:uper}).

\begin{lemma}
	\label{lem13}
	It holds that 
	\[
	\mathbb{P}\left(  \left\{\frac{\nu_1}{c_0} \|\beta(\tau) \|_{1,n,T} \geq    \frac{c_0 \eta}{3}\right\} \cap \mathcal{E} \right) \,= \,0,
	\]
	provided that 
	\[
	\eta  \,>\,  \frac{15\nu_1}{4c_0} \|\beta(\tau) \|_{1}. 
	\]
\end{lemma}
 
\begin{proof}
	The claim follows since in the event  $\mathcal{E}$ it holds  that  $\|\beta(\tau) \|_{1,n,T}  \leq  5\|\beta(\tau) \|_{1}/4$.
\end{proof}

	We now proceed to control the first term in the right hand side of  (\ref{eqn:uper}).

\begin{lemma}
	\label{lem14}
	With the notation from before we have that 
	\[
	\begin{array}{lll}
		\mathbb{P}\left(  \left\{    \underset{(\tilde{\beta},\tilde{ \Pi}) \in K  }{\sup}\left[   M(\tilde{\beta},\tilde{ \Pi})     - \hat{M}(\tilde{\beta},\tilde{ \Pi})     \right] \geq \frac{c_0 \eta}{3} \right\} \cap   \mathcal{E} \cap \mathcal{B} \right) & \leq& \displaystyle \frac{\tilde{C} }{\eta}\left[ \frac{   \phi_{n,T} \sqrt{c_T\log p}   }{\sqrt{ nd_T}}  \, +\,\frac{\psi_{n,T} (  \sqrt{n} +\sqrt{d_T} )}{  nd_T  } \right]\\
		& &\,+\, 2 n T  \left( \frac{1}{c_T} \right)^{\mu},
	\end{array}
	\]
	for a positive constant $\tilde{C}$.
\end{lemma}
 
\begin{proof}
	First, we notice that
	\[
	\begin{array}{l}
		\mathbb{P}\left(  \left\{    \underset{(\tilde{\beta},\tilde{ \Pi}) \in K  }{\sup}\left[   M(\tilde{\beta},\tilde{ \Pi})     - \hat{M}(\tilde{\beta},\tilde{ \Pi})     \right] \geq \frac{c_0 \eta}{3} \right\} \cap   \mathcal{E} \cap \mathcal{B} \right)\\
		\, \leq\,    \mathbb{P}\left(  \left\{\underset{(\tilde{\beta},\tilde{ \Pi}) \in K  }{\sup}\left[   M(\tilde{\beta},\tilde{ \Pi})     - \hat{M}(\tilde{\beta},\tilde{ \Pi})     \right] \geq \frac{c_0 \eta}{3}  \right\}  \,\bigg| \mathcal{E}\right)   \mathbb{P}(\mathcal{E}). \\
	\end{array}
	\]
	Next let  
	\[
     \begin{array}{lll}
     	U_1(X)& :=& \mathbb{P}\left(  \underset{(\tilde{\beta},\tilde{ \Pi}) \in K  }{\sup}\left[   M(\tilde{\beta},\tilde{ \Pi})     - \hat{M}(\tilde{\beta},\tilde{ \Pi})     \right] \geq \frac{c_0 \eta}{3}    \bigg|   \{X_{i,t}\},\mathcal{E} \right).
     \end{array}
	\]
	
	Using the notation from Section \ref{sec:proofs}, we define
	$t_{l,m} = 2c_T l + m$, 
	for  $l=1,\ldots,d_T -1$ and  $m=1,\ldots,c_T$. We also set  
	\[
	\begin{array}{lll}
		Z_{i,t }(\tilde{\beta},\tilde{ \Pi})    &=&  \rho_{\tau}\left( \tilde{Y}_{i,t} - X_{i,t}^{\prime } \tilde{\beta} - \tilde{ \Pi}_{i,t} \right) -  \rho_{\tau}\left( \tilde{Y}_{i,t} - X_{i,t}^{\prime } \beta(\tau) -  \Pi_{i,t}(\tau) \right)\\
		&=&  \rho_{\tau}\left(    \varepsilon_{i,t}G(X_{i,t})   +  X_{i,t}^{\prime } (\beta-\tilde{\beta})   +  (\Pi_{i,t}- \tilde{\Pi}_{i,t} )  \right) -  \rho_{\tau}\left( \varepsilon_{i,t}G(X_{i,t}) \right).
	\end{array}
		\]
		Hence by Assumption \ref{cond1}, conditioning on $X_{i,t}$ and $\Pi_{i,t}$, $	Z_{i,t }(\tilde{\beta},\tilde{ \Pi})  $  belongs to the sigma algebra generated by $\varepsilon_{i,t}$.
	
	Then by Lemma 4.3 from \cite{yu1994rates},
	\[
	\begin{array}{l}
		U_1(X)\\
		\,  \leq \,\displaystyle   2\mathbb{P}\left( \frac{   C_1}{c_T d_T n}  \sum_{m=1}^{c_T}  \underset{ (\tilde{\beta},\tilde{ \Pi}) \in K   }{\sup} \left\{    \sum_{i=1}^{n} \sum_{l=1}^{d_T-1}     \,\,\,\,\,\,\left[\mathbb{E}\left(Z_{i,t_{l,m}}(\tilde{\beta},\tilde{ \Pi})  \,\big|   \{X_{i,t}\}\right)  - Z_{i,t_{l,m}}(\tilde{\beta},\tilde{ \Pi}) \right]\right\}   \geq  \frac{c_0 \eta}{9}    \,\bigg|   \{X_{i,t}\}, \mathcal{E}\right) \\
		\,\,	\,\,\,\,	\,+\,
		\displaystyle   \mathbb{P}\left(\frac{   C_1}{c_T  d_T n}\sum_{t^{\prime} \in R}^{ }   \underset{ (\tilde{\beta},\tilde{ \Pi}) \in K   }{\sup} \left\{    \sum_{i=1}^{n} \left[ \mathbb{E}\left(Z_{i, t^{\prime}   }(\tilde{\beta},\tilde{ \Pi}) \,\big|   \{X_{i,t}\}\right)  \,-\, Z_{i, t^{\prime}  }(\tilde{\beta},\tilde{ \Pi}) \right]     \right\}   \geq   \frac{c_0 \eta}{9}   \,\bigg|   \{X_{i,t}\}, \mathcal{E} \right)\\
		\,\,\,\,\,	\,+\, 2 n T  \left( \frac{1}{c_T} \right)^{\mu}
	\end{array}
	\]
	for a constant $C_1>0$. Hence, by  Markov's inequality and Lemma 2.3.1 in \cite{wellner2013weak} (symmetrization), we have for  $\{ \xi_{i,t}\}$ independent Rademacher variables 
	with $\{ \xi_{i,t}\}  \,\perp \!\!\! \perp\,  \{\tilde{Y}_{i,t}\}   \,\big| \,  \{X_{i,t}\}$ that 
	\begin{equation}
		\label{eqn:upper_bound}
		\begin{array}{l}
			U_1(X)  \\
			\leq \,\displaystyle    \frac{18}{c_0 \eta} \frac{   C_1}{c_T d_T n}   \sum_{m=1}^{c_T}\mathbb{E}\left( \underset{ (\tilde{\beta},\tilde{ \Pi}) \in K   }{\sup} \left\{    \sum_{i=1}^{n} \sum_{l=1}^{d_T-1}     \left[\mathbb{E}\left(Z_{i,t_{l,m}}(\tilde{\beta},\tilde{ \Pi})  \,\bigg|   \{X_{i,t}\}\right)  - Z_{i,t_{l,m}}(\tilde{\beta},\tilde{ \Pi}) \right]\right\}  \,\bigg|   \{X_{i,t}\},\mathcal{E} \right) \\
			\displaystyle  \,\,	\,\,\,\,	\,+\,\frac{9}{c_0 \eta} \frac{   C_1}{c_T d_T n}\sum_{ t^{\prime} \in R   }^{ } \mathbb{E}\left(  \underset{ (\tilde{\beta},\tilde{ \Pi}) \in K   }{\sup} \left\{    \sum_{i=1}^{n} \left[ \mathbb{E}\left(Z_{i, t^{\prime}   }(\tilde{\beta},\tilde{ \Pi}) \,\bigg|   \{X_{i,t}\}\right)  \,-\, Z_{i, t^{\prime}  }(\tilde{\beta},\tilde{ \Pi}) \right]     \right\}  \,\bigg|   \{X_{i,t}\},\mathcal{E} \right)\\
			\,\,	\,\,\,\,	\,+\,2 n T  \left( \frac{1}{c_T} \right)^{\mu}\\
			\displaystyle  \leq   \,\,	\,\,\,\,\frac{36}{c_0 \eta} \frac{   C_1}{c_T d_T n}  \sum_{m=1}^{c_T}  \mathbb{E}\left( \underset{ (\tilde{\beta},\tilde{ \Pi}) \in K   }{\sup} \left\{    \sum_{i=1}^{n} \sum_{l=1}^{d_T-1}    \xi_{i,t_{l,m}}  Z_{i,t_{l,m}}(\tilde{\beta},\tilde{ \Pi}) \right\}   \,\bigg|   \{X_{i,t}\},\mathcal{E}\right) \\
			\displaystyle \,\,	\,\,\,\,	\,+\,\frac{18}{c_0 \eta} \frac{   C_1}{c_T d_T n}\sum_{ t^{\prime}\in R }^{} \mathbb{E}\left(  \underset{ (\tilde{\beta},\tilde{ \Pi}) \in K   }{\sup} \left\{    \sum_{i=1}^{n}  \xi_{i, t^{\prime}} Z_{i, t^{\prime}  }(\tilde{\beta},\tilde{ \Pi})      \right\}   \,\bigg|   \{X_{i,t}\},\mathcal{E}\right) \\
			\,\,	\,\,\,\,	\,+\, 2 n T  \left( \frac{1}{c_T} \right)^{\mu},\\
		\end{array}
	\end{equation}
	Therefore,  from Ledoux-Talagrand's inequality
	\begin{equation}
		\label{eqn:upper_bound2}
		\begin{array}{lll}
			U_1(X)	 &\leq &\displaystyle   \frac{36}{c_0 \eta} \frac{   C_1}{c_T d_T n}\cdot\\
			& &  \displaystyle  \sum_{m=1}^{c_T}\mathbb{E}\left( \underset{ (\tilde{\beta},\tilde{ \Pi}) \in K   }{\sup} \left\{    \sum_{i=1}^{n} \sum_{l=1}^{d_T-1}    \xi_{i,t_{l,m}}\left(   X_{i,t_{l,m}}^{\prime}(  \tilde{\beta} -  \beta(\tau)  ) +  \tilde{ \Pi}_{i,t_{l,m}} - \Pi_{i,t_{l,m}}(\tau)       \right)   \right\}   \,\bigg|   \{X_{i,t}\},\mathcal{E}\right) \\
			& &
			\displaystyle    \,+\, \frac{18}{c_0 \eta} \frac{   C_1}{c_T d_T n}\sum_{t^{\prime} \in R}^{} \mathbb{E}\left(  \underset{ (\tilde{\beta},\tilde{ \Pi}) \in K   }{\sup} \left\{    \sum_{i=1}^{n}  \xi_{i, t^{\prime}}  \left(   X_{i, t^{\prime}}^{\prime}(  \tilde{\beta} -  \beta(\tau)  ) +  \tilde{ \Pi}_{i, t^{\prime}} - \Pi_{i, t^{\prime}}(\tau)       \right)      \right\} \,\bigg|   \{X_{i,t}\},\mathcal{E}  \right) \\
			& &\displaystyle  \,+\,  2 n T  \left( \frac{1}{c_T} \right)^{\mu}\\
			&  =  & \displaystyle   U_2(X)  +  U_3(X)   +  2 n T  \left( \frac{1}{c_T} \right)^{\mu}
		\end{array}
	\end{equation}
	Next we proceed to bound  $U_2(X)$ and $U_3(X)$. To bound $U_2(X)$ notice that for some positive constants  $C$ and $C_3$ we have that 
	\begin{equation}
		\label{eqn:upper_bound3}
		\begin{array}{lll}
			\displaystyle  U_2(X)&\leq &\displaystyle   \frac{36}{c_0 \eta} \frac{   C_1}{c_T d_T n}\,  \sum_{m=1}^{c_T}\mathbb{E}\left( \underset{    (\tilde{\beta},\tilde{ \Pi}) \in K   }{\sup} \left\{    \sum_{i=1}^{n} \sum_{l=1}^{d_T-1}    \xi_{i,t_{l,m}}\left(   X_{i,t_{l,m}}^{\prime}(  \tilde{\beta} -  \beta(\tau)  )     \right)   \right\}   \,\bigg|   \{X_{i,t}\},\mathcal{E}\right)\\
			& &   \displaystyle    \,+\, \frac{18}{c_0 \eta} \frac{   C_1}{c_T d_T n}\,  \sum_{m=1}^{c_T}\mathbb{E}\left( \underset{    (\tilde{\beta},\tilde{ \Pi}) \in K   }{\sup} \left\{    \sum_{i=1}^{n} \sum_{l=1}^{d_T-1}    \xi_{i,t_{l,m}}\left(    \tilde{ \Pi}_{i,t_{l,m}} - \Pi_{i,t_{l,m}}(\tau)   \right)   \right\}   \,\bigg|   \{X_{i,t}\},\mathcal{E}\right)\\ 
			&\leq  &\displaystyle   \frac{36}{c_0 \eta} \frac{   C_1}{c_T d_T n}\,  \sum_{m=1}^{c_T} 2C_0 \phi_{n,T}\mathbb{E}\left( \underset{    \tilde{\beta}  \,:\, \| \tilde{\beta} \|_1 \leq 1  }{\sup} \left\{    \sum_{i=1}^{n} \sum_{l=1}^{d_T-1}    \xi_{i,t_{l,m}}\,   X_{i,t_{l,m}}^{\prime}\tilde{\beta}      \right\}   \,\bigg|   \{X_{i,t}\},\mathcal{E}\right)\,+\,\\
			& & \displaystyle    \,+\, \frac{18}{c_0 \eta} \frac{   C_1}{c_T d_T n}\,  \sum_{m=1}^{c_T}\mathbb{E}\left(     2 \left\|\{\xi_{i,t_{l,m}}\}_{i \in [n], l \in [d_T-1]  } \right\|_2 \,\underset{ \tilde{\Pi} \,:\,  \|\tilde{ \Pi}\|_* \leq  C_0 \psi_{n,T}   }{\sup}  \|\tilde{ \Pi}\|_*  \,\bigg|   \{X_{i,t}\},\mathcal{E}\right)\\ 
			& \leq &\displaystyle   \frac{36}{c_0 \eta} \frac{   C_1}{c_T d_T n}\,  \sum_{m=1}^{c_T} 2C_0 \phi_{n,T}\mathbb{E}\left(    C \sqrt{\log p} \underset{j=1,\ldots,p}{\max} \sqrt{  \sum_{i=1}^{n}  \sum_{l=1}^{d_T-1}     X_{i,t_{l,m},j}^2 } \,\bigg|   \{X_{i,t}\},\mathcal{E}\right)\,+\,\\
			& &\displaystyle    \,+\, \frac{18}{c_0 \eta} \frac{   C_1}{c_T d_T n}\,  \sum_{m=1}^{c_T} 2C_0 \psi_{n,T}  C \left(  \sqrt{n}  + \sqrt{d_T} \right)\\
			& \leq & \displaystyle \frac{C_3}{\eta}\left[ \frac{   \phi_{n,T} \sqrt{c_T\log p}   }{\sqrt{ nd_T}}  \, +\,\frac{\psi_{n,T} (  \sqrt{n} +\sqrt{d_T} )}{  nd_T  } \right],
		\end{array}
	\end{equation}
	where the third inequality follows from  the proof of Theorem 2.4 in \cite{rigollet2015high} and  Theorem  3.4 from \cite{chatterjee2015matrix}, and the fourth inequality follows from the definition of $\mathcal{E}$.
	
	Similarly, 
	\begin{equation}
		\label{eqn:upper_bound4}
		\begin{array}{lll}
			U_3(\tilde{X}) & \leq& \displaystyle \frac{C_3}{\eta}\left[ \frac{   \phi_{n,T} \sqrt{c_T\log p}   }{n\sqrt{ d_T}}  \, +\,\frac{\psi_{n,T} (  \sqrt{n} +1)}{  nd_T  }\right].
		\end{array}
	\end{equation}
	Combining (\ref{eqn:upper_bound2}),  (\ref{eqn:upper_bound3}) and (\ref{eqn:upper_bound4}) the claim follows.
	
\end{proof}

	\subsection{Proof of Theorem \ref{thm_c}}

The proof follows from  Lemmas \ref{lem11}--\ref{lem14} by setting
\[
\eta\,\asymp\, m_n \left[\frac{   \phi_{n,T} \sqrt{c_T\log p}   }{n\sqrt{ d_T}}  \, +\,\frac{\psi_{n,T} (  \sqrt{n} +1)}{  nd_T  }\right],
\]
for any sequence  $m_n$  satisfying $m_{n} \rightarrow \infty$.

\newpage	

\begin{center}
	\begin{longtable}{llp{18.855em}ccc} 
		\caption{Firm Characteristics Construction.\label{Table-names}}\\
		\toprule
		Characteristics  & Name  & \multicolumn{1}{l}{Construction}\tabularnewline
		\hline
		acc  & Working capital accruals  & Annual income before extraordinary items (ib) minus operating cash
		flows (oancf) divided by average total assets (at) \tabularnewline
		agr & Asset growth  & \multicolumn{1}{l}{Annual percent change in total assets (at)}\tabularnewline
		beta  & Beta  & Estimated market beta from weekly returns and equal weighted market
		returns for 3 years \tabularnewline
		bm & Book-to-market  & Book value of equity (ceq) divided by end of fiscal year-end market
		capitalization \tabularnewline
		chinv  & Change in inventory  & Change in inventory (inv) scaled by average total assets (at) \tabularnewline
		chmom & Change in 6-month momentum  & Cumulative returns from months t-6 to t-1 minus months t-12 to t-7 \tabularnewline
		dolvol  & Dollar trading volume  & Natural log of trading volume times price per share from month t-2 \tabularnewline
		dy & Dividend to price  & Total dividends (dvt) divided by market capitalization at fiscal year-end \tabularnewline
		egr  & Earnings announcement return  & Annual percent change in book value of equity (ceq) \tabularnewline
		ep & Earnings to price  & Annual income before extraordinary items (ib) divided by end of fiscal
		year market cap \tabularnewline
		gma  & Gross profitability  & Revenues (revt) minus cost of goods sold (cogs) divided by lagged
		total assets (at) \tabularnewline
		idiovol & Idiosyncratic return volatility  & Standard deviation of residuals of weekly returns on weekly equal
		weighted market returns for 3 years prior to month end \tabularnewline
		ill & Illiquidity (Amihud) & Average of daily (absolute return / dollar volume).\tabularnewline
		invest & Capital expenditures and inventory & Annual change in gross property, plant, and equipment (ppegt) + annual
		change in inventories (invt) all scaled by lagged total assets (at)\tabularnewline
		lev & Leverage & Annual change in gross property, plant, and equipment (ppegt) + annual
		change in inventories (invt) all scaled by lagged total assets (at)\tabularnewline
		lgr & Growth in long-term debt & Annual percent change in total liabilities (lt)\tabularnewline
		mom1m & 1-month momentum & 1-month cumulative return\tabularnewline
		
		mom6m & 6-month momentum & 5-month cumulative returns ending one month before month end\tabularnewline
		mve & Size & Natural log of market capitalization at end of month t-1\tabularnewline
		operprof & Operating profitability & Revenue minus cost of goods sold - SG\&A expense - interest expense
		divided by lagged common shareholders' equity\tabularnewline
		range & Range of stock price & Monthly average of daily price range: (high-low)/((high+low)/2) (alternative
		measure of volatility)\tabularnewline

		retvol & Return volatility & Standard deviation of daily returns from month t-1\tabularnewline
    
		roaq & Return on assets & Income before extraordinary items (ibq) divided by one quarter lagged
		total assets (atq)\tabularnewline
		roeq & Return on equity & Earnings before extraordinary items divided by lagged common shareholders'
		equity\tabularnewline
		sue & Unexpected quarterly earnings & Unexpected quarterly earnings divided by fiscal-quarter-end market
		cap. Unexpected earnings is I/B/E/S actual earnings minus median forecasted
		earnings if available, else it is the seasonally differenced quarterly
		earnings before extraordinary items from Compustat quarterly file\tabularnewline
		turn & Share turnover & Average monthly trading volume for most recent 3 months scaled by
		number of shares outstanding in current month\tabularnewline
		\hline
	\end{longtable}
\end{center}	
\vspace{-1cm}
Note: Estimated under different values of turning parameter $ \nu_{2}$, when $ \nu_{1}=10^{-5}$ is fixed. The results are  
reported for quantiles 10\%, 50\% and 90\%.

\end{document}